\documentclass{article} 

\newif\ifignore 
\ignorefalse
\newcommand{\auxproof}[1]{
\ifignore\mbox{}\newline
\textbf{PROOF:} \dotfill\newline
{\it #1}\mbox{}\newline
\textbf{ENDPROOF}\dotfill
\fi}

\usepackage{amsmath}
\usepackage{amssymb}
\usepackage{amstext}
\usepackage{xspace}
\usepackage{url}
\usepackage{xy}
\xyoption{v2}
\xyoption{curve}
\xyoption{all}
\xyoption{2cell}
\UseTwocells
\usepackage{stmaryrd}

\newdimen\proofrulebreadth \proofrulebreadth=.05em
\newdimen\proofdotseparation \proofdotseparation=1.25ex
\newdimen\proofrulebaseline \proofrulebaseline=2ex
\newcount\proofdotnumber \proofdotnumber=3
\let\then\relax
\def\hfi{\hskip0pt plus.0001fil}
\mathchardef\squigto="3A3B
\newif\ifinsideprooftree\insideprooftreefalse
\newif\ifonleftofproofrule\onleftofproofrulefalse
\newif\ifproofdots\proofdotsfalse
\newif\ifdoubleproof\doubleprooffalse
\let\wereinproofbit\relax
\newdimen\shortenproofleft
\newdimen\shortenproofright
\newdimen\proofbelowshift
\newbox\proofabove
\newbox\proofbelow
\newbox\proofrulename
\def\shiftproofbelow{\let\next\relax\afterassignment\setshiftproofbelow\dimen0 }
\def\shiftproofbelowneg{\def\next{\multiply\dimen0 by-1 }%
\afterassignment\setshiftproofbelow\dimen0 }
\def\setshiftproofbelow{\next\proofbelowshift=\dimen0 }
\def\setproofrulebreadth{\proofrulebreadth}

\def\prooftree{
\ifnum  \lastpenalty=1
\then   \unpenalty
\else   \onleftofproofrulefalse
\fi
\ifonleftofproofrule
\else   \ifinsideprooftree
        \then   \hskip.5em plus1fil
        \fi
\fi
\bgroup
\setbox\proofbelow=\hbox{}\setbox\proofrulename=\hbox{}%
\let\justifies\proofover\let\leadsto\proofoverdots\let\Justifies\proofoverdbl
\let\using\proofusing\let\[\prooftree
\ifinsideprooftree\let\]\endprooftree\fi
\proofdotsfalse\doubleprooffalse
\let\thickness\setproofrulebreadth
\let\shiftright\shiftproofbelow \let\shift\shiftproofbelow
\let\shiftleft\shiftproofbelowneg
\let\ifwasinsideprooftree\ifinsideprooftree
\insideprooftreetrue
\setbox\proofabove=\hbox\bgroup$\displaystyle 
\let\wereinproofbit\prooftree
\shortenproofleft=0pt \shortenproofright=0pt \proofbelowshift=0pt
\onleftofproofruletrue\penalty1
}

\def\eproofbit{
\ifx    \wereinproofbit\prooftree
\then   \ifcase \lastpenalty
        \then   \shortenproofright=0pt  
        \or     \unpenalty\hfil         
        \or     \unpenalty\unskip       
        \else   \shortenproofright=0pt  
        \fi
\fi
\global\dimen0=\shortenproofleft
\global\dimen1=\shortenproofright
\global\dimen2=\proofrulebreadth
\global\dimen3=\proofbelowshift
\global\dimen4=\proofdotseparation
\global\count255=\proofdotnumber
$\egroup  
\shortenproofleft=\dimen0
\shortenproofright=\dimen1
\proofrulebreadth=\dimen2
\proofbelowshift=\dimen3
\proofdotseparation=\dimen4
\proofdotnumber=\count255
}

\def\proofover{
\eproofbit 
\setbox\proofbelow=\hbox\bgroup 
\let\wereinproofbit\proofover
$\displaystyle
}%
\def\proofoverdbl{
\eproofbit 
\doubleprooftrue
\setbox\proofbelow=\hbox\bgroup 
\let\wereinproofbit\proofoverdbl
$\displaystyle
}%
\def\proofoverdots{
\eproofbit 
\proofdotstrue
\setbox\proofbelow=\hbox\bgroup 
\let\wereinproofbit\proofoverdots
$\displaystyle
}%
\def\proofusing{
\eproofbit 
\setbox\proofrulename=\hbox\bgroup 
\let\wereinproofbit\proofusing
\kern0.3em$
}

\def\endprooftree{
\eproofbit 
  \dimen5 =0pt
\dimen0=\wd\proofabove \advance\dimen0-\shortenproofleft
\advance\dimen0-\shortenproofright
\dimen1=.5\dimen0 \advance\dimen1-.5\wd\proofbelow
\dimen4=\dimen1
\advance\dimen1\proofbelowshift \advance\dimen4-\proofbelowshift
\ifdim  \dimen1<0pt
\then   \advance\shortenproofleft\dimen1
        \advance\dimen0-\dimen1
        \dimen1=0pt
        \ifdim  \shortenproofleft<0pt
        \then   \setbox\proofabove=\hbox{%
                        \kern-\shortenproofleft\unhbox\proofabove}%
                \shortenproofleft=0pt
        \fi
\fi
\ifdim  \dimen4<0pt
\then   \advance\shortenproofright\dimen4
        \advance\dimen0-\dimen4
        \dimen4=0pt
\fi
\ifdim  \shortenproofright<\wd\proofrulename
\then   \shortenproofright=\wd\proofrulename
\fi
\dimen2=\shortenproofleft \advance\dimen2 by\dimen1
\dimen3=\shortenproofright\advance\dimen3 by\dimen4
\ifproofdots
\then
        \dimen6=\shortenproofleft \advance\dimen6 .5\dimen0
        \setbox1=\vbox to\proofdotseparation{\vss\hbox{$\cdot$}\vss}%
        \setbox0=\hbox{%
                \advance\dimen6-.5\wd1
                \kern\dimen6
                $\vcenter to\proofdotnumber\proofdotseparation
                        {\leaders\box1\vfill}$%
                \unhbox\proofrulename}%
\else   \dimen6=\fontdimen22\the\textfont2 
        \dimen7=\dimen6
        \advance\dimen6by.5\proofrulebreadth
        \advance\dimen7by-.5\proofrulebreadth
        \setbox0=\hbox{%
                \kern\shortenproofleft
                \ifdoubleproof
                \then   \hbox to\dimen0{%
                        $\mathsurround0pt\mathord=\mkern-6mu%
                        \cleaders\hbox{$\mkern-2mu=\mkern-2mu$}\hfill
                        \mkern-6mu\mathord=$}%
                \else   \vrule height\dimen6 depth-\dimen7 width\dimen0
                \fi
                \unhbox\proofrulename}%
        \ht0=\dimen6 \dp0=-\dimen7
\fi
\let\doll\relax
\ifwasinsideprooftree
\then   \let\VBOX\vbox
\else   \ifmmode\else$\let\doll=$\fi
        \let\VBOX\vcenter
\fi
\VBOX   {\baselineskip\proofrulebaseline \lineskip.2ex
        \expandafter\lineskiplimit\ifproofdots0ex\else-0.6ex\fi
        \hbox   spread\dimen5   {\hfi\unhbox\proofabove\hfi}%
        \hbox{\box0}%
        \hbox   {\kern\dimen2 \box\proofbelow}}\doll%
\global\dimen2=\dimen2
\global\dimen3=\dimen3
\egroup 
\ifonleftofproofrule
\then   \shortenproofleft=\dimen2
\fi
\shortenproofright=\dimen3
\onleftofproofrulefalse
\ifinsideprooftree
\then   \hskip.5em plus 1fil \penalty2
\fi
}

\makeatother
 \newdir{ >}{{}*!/-7.5pt/@{>}}
 \newdir{|>}{!/4.5pt/@{|}*:(1,-.2)@^{>}*:(1,+.2)@_{>}}
 \newdir{ |>}{{}*!/-3pt/@{|}*!/-7.5pt/:(1,-.2)@^{>}*!/-7.5pt/:(1,+.2)@_{>}}
\newcommand{\xyline}[2][]{\ensuremath{\smash{\xymatrix@1#1{#2}}}}
\newcommand{\xyinline}[2][]{\ensuremath{\smash{\xymatrix@1#1{#2}}}^{\rule[8.5pt]{0pt}{0pt}}}
\newcommand{\filter}{\raisebox{5.5pt}{$\xymatrix@=6pt@H=0pt@M=0pt@W=4pt{\\ \ar@{>->}[u]}$}}
\newcommand{\ideal}{\raisebox{1pt}{$\xymatrix@=5pt@H=0pt@M=0pt@W=4pt{\ar@{>->}[d] \\ \mbox{}}$}}
\makeatletter

\newtheorem{theorem}{Theorem}
\newtheorem{lemma}[theorem]{Lemma}
\newtheorem{proposition}[theorem]{Proposition}

\newtheorem{definition}[theorem]{Definition}
\newtheorem{example}[theorem]{Example}
\newtheorem{remark}[theorem]{Remark}

\newenvironment{proof}[1][Proof]%
   { \begin{trivlist}%
     \item[\hskip \labelsep {\bfseries #1}]%
   }%
   { \end{trivlist}%
   }
\newcommand{\QEDbox}{\square}
\newcommand{\QED}{\hspace*{\fill}$\QEDbox$}

\newcommand{\after}{\mathrel{\circ}}

\newcommand{\cat}[1]{\ensuremath{\mathbf{#1}}}
\newcommand{\Cat}[1]{\ensuremath{\mathbf{#1}}}

\newcommand{\op}{\ensuremath{^{\mathrm{op}}}}
\newcommand{\idmap}[1][]{\ensuremath{\mathrm{id}_{#1}}}

\newcommand{\support}{\ensuremath{\mathrm{supp}}}

\newcommand{\Vect}[1][]{\ensuremath{\Cat{Vect}_{#1}}}
\newcommand{\Mod}[1][]{\ensuremath{\Cat{Mod}_{#1}}}

\newcommand{\SA}{\ensuremath{\mathcal{S}{\kern-.95ex}\mathcal{A}}}
\newcommand{\BL}{\ensuremath{\mathcal{B}}}
\newcommand{\Pos}{\ensuremath{\mathcal{P}{\kern-.6ex}o{\kern-.35ex}s}}
\newcommand{\Ef}{\ensuremath{\mathcal{E}{\kern-.5ex}f}}
\newcommand{\Proj}{\ensuremath{\mathcal{P}{\kern-.45ex}r}}
\newcommand{\DM}{\ensuremath{\mathcal{D}{\kern-.85ex}\mathcal{M}}}

\newcommand{\Hom}{\textsl{Hom}}

\newcommand{\inprod}[2]{\ensuremath{\langle #1\,|\,#2 \rangle}}

\newcommand{\Alg}{\textsl{Alg}\xspace}

\newcommand{\Mon}{\textsl{Mon}\xspace}

\newcommand{\Act}{\textsl{Act}\xspace}

\newcommand{\Mlt}{\ensuremath{\mathcal{M}}}
\newcommand{\Dstr}{\ensuremath{\mathcal{D}}}
\newcommand{\UF}{\ensuremath{\mathcal{U}{\kern-.75ex}\mathcal{F}}}
\newcommand{\CS}{\ensuremath{\mathcal{C}{\kern-.75ex}\mathcal{S}}}

\newcommand{\NNO}{\ensuremath{\mathbb{N}}}
\newcommand{\C}{\mathbb{C}}
\newcommand{\R}{\mathbb{R}}
\newcommand{\Rnn}{\mathbb{R}_{\geq0}} 
\newcommand{\unitR}{\ensuremath{[0,1]}}

\newcommand{\closed}{\ensuremath{\mathcal{C}{\kern-.45ex}\ell}}

\newcommand{\orthogonal}{\mathrel{\bot}}

\newcommand{\scalar}{\mathrel{\bullet}}

\newcommand{\powerset}{\mathcal{P}}
\newcommand{\Pow}{\powerset}

\newcommand{\powersetfin}{\mathcal{P}_{\mathit{fin}}}

\newcommand{\tr}{\ensuremath{\textrm{tr}}\xspace}
\newcommand{\ket}[1]{|\,#1\,\rangle}
\newcommand{\bra}[1]{\langle\,#1\,|}
\newcommand{\hs}[1][]{{\frak{hs}_{#1}}}
\newcommand{\weakprec}{\ensuremath{\textrm{wp}}\xspace}
\newcommand{\leftScottint}{[{\kern-.3ex}[}
\newcommand{\rightScottint}{]{\kern-.3ex}]}

\newcommand{\Hilb}{\Cat{Hilb}\xspace}
\newcommand{\FdHilb}{\Cat{FdHilb}\xspace}
\newcommand{\FdHilbUn}{\Cat{FdHilb}_{\mathrm{Un}}\xspace}
\newcommand{\Conv}{\Cat{Conv}\xspace}
\newcommand{\CCH}{\Cat{CCH}\xspace}
\newcommand{\CCHobs}{\CCH_{\mathrm{obs}}}
\newcommand{\EMod}{\Cat{EMod}\xspace}
\newcommand{\BEMod}{\Cat{BEMod}\xspace}
\newcommand{\Sets}{\Cat{Sets}\xspace}

\newcommand{\EA}{\Cat{EA}\xspace}

\newcommand{\sotimes}{\mathrel{\raisebox{.05pc}{$\scriptstyle \otimes$}}}

\newcommand{\set}[2]{\{#1\;|\;#2\}}
\newcommand{\setin}[3]{\{#1\in#2\;|\;#3\}}

\newcommand{\all}[2]{\forall{#1}.\,#2}
\newcommand{\allin}[3]{\forall{#1\in#2}.\,#3}
\newcommand{\ex}[2]{\exists{#1}.\,#2}
\newcommand{\exin}[3]{\exists{#1\in#2}.\,#3}
\newcommand{\lamin}[3]{\lambda{#1\in#2}.\,#3}
\newcommand{\lam}[2]{\lambda{#1}.\,#2}

\newcommand{\congrightarrow}{\mathrel{\stackrel{
           \raisebox{.5ex}{$\scriptstyle\cong\,$}}{
           \raisebox{0ex}[0ex][0ex]{$\rightarrow$}}}}
\newcommand{\conglongrightarrow}{\mathrel{\stackrel{
           \raisebox{.5ex}{$\scriptstyle\cong\,$}}{
           \raisebox{0ex}[0ex][0ex]{$\longrightarrow$}}}}

\newcommand{\BCM}{\ensuremath{\mathbf{BCM}}\xspace}
\newcommand{\Bmodu}{\ensuremath{\mathbf{BModu}}\xspace}
\newcommand{\poVectu}{\ensuremath{\mathbf{poVectu}}\xspace}
\newcommand{\partot}{\ensuremath{\mathcal{T}\hspace{-3pt}{\scriptstyle{o}}}}
\newcommand{\totpar}{\ensuremath{\mathcal{P}\hspace{-3pt}{\scriptstyle{a}}}}

\newcommand{\ojoin}{\ovee}
\newcommand{\eps}{\varepsilon}
\newcommand{\reals}{\ensuremath{\mathbb{R}}}

\renewcommand{\arraycolsep}{3pt}

\title{Relating Operator Spaces via Adjunctions}
\author{Bart Jacobs\ and\ Jorik Mandemaker \\
{\small Institute for Computing and Information Sciences (iCIS),} \\[-.5em]
{\small Radboud University Nijmegen, The Netherlands.} \\[-.5em]
{\small Contact: \texttt{\{bart,mandemak\}@cs.ru.nl}}}

\date{\small \today}

\begin{document}
\maketitle

\begin{abstract}
This chapter uses categorical techniques to describe relations between
various sets of operators on a Hilbert space, such as self-adjoint,
positive, density, effect and projection operators. These relations,
including various Hilbert-Schmidt isomorphisms of the form $\tr(A-)$,
are expressed in terms of dual adjunctions, and maps between them. Of
particular interest is the connection with quantum structures, via a
dual adjunction between convex sets and effect modules. The approach
systematically uses categories of modules, via their description as
Eilenberg-Moore algebras of a monad.
\end{abstract}

\section{Introduction}\label{IntroSec}

There is a recent exciting line of work connecting research in the
semantics of programming languages and logic, and research in the
foundations of quantum physics, including quantum computation and
logic, see~\cite{Coecke11a} for an overview. This paper fits in that
line of work. It concentrates on operators (on Hilbert spaces) and
organises and relates these operators according to their algebraic
structure. This is to a large extent not more than a systematic
presentation of known results and connections in the (modern) language
of category theory. However, the approach leads to clarifying results,
like Theorem~\ref{ConvEModAdjMapThm} that relates density operators
and effects via a dual adjunction between convex sets and effect
modules (extending earlier work~\cite{Jacobs10e}). It is in line with
many other dual adjunctions and dualities that are relevant in
programming logics~\cite{Johnstone82,Abramsky91,JacobsS10}. Indeed,
via this dual adjunction we can put the work~\cite{dHondtP06a} on
quantum weakest preconditions in perspective (see especially
Remark~\ref{WPRem}).

The article begins by describing the familiar sets of operators
(bounded, self-adjoint, positive) on a (finite-dimensional) Hilbert
space in terms of functors to categories of modules. The dual
adjunctions involved are made explicit, basically via dual operation
$V \mapsto V^{*}$, see Section~\ref{FunctOperSec}.  Since the
algebraic structure of these sets of operators is described in terms
of modules over various semirings, namely over complex numbers $\C$
(for bounded operators), over real numbers $\R$ (for self-adjoint
operators), and over non-negative real numbers $\Rnn$ (for positive
operators), it is useful to have a uniform description of such
modules. It is provided in Section~\ref{ModuleSec}, via the notion of
algebra of a monad (namely the multiset monad). This abstract
description provides (co)limits and the monoidal closed structure of
such algebras (from~\cite{Kock71a}) for free. We then use that convex
sets can also be described as such algebras of a monad (namely the
distribution monad), and elaborate the connection with effect modules
(also known as convex effect algebras, see~\cite{PulmannovaG98}). In
this setting we discuss various `Gleason-style' correspondences,
between projections, effects and density matrices. We borrow the
probabilistic Gelfand duality between (Banach) effect modules and
(compact) convex sets from~\cite{JacobsM12b} for the final steps in
our analysis. This duality formalises the difference between the
approaches of Heisenberg (focusing on observables/effects) and
Schr{\"o}dinger (focusing on states), see
\textit{e.g.}~\cite{HeinosaariZ12}. It allows us to reconstruct all
sets of operators on a Hilbert space from its projections, see
Figure~\ref{FreeConstructsFig} for an overview. The main contribution
of the paper thus lies in a systematic description.

We should emphasise that the investigations in this paper concentrate
on finite-dimensional Hilbert and vector spaces.

\subsection{Operator overview}\label{OperatorSubsec}

For a (finite-dimensional) Hilbert space $H$ we shall study the
following sets of operators $H\rightarrow H$.
\begin{equation}
\label{OperatorSetsDiag}
\vcenter{\xymatrix@R-2pc{
& & & & \Proj(H)\ar@{_{(}->}[dl] \\
\BL(H) &
\;\SA(H)\ar@{_{(}->}[l] &
\;\Pos(H)\ar@{_{(}->}[l] & 
\;\Ef(H)\ar@{_{(}->}[l] \\
& & & & \DM(H)\ar@{_{(}->}[ul] \\
}}
\end{equation}

\noindent where:
\begin{center}
\begin{tabular}{c||c|c}
\textbf{Notation} & \textbf{Description} & \textbf{Structure}
   \vrule height5mm depth3mm width0mm \\
\hline\hline 
$\BL(H)$ & bounded/continuous linear & vector space over $\C$
   \vrule height5mm depth3mm width0mm \\
\hline 
$\SA(H)$ & self-adjoint: $A^{\dag} = A$ & vector space over $\R$
   \vrule height5mm depth3mm width0mm \\
\hline
$\Pos(H)$ & positive: $A\geq 0$ & module over $\Rnn$
   \vrule height5mm depth3mm width0mm \\
\hline
$\Ef(H)$ & effect: $0 \leq A \leq I$ & effect module over $[0,1]$
   \vrule height5mm depth3mm width0mm \\
\hline
$\Proj(H)$ & projection: $A^{2} = A = A^{\dag}$ & orthomodular lattice
   \vrule height5mm depth3mm width0mm \\
\hline
$\DM(H)$ & density: $A\geq 0$ and $\tr(A) = 1$ & convex set
\end{tabular}
\end{center}

\noindent The emphasis lies on the `structure' column. It describes
the algebraic structure of the sets of operators that will be relevant
here. It is not meant to capture all the structure that is
present. For instance, the set $\BL(H)$ of endomaps is not only a
vector space over the complex numbers, but actually a $C^*$-algebra.

As is well-known, operators on Hilbert spaces behave in a certain
sense as numbers. For instance, by taking $H$ to be the trivial space
$\C$ of complex numbers, the diagram~(\ref{OperatorSetsDiag})
becomes:
$$\xymatrix@R-2pc{
& & & & \{0,1\}\ar@{_{(}->}[dl] \\
\C &
\;\R\ar@{_{(}->}[l] &
\;\Rnn\ar@{_{(}->}[l] & 
\;[0,1]\ar@{_{(}->}[l] \\
& & & & \{1\}\ar@{_{(}->}[ul] \\
}$$

\section{Operators and duality}\label{FunctOperSec}

This section concentrates on the first three sets of operators
in~(\ref{OperatorSetsDiag}), namely on $\BL(H) \hookleftarrow \SA(H)
\hookleftarrow \Pos(H)$. It will focus on isomorphisms $V \cong
V^{*}$, for $V = \BL(H), \SA(H), \Pos(H)$. These isomorphisms turn out
to be natural in $H$, in categories of modules (or vector
spaces). This serves as motivation for further investigation of the
structures involved, in subsequent sections. Only later will we study
the density and effect operators $\DM(H)$ and $\Ef(H)$, capturing
`states and statements' in quantum logic. The material in this section
thus serves as preparation. It is not new, except possibly for the
presentation in terms of maps of adjunctions.

We start by recalling that the category $\Vect[\C]$ of vector spaces
over the complex numbers $\C$ carries an involution given by
conjugation: for a vector space $V$ we write $\overline{V}$ for the
conjugate space, with the same vectors as $V$, but with scalar
multiplication given by $z\scalar_{\overline{V}} x =
\overline{z}\scalar_{V} x$, where the complex number $z\in\C$ has
conjugate $\overline{z}\in\C$. This yields an `involution' endofunctor
$\overline{(-)} \colon \Vect[\C] \rightarrow \Vect[\C]$ which is the
identity on morphisms. A linear map $f\colon \overline{V} \rightarrow
W$ is sometimes called \emph{conjugate} linear, because it satisfies
$f(z\scalar v) = \overline{z}\scalar f(v)$. Complex conjugation $z
\mapsto \overline{z}$ is an example of a conjugate linear
(isomorphism) $\overline{\C} \conglongrightarrow \C$ in
$\Vect[\C]$. We refer to~\cite{BeggsM09,Egger11,Jacobs11e} for more
information on involutions in a categorical setting.

We shall write $V\multimap W$ for the `exponent' vector space of
linear maps $V\rightarrow W$ between vector spaces $V$ and $W$. There
is the standard correspondence between linear functions $U\rightarrow
(V\multimap W)$ and $U\otimes V \rightarrow W$. 

One uses this exponent $\multimap$ to form the dual space $V^{*} =
\overline{V} \multimap \C$. If $V$ is finite-dimensional, say with a
basis $e_{1}, \ldots, e_{n}$, written in `ket' notation as $\ket{j} =
e_{j}$, there is the familiar isomorphism of $V$ with its dual space
$V^{*} = \overline{V}\multimap\C$ given as:
\begin{equation}
\label{DualSpaceEqn}
\vcenter{\xymatrix@R-2pc{
V\ar[r]^-{\cong} & V^{*} 
   \rlap{$\;\;\cong\; \overline{V}\multimap\C$} \\
\Big(\sum_{j}z_{j}\ket{j}\Big)\ar@{|->}[r] &
   \Big(\sum_{j}z_{j}\bra{j}\Big),
}}
\end{equation}

\noindent where the `bra' $\bra{j}\colon \overline{V} \rightarrow \C$
sends a vector $w = (\sum_{k}w_{k}\ket{k})$ to its $j$-th coordinate
$\inprod{j}{w} = \sum_{k}\overline{w_k}\inprod{j}{k} =
\overline{w_j}$. Clearly, this yields an isomorphism $V
\congrightarrow V^{*}$, because these functions $\bra{j}$ form a
`dual' basis for $V^*$. This isomorphism~(\ref{DualSpaceEqn}) is a
famous example of a non-natural mapping, depending on a choice of
basis. It will play a crucial role below, where $V$ is a vector space
of operators on a Hilbert space.


\auxproof{ 
We check that $\bra{j}$ is an antilinear map:
$$\begin{array}{rcl}
\bra{j}(z\cdot_{\overline{V}} w)
& = &
\bra{j}(\overline{z}\cdot_{V} w) \\
& = &
\bra{j}(\overline{z}w_{1}, \ldots, \overline{z}w_{n}) \\
& = &
\overline{\overline{z}w_{j}} \\
& = &
z\overline{w_{j}} \\
& = &
z\bra{j}(w).
\end{array}$$

In a Hilbert space this $\bra{j}$ may be understood as a function
sending an arbitrary vector $\ket{x}$ to $\bra{j}\ket{x} =
\inprod{j}{x} = \overline{\inprod{x}{j}} = x_{j}$ if $x =
\sum_{j}x_{j}\ket{j}$. The antilinearity works here in the second
coordinate of the inner product.
}

The mapping $V \mapsto V^{*} = (\overline{V}\multimap\C)$ yields a
functor $\Vect[\C] \rightarrow \big(\Vect[\C]\big)\op$; for a map
$C\colon H \rightarrow K$ we have $C^{*} \colon K^{*} \rightarrow
H^{*}$ given by $f \mapsto f \after \overline{C}$. This functor
$(-)^{*}$ is adjoint to itself, in the sense that there is a bijective
correspondence (suggested by the double lines) as on the left below,
forming an adjunction as on the right.
$$\qquad\qquad\begin{prooftree}
\begin{prooftree}
V \longrightarrow (\overline{W}\multimap\C)
\Justifies
V\otimes \overline{W} \longrightarrow \C \rlap{$\;\cong\overline{\C}$}
\end{prooftree}
\Justifies
\begin{prooftree}
\llap{$\overline{V}\otimes W \cong\;$}
   \overline{V\otimes\overline{W}} \longrightarrow \C
\Justifies
W \longrightarrow (\overline{V}\multimap\C)
\end{prooftree}
\end{prooftree}
\qquad\qquad
\vcenter{\xymatrix{
\Vect[\C]\ar@/^2ex/[rr]^-{(-)^{*} = \overline{(-)}\multimap \C} & \bot & 
   \big(\Vect[\C]\big)\rlap{$\op$}
      \ar@/^2ex/[ll]^-{(-)^{*} = \overline{(-)}\multimap \C}
}}$$

In the next step, let $\FdHilb$ be the category of finite-dimensional
Hilbert spaces with bounded linear maps between them. One can drop the
boundedness requirement, because a linear map between
finite-dimensional spaces is automatically bounded
(\textit{i.e.}~continuous). As is usual, we write $\BL(H)$ for the
homset of endomaps $H\rightarrow H$ in $\FdHilb$.  This set $\BL(H)$
of ``operators on $H$'' is a vector space over $\C$, of
dimension $n^2$ with the outer products $\ket{j}\bra{k}$, for $j,k\leq
n$, as basis---assuming a basis $\ket{1}, \ldots, \ket{n}$ for
$H$. Such outer product projections $\ket{j}\bra{k}$ may be understood
as the matrix with only 0s except for a single 1 in the $j$-th row of
the $k$-th column. In general, an operator $A\colon H\rightarrow H$
can be written as matrix $A = \sum_{j,k}A_{jk}\ket{j}\bra{k}$, where
the matrix entries $A_{jk}$ may be described as $\bra{j}A\ket{k}$.


The mapping $H\mapsto\BL(H)$ is functorial, and will be used here as
functor $\BL\colon\FdHilb\rightarrow\Vect[\C]$. On a map
$C\colon H\rightarrow K$ it yields a linear function $\BL(H)
\rightarrow \BL(K)$, written as $\BL(C)$, and given by:
\begin{equation}
\label{BLVectEqn}
\begin{array}{rcl}
\BL(C)\Big(H\stackrel{A}{\longrightarrow} H\Big)
& = &
\Big(K\stackrel{C^{\dag}}{\longrightarrow} H \stackrel{A}{\longrightarrow}
   H \stackrel{C}{\longrightarrow} K\Big).
\end{array}
\end{equation}

\auxproof{
We check functoriality of $\BL$
$$\begin{array}{rcl}
\BL(C \after D)(A)
& = &
(C \after D)^{\dag} \after A \after (C \after D) \\
& = &
D^{\dag} \after C^{\dag} \after A \after C \after D \\
& = &
D^{\dag} \after \BL(C)(A) \after D \\
& = &
\BL(D)(\BL(C)(A)) \\
& = &
\big(\BL(D) \after \BL(C)\big)(A) \\
\end{array}$$

\noindent Also, $\BL(C) \colon \BL(H) \rightarrow \BL(K)$ preserves
involution:
$$\begin{array}{rcl}
\BL(C)(A^{\dag})
& = &
CA^{\dag}C^{\dag} \\
& = &
\big(CAC^{\dag}\big)^{\dag} \\
& = &
\BL(C)(A)^{\dag}.
\end{array}$$
}

\noindent The operator $C^{\dag} = \overline{C}^{T}$ is the conjugate
transpose of $C$, satisfying $\inprod{Cv}{w} =
\inprod{v}{C^{\dag}w}$. It makes $\FdHilb$ into a dagger category, see
\textit{e.g.}~\cite{AbramskyC09}. This dagger forms an involution on
the vector space $\BL(H)$. Also, it is adjoint to itself, as in:
$$\begin{prooftree}
V \;\xrightarrow{\;\,f\,\;}\; W
\Justifies
W \;\xrightarrow[\;f^{\dag}\;]\; V
\end{prooftree}
\qquad\qquad\qquad
\vcenter{\xymatrix{
\FdHilb\ar@/^2ex/[rr]^-{(-)^{\dag}} & \bot & 
   \FdHilb\rlap{$\op$}\ar@/^2ex/[ll]^-{(-)^{\dag}}
}}$$

In the next result we apply the duality isomorphism $V \cong V^{*}$
in~(\ref{DualSpaceEqn}) for $V=\BL(H)$. As we shall see, it involves
the trace operation $\tr\colon \BL(H) \rightarrow \C$ of which
we first recall some basic facts. For $A\in\BL(H)$ the trace $\tr(A)$
can be defined as the sum $\sum_{j}A_{jj}$ of the diagonal matrix
values. This definition is independent of the choice of
matrix/basis. This trace $\tr$ satisfies the following basic
properties.
$$\begin{array}{rcll}
\tr(A+B)
& = &
\tr(A) + \tr(B) \\
\tr(zA)
& = &
z\,\tr(A) & \mbox{where $z\in\C$} \\
\tr(AB)
& = &
\tr(BA) & \mbox{the so-called cyclic property} \\
\tr(A^{T})
& = &
\tr(A) & \mbox{where $(-)^{T}$ is the transpose operation} \\
\tr(A^{\dag})
& = &
\overline{\tr(A)} & \mbox{which results from previous points} \\
\tr(A)
& \geq &
0 & \mbox{when $A$ is positive: $A \geq 0$, 
    \textit{i.e.}~$\inprod{v}{Av} \geq 0$.}
\end{array}$$

\auxproof{
We check the dagger property, because it is maybe less standard.
$$\begin{array}{rcl}
\tr(A^{\dag})
& = &
\sum_{j} (A^{\dag})_{jj} \\
& = &
\sum_{j} \overline{A_{jj}} \\
& = &
\overline{\sum_{j}A_{jj}} \\
& = &
\overline{\tr(A)}.
\end{array}$$
}

\begin{proposition}
\label{BLDualProp}
For a finite-dimensional Hilbert space $H$ the duality
isomorphism~(\ref{DualSpaceEqn}) applied to the vector space $\BL(H)$
of endomaps boils down to a trace calculation, namely:
$$\xymatrix@R-2pc{
\BL(H)\ar[r]^-{\hs[\BL]}_-{\cong} & 
   \BL(H)^{*} = \overline{\BL(H)} \multimap \C
}
\quad\mbox{is}\quad
\begin{array}{rcl}
\hs[\BL](A)
& = &
\lam{B}{\tr(AB^{\dag})},
\end{array}$$

\noindent where the $\lambda$-notation is borrowed from the
$\lambda$-calculus, and used for function abstraction:
$\lam{B}{\cdots}$ describes the function $B \mapsto \cdots$.

This map $\BL(H) \conglongrightarrow \BL(H)^{*}$ is independent of the
choice of basis. More categorically, it yields a natural isomorphism
involving adjoint $(-)^{\dag}$ and dual $(-^{*} =
\overline{(-)}\multimap\C$ in:
$$\xymatrix{
\BL \after (-)^{\dag} \ar@{=>}[r]^-{\hs[\BL]}_-{\cong} & (-)^{*} \after \BL,
}$$

\noindent Pictorially, this $\hs[\BL]$ is a natural transformation
$\smash{\xymatrix@R-1pc@C+1pc{\FdHilb\rtwocell{\cong} & \Vect[\C]\op}}$
between the two functors $\FdHilb \rightrightarrows \Vect[\C]\op$
given in:
$$\xymatrix@R-2pc{
& & \FdHilb\rlap{$\op$}\ar[drr]^-{\BL} & & \\
\FdHilb\ar[urr]^-{(-)^{\dag}}\ar[drr]_{\BL} & & & & \Vect[\C]\rlap{$\op$} \\
& & \Vect[\C]\ar[urr]_{(-)^{*}}
}$$

\noindent Moreover, this $\hs[\BL]$ is part of a \emph{map of
  adjunctions} (see~\cite[IV,7]{MacLane71}) in the following
situation.
$$\xymatrix{
\FdHilb\ar@/^1.2ex/[rr]^-{(-)^{\dag}}\ar[d]_{\BL} & \bot & 
   \FdHilb\rlap{$\op$}\ar@/^1.2ex/[ll]^-{(-)^{\dag}}\ar[d]^{\BL} \\
\Vect[\C]\ar@/^1.2ex/[rr]^-{(-)^{*}} & \bot & 
   \Vect[\C]\rlap{$\op$}\ar@/^1.2ex/[ll]^-{(-)^{*}}
}$$
\end{proposition}

The letters `h' and `s' in the map $\hs[\BL]$ stand for Hilbert and
Schmidt, since the inner product $(A,B) \mapsto \tr(AB^{\dag}) =
\hs[\BL](A)(B)$ is commonly named after them. The subscript $\BL$ is
added because we shall encounter analogues of this isomorphism for
other operators. We drop the subscript when confusion is unlikely.


\begin{proof}
If $\ket{1}, \ldots, \ket{n}$ is a basis for $H$, then the map
$\hs \colon \BL(H)\rightarrow \BL(H)^{*}$ becomes, according
to~(\ref{DualSpaceEqn}),
$$\textstyle A
\hspace*{\arraycolsep} = \hspace*{\arraycolsep}
\big(\sum_{j,k}A_{jk}\ket{j}\bra{k}\big)
\hspace*{\arraycolsep} \longmapsto 
\begin{array}[t]{rcl}
\lam{B}{\sum_{j,k}A_{jk}\overline{B_{jk}}} 
& = &
\lam{B}{\sum_{j,k}A_{jk}(B^{\dag})_{kj}} \\
& = &
\lam{B}{\sum_{j}\,(AB^{\dag})_{jj}} \\
& = &
\lam{B}{\tr(AB^{\dag})}.
\end{array}$$

\noindent Since the trace of a matrix is basis-independent, so is this
isomorphism $\hs$. Naturality amounts to commutation of the following
diagram, for each map $C\colon H\rightarrow K$ in $\FdHilb$.
$$\xymatrix@R-1pc{
\BL(H)\ar[rr]^-{{\hs}_{H}}_-{\cong} & & 
   \BL(H)^{*} \rlap{$\;=\overline{\BL(H)}\multimap\C$} \\
\BL(K)\ar[u]^{\BL(C^{\dag})}\ar[rr]^-{\hs_K}_-{\cong} & & 
   \BL(K)^{*} \rlap{$\;=\overline{\BL(K)}\multimap\C$}
   \ar[u]_{\BL(C)^{*}}
}$$

\noindent This diagram commutes because:
$$\begin{array}{rcl}
\big(\BL(C)^{*} \after \hs_{K}\big)(A)(B)
& = &
\big(\hs_{K}(A) \after \BL(C)\big)(B) \\
& = &
\hs_{K}(A)\big(\BL(C)(B)\big) \\
& = &
\hs_{K}(A)\big(CBC^{\dag}\big) \\
& = &
\tr(A(CBC^{\dag})^{\dag}) \\
& = &
\tr(ACB^{\dag}C^{\dag}) \\
& = &
\tr(C^{\dag}ACB^{\dag}) \qquad \mbox{by the cyclic property} \\
& = &
\hs_{H}(C^{\dag}AC)(B) \\
& = &
\hs_{H}\big(\BL(C^{\dag})(A)\big)(B) \\
& = &
\Big(\hs_{H} \after \BL(C^{\dag})\Big)(A)(B).
\end{array}$$

\noindent Finally, we use the basic fact that, because these
$\hs_{H}$'s are natural in $H$ and componentwise isomorphisms, the
inverses $\hs_{H}^{-1}$ are also natural in $H$, see
\textit{e.g.}~\cite[Lemma~7.11]{Awodey06}. The details of the map of
adjunctions in the above diagram are left to the interested
reader. \QED

\auxproof{
First, the adjunction $(-)^{\dag} \dashv (-)^{\dag}$ is given by
the correspondences:
$$\begin{prooftree}
{\xymatrix{ X\ar[r]^-{f} & Y \rlap{$\;= Y^{\dag}$} }}
\Justifies
{\xymatrix{ Y\ar[r]_-{f^{\dag}} & X \rlap{$\;= X^{\dag}$} }}
\end{prooftree}$$

\noindent The unit and counit of this adjunction are the identities.

The unit and counit of the adjunction between vector spaces both have
type $V \rightarrow (\overline{(\overline{V}\multimap
  \C)}\multimap \C)$, and are given by
$\lam{x}{\lam{A}{\overline{Ax}}}$. 

In general, a map of adjunctions:
$$\xymatrix@R-.5pc{
\cat{A}\ar@/^1.2ex/[rr]^-{F}\ar[d]_{K} & \bot & 
   \cat{B}\ar@/^1.2ex/[ll]^-{G}\ar[d]^{L} \\
\cat{A}'\ar@/^1.2ex/[rr]^-{F'} & \bot & 
   \cat{B}'\ar@/^1.2ex/[ll]^-{G'}
}$$

\noindent involves isomorphisms
$$\xymatrix{
F'K\ar[r]^-{\alpha}_-{\cong} & LF
& &
G'L\ar[r]^-{\beta}_-{\cong} & KG
}$$

\noindent making the following diagrams commute.
$$\xymatrix@R-1pc@C+1pc{
& G'F'K\ar[d]^{G'\alpha} 
& &
F'G'L\ar[dr]^{\varepsilon'L}\ar[d]_{F'\beta} & \\
K\ar[ur]^-{\eta'}\ar[dr]_{K\eta} & G'LF\ar[d]^{\beta F} 
& &
F'KG\ar[d]_{\alpha G} & L \\
& KGF
& &
LFG\ar[ur]_{L\varepsilon}
}$$

\noindent In the present situation:
$$\xymatrix{
\FdHilb\ar@/^1.2ex/[rr]^-{(-)^{\dag}}\ar[d]_{\BL} & \bot & 
   \FdHilb\rlap{$\op$}\ar@/^1.2ex/[ll]^-{(-)^{\dag}}\ar[d]^{\BL} \\
\Vect[\C]\ar@/^1.2ex/[rr]^-{\overline{(-)}\multimap \C} & \bot & 
   \Vect[\C]\rlap{$\op$}\ar@/^1.2ex/[ll]^-{\overline{(-)}\multimap \C}
}$$

\noindent this involves isomorphisms ($\alpha$ and $\beta$ above):
$$\xymatrix{
\big(\overline{\BL(H)}\multimap\C\big)\ar[r]^-{\hs_{H}^{-1}}_-{\cong} &
   \BL(H^{\dag})\quad\mbox{in }\Vect[\C]
\qquad
\big(\overline{\BL(H)}\multimap\C\big)\ar[r]^-{\hs_{H}^{-1}}_-{\cong} &
   \BL(H^{\dag})\quad\mbox{in }\Vect[\C]\op
}$$

$$\xymatrix@R-1pc@C+1pc{
& 
\overline{(\overline{\overline{\BL(H)}}\multimap\C)}\multimap\C
   \ar[d]^{\overline{\hs}\multimap\C =
      \lam{f}{\lam{B}{f(\hs(B))}}} \\
\overline{\BL(H)}\ar[ur]^-{\eta}\ar@{=}[dr] & 
   \overline{\overline{\BL(H)}}\multimap\C\ar[d]^{\hs^{-1}} \\
& \overline{\BL(H)}
}$$

\noindent This commutes since:
$$\begin{array}{rcl}
\big(\hs^{-1} \after (\overline{\hs}\multimap\C) 
   \after \eta\big)(A)
& = &
\big(\hs^{-1} \after (\overline{\hs}\multimap\C)\big)
   (\eta(A)) \\
& = &
\hs^{-1}\big(\lam{B}{\eta(A)(\hs(B))}\big) \\
& = &
\hs^{-1}\big(\lam{B}{\overline{\hs(B)(A)}}\big) \\
& = &
\hs^{-1}\big(\lam{B}{\overline{\tr(BA^{\dag})}}\big) \\
& = &
\hs^{-1}\big(\lam{B}{\tr((BA^{\dag})^{\dag})}\big) \\
& = &
\hs^{-1}\big(\lam{B}{\tr(AB^{\dag})}\big) \\
& = &
\hs^{-1}\big(\hs(A)\big) \\
& = &
A.
\end{array}$$

\noindent Since everything is dual, the same works for $\varepsilon$.
}
\end{proof}

The remarkable thing about this result is that whereas the maps $V
\conglongrightarrow V^{*}$ in~(\ref{DualSpaceEqn}) are not natural,
the instantiations $\hs\colon \BL(H) \conglongrightarrow \BL(H)^{*}$
are, because they involve a trace calculation that is
base-independent.  We briefly describe the inverse of
$\hs = \lam{A}{\tr(A-)} \colon \BL(H) \conglongrightarrow \BL(H)^{*} =
(\overline{\BL(H)}\multimap \C)$, via a choice of basis $\ket{1},
\ldots, \ket{n}$ for $H$. So suppose we have a linear map $f\colon
\overline{\BL(H)} \rightarrow \C$. Define an operator $\hs^{-1}(f)\in
\BL(H)$ with matrix entries:
\begin{equation}
\label{hsBLinvEqn}
\begin{array}{rcl}
\big(\hs^{-1}(f)\big)_{jk}
& = &
f(\ket{j}\bra{k}).
\end{array}
\end{equation}

\noindent Then we recover $f$ via the trace calculation:
$$\begin{array}{rcl}
\hs\Big(\hs^{-1}(f)\Big)(B)
& = &
\tr\big(\hs^{-1}(f)B^{\dag}\big) \\
& = &
\tr\big((\sum_{j,k}\hs^{-1}(f)_{jk}\ket{j}\bra{k})B^{\dag}\big) \\
& = &
\sum_{j,k}\hs^{-1}(f)_{jk}\,\tr\big(\ket{j}\bra{k}B^{\dag}\big) \\
& = &
\sum_{j,k}f(\ket{j}\bra{k})\,\tr\big(\bra{k}B^{\dag}\ket{j}\big) \\
& = &
\sum_{j,k}f(\ket{j}\bra{k})\,\tr\big((B^{\dag})_{kj}\big) \\
& = &
\sum_{j,k}f(\ket{j}\bra{k})\overline{B_{jk}} \\
& = &
\sum_{j,k}f(B_{jk}\ket{j}\bra{k}), \qquad
   \mbox{because $f$ is conjugate linear} \\
& = &
f\big(\sum_{j,k}B_{jk}\ket{j}\bra{k}\big) \\
& = &
f(B).
\end{array}$$

\noindent Again, this mapping $f\mapsto A_{f}$ is independent of a
choice of basis, because its inverse $A\mapsto \tr(A-)$ does not
depend on such a choice.

\auxproof{
We also check the equality in the other direction:
$$\begin{array}{rcl}
\hs^{-1}(\hs(A))_{jk}
& = &
\hs(A)(\ket{j}\bra{k}) \\
& = &
\tr(A(\ket{j}\bra{k})^{\dag}) \\
& = &
\tr(A\ket{k}\bra{j}) \\
& = &
\tr(\bra{j}A\ket{k}) \\
& = &
A_{jk}.
\end{array}$$
}

\subsection*{Self-adjoint operators}

We now restrict ourselves to self-adjoint operators $\SA(H)
\hookrightarrow \BL(H)$. We recall that an operator $A\colon
H\rightarrow H$ is called self-adjoint (or Hermitian) if $A^{\dag} =
A$. In terms of matrices this means that $A_{jk} =
\overline{A_{kj}}$. In particular, all entries $A_{jj}$ on the
diagonal are real numbers, and so is the trace (as sum of these
$A_{jj}$). The set of self-adjoint operators $\SA(H)$ forms a vector
space over $\R$. The mapping $H \mapsto \SA(H)$ can be extended to a
functor $\SA\colon \Hilb\rightarrow\Vect[\R]$, by:
$$\begin{array}{rcl}
\SA(C)\Big(H\stackrel{A}{\longrightarrow} H\Big)
& = &
\Big(K\stackrel{C^{\dag}}{\longrightarrow} H \stackrel{A}{\longrightarrow}
   H \stackrel{C}{\longrightarrow} K\Big).
\end{array}$$

\noindent like for $\BL$ in~(\ref{BLVectEqn}). This is well-defined
since if $A$ is self-adjoint then so is $\SA(C)(A)$, since:
$$\big(\SA(C)(A)\big)^{\dag}
=
\big(CAC^{\dag}\big)^{\dag}
=
C^{\dag\dag}A^{\dag}C^{\dag}
=
CAC^{\dag}
=
\SA(C)(A).$$

\noindent There are serveral ways to turn a linear operator into a
self-adjoint one. For instance, for each complex number $z\in\C$ and
$B\in\BL(H)$ we have self-adjoint operators:
\begin{equation}
\label{BLtoSAEqn}
zB + \overline{z}B^{\dag}
\qquad\mbox{and}\qquad
izB - i\overline{z}B^{\dag}.
\end{equation}

\noindent In this way we obtain mappings $\BL(H) \rightarrow \SA(H)$
in $\Vect[\R]$. If the real part $\textit{Re}(z)$ is non-zero, the
mapping:
$$\begin{array}{rcl}
B 
& \longmapsto & 
\frac{1}{2\textit{Re}(z)}\big(zB + \overline{z}B^{\dag}\big)
\end{array}$$

\noindent is a left-inverse of the inclusion a $\SA(H) \hookrightarrow
\BL(H)$, making it a split mono.

\auxproof{
Write $h(B) = zB + \overline{z}B^{\dag}$, then $h(B)$ is self-adjoint
since:
$$h(B)^{\dag}
=
(zB)^{\dag} + (\overline{z}B^{\dag})^{\dag}
=
\overline{z}B^{\dag} + zB
=
h(B).$$

\noindent Moreover, $h$ is linear (in $\Vect[\R]$):
$$\begin{array}{rcl}
h(A+B)
& = &
z(A+B) + \overline{z}(A+B)^{\dag} \\
& = &
zA + zB + \overline{z}A^{\dag} + \overline{z}B^{\dag} \\
& = &
h(A) + h(B) \\
h(rB)
& = &
z(rB) + \overline{z}(rB)^{\dag} \\
& = &
r(zB) + \overline{z}rB^{\dag}, \qquad \mbox{since }r\in\R \\
& = &
rh(B).
\end{array}$$

Similarly, write $k(B) = izB - i\overline{z}B^{\dag}$, then $k(B)$ 
is self-adjoint since:
$$k(B)^{\dag}
=
(izB)^{\dag} - (i\overline{z}B^{\dag})^{\dag}
=
-i\overline{z}B^{\dag} + izB
=
k(B).$$

\noindent Moreover, $k$ is linear (in $\Vect[\R]$):
$$\begin{array}{rcl}
k(A+B)
& = &
iz(A+B) - i\overline{z}(A+B)^{\dag} \\
& = &
izA + izB - i\overline{z}A^{\dag} - i\overline{z}B^{\dag} \\
& = &
k(A) + k(B) \\
k(rB)
& = &
iz(rB) - i\overline{z}(rB)^{\dag} \\
& = &
r(izB) - i\overline{z}rB^{\dag}, \qquad \mbox{since }r\in\R \\
& = &
rk(B).
\end{array}$$

Now write $\ell(B) = \frac{1}{2\textit{Re}(z)}\big(zB +
\overline{z}B^{\dag}\big)$ and $i\colon\SA(H)\hookrightarrow\BL(H)$.
Then for $A\in\SA(H)$,
$$\begin{array}{rcl}
\ell(i(A))
& = &
\frac{1}{2\textit{Re}(z)}\big(zA + \overline{z}A^{\dag}\big) \\
& = &
\frac{1}{2\textit{Re}(z)}\big(zA + \overline{z}A\big) \\
& = &
\frac{z+\overline{z}}{2\textit{Re}(z)}A \\
& = &
A.
\end{array}$$
}

By moving from $\BL$ to $\SA$ we get the following
analogue of Proposition~\ref{BLDualProp}.

\begin{proposition}
\label{SADualProp}
For $H\in\FdHilb$, the subset $\SA(H) \hookrightarrow \BL(H)$ of
self-adjoint operators on $H$ is a vector space over $\R$, for
which one obtains a natural isomorphism in $\Vect[\R]$:
\begin{equation}
\label{SADualEqn}
\vcenter{\xymatrix{
\SA(H)\ar[r]^-{\hs[\SA]}_-{\cong} & \SA(H)^{*} =
\SA(H) \multimap \R
\qquad\mbox{by}\qquad
{\begin{array}{rcl}
\hs[\SA](A)(B)
& = &
\tr(AB).
\end{array}}
}}
\end{equation}

\noindent It gives rise to a map of adjunctions:
$$\xymatrix{
\FdHilb\ar@/^1.2ex/[rr]^-{(-)^{\dag}}\ar[d]_{\SA} & \bot & 
   \FdHilb\rlap{$\op$}\ar@/^1.2ex/[ll]^-{(-)^{\dag}}\ar[d]^{\SA} \\
\Vect[\R]\ar@/^1.2ex/[rr]^-{(-)\multimap \R} & \bot & 
   \Vect[\R]\rlap{$\op$}\ar@/^1.2ex/[ll]^-{(-)\multimap \R}
}$$
\end{proposition}

\begin{proof}
If $A,B\colon H \rightarrow H$ are self adjoint operators, then
$\tr(AB^{\dag}) = \tr(AB)$ is a real number, since:
$$\overline{\tr(AB)}
=
\tr\big((AB)^{\dag}\big)
=
\tr\big(B^{\dag}A^{\dag}\big)
=
\tr\big(BA\big)
=
\tr(AB).$$

\noindent Conversely, suppose we have a (linear) map $f\colon \SA(H)
\rightarrow \R$ in $\Vect[\R]$. It can be extended to a function
$f'\colon \overline{\BL(H)} \rightarrow \C$ via
$$\begin{array}{rcl}
f'(B)
& = &
\frac{1}{2}\Big(f(B+B^{\dag}) + if(iB - iB^{\dag})\Big)
\end{array}$$

\noindent using, as described in~(\ref{BLtoSAEqn}), that $B+B^{\dag}$
and $iB-iB^{\dag}$ are self-adjoint. It is not hard to see that $f'$
preserves sums of operators and satisfies $f'(zB) =
\overline{z}f'(B)$. This $f'$ really extends $f$ since in the special
case when $B$ is self-adjoint we get $f'(B) = \frac{1}{2}\big(f(2B) +
if(0)\big) = f(B)$ by linearity.

\auxproof{
Linearity of $f'$.
$$\begin{array}{rcl}
\lefteqn{f'(0)} \\
& = &
\frac{1}{2}\Big(f(0) + if(0)\Big) \\
& = &
0 \\
\lefteqn{f'(B+C)} \\
& = &
\frac{1}{2}\Big(f((B+C) + (B+C)^{\dag}) +
   if(i(B+C) - i(B+C)^{\dag})\Big) \\
& = &
\frac{1}{2}\Big(f(B+C + B^{\dag} +C^{\dag}) +
   if(iB + iC - iB^{\dag} - iC^{\dag})\Big) \\
& = &
\frac{1}{2}\Big(f(B+B^{\dag}) + f(C+C^{\dag}) +
   if(iB - iB^{\dag}) + if(iC  - iC^{\dag})\Big) 
   \qquad \mbox{by~(\ref{BLtoSAEqn})} \\
& = &
\frac{1}{2}\Big(f(B+B^{\dag}) + if(iB - iB^{\dag})\Big) + 
   \frac{1}{2}\Big(f(C+C^{\dag}) + if(iC - iC^{\dag})\Big) \\
& = &
f'(B) + f'(C) \\
\lefteqn{f(zB) \qquad \mbox{for }z = a+ib} \\
& = &
\frac{1}{2}\Big(f(zB+(zB)^{\dag}) + if(izB - i(zB)^{\dag})\Big) \\
& = &
\frac{1}{2}\Big(f(zB+\overline{z}B^{\dag}) + 
   if(izB - i\overline{z}B^{\dag})\Big) \\
& = &
\frac{1}{2}\Big(f(aB + ibB + aB^{\dag} - ibB^{\dag}) + 
   if(iaB - bB - iaB^{\dag} - bB^{\dag})\Big) \\
& = &
\frac{1}{2}\Big(f(aB + aB^{\dag}) + f(ibB - ibB^{\dag}) + 
   if(iaB - iaB^{\dag}) - if(bB + bB^{\dag})\Big)
   \quad \mbox{by~(\ref{BLtoSAEqn})} \\
& = &
\frac{1}{2}\Big(af(B + B^{\dag}) - ibf(B + B^{\dag} + 
   iaf(iB - iB^{\dag}) + bf(iB - iB^{\dag}) )\Big) \\
& = &
\frac{1}{2}\Big((a-ib)f(B + B^{\dag}) +
   i(a-ib)f(iB - iB^{\dag})\Big) \\
& = &
\frac{1}{2}\overline{z}\Big(f(B + B^{\dag}) + if(iB - iB^{\dag})\Big) \\
& = &
\overline{z}f'(B).
\end{array}$$
}

By Proposition~\ref{BLDualProp} there is a unique $A\in\BL(H)$ with:
$$\begin{array}{rcccl}
f'
& = &
\hs[\BL](A)
& = &
\tr(A(-)^{\dag}) \;\colon\; \overline{\BL(H)} \longrightarrow \C. 
\end{array}$$

\noindent We now put $\hs[\SA]^{-1}(f) = \frac{1}{2}(A + A^{\dag})
\in \SA(H)$, and check for $B\in\SA(H)$:
$$\begin{array}{rcll}
\lefteqn{\hs[\SA]\big(\hs[\SA]^{-1}(f)\big)(B)} \\
& = &
\tr\big(\hs[\SA]^{-1}(f)B\big) \\
& = &
\frac{1}{2}\Big(\tr(AB) + \tr(A^{\dag}B)\Big) \\
& = &
\frac{1}{2}\Big(\tr(AB^{\dag}) + \tr((BA)^{\dag})\Big) 
   & \mbox{since $B$ is self-adjoint} \\
& = &
\frac{1}{2}\Big(f'(B) + \overline{\tr(BA)}\Big) \\
& = &
\frac{1}{2}\Big(f(B) + \overline{\tr(AB)}\Big) &
   \mbox{since $f(B) = f'(B)$ when $B\in\SA(H)$} \\
& = &
\frac{1}{2}\Big(f(B) + \overline{f(B)}\Big) \\
& = &
\frac{1}{2}\Big(f(B) + f(B)\Big) & 
   \mbox{because $f(B)$ is real valued} \\
& = &
f(B) & \mbox{since $f$ is linear.}
\end{array}$$

\noindent In the other direction, one obtains 
$\hs[\SA]^{-1}\big(\hs[\SA](A)\big) = A$ by uniqueness.

\auxproof{
Write $f = \tr(A-)\colon \SA(H) \rightarrow \R$, for $A\in\SA(H)$.
Then, for $B\in\BL(H)$ one has:
$$\begin{array}{rcl}
f'(B)
& = &
\frac{1}{2}\Big(\tr(AB) + \tr(AB^{\dag}) + i\tr(AiB) -
  i\tr(AiB^{\dag})\Big) \\
& = &
\frac{1}{2}\Big(\tr(AB) + \tr(AB^{\dag}) - \tr(AB) + \tr(AB^{\dag})\Big) \\
& = &
\tr(AB^{\dag}).
\end{array}$$

\noindent Hence $\hs[\BL]^{-1}(f') = A$, and thus:
$$\begin{array}{rcccl}
\hs[\SA]^{-1}(f)
& = &
\frac{1}{2}(A + A^{\dag}) 
& = &
A,
\end{array}$$

\noindent since $A$ is self-adjoint. Thus 
$\hs[\SA]^{-1}\big(\hs[\SA](A)\big) = A$.
}

We prove uniqueness in the self-adjoint case too. Assume a
self-adjoint operator $C\in\SA(H)$ also satisfies $f = \hs[\SA](C)
\colon \SA(H) \rightarrow \R$. We need to prove $C = A =
\frac{1}{2}(A' + A'^{\dag})$.  We plan to show $A_{jk} = C_{jk}$
wrt.\ an arbitrary basis, and thus $A=C$. We prove the equality
$A_{jk} = C_{jk}$ in two steps, by proving that both their real and
imaginary parts are the same.
$$\begin{array}{rcl}
\textit{Re}(C_{jk})
& = &
\frac{1}{2}\big(C_{jk} + \overline{C_{jk}}\big) \\
& = &
\frac{1}{2}\big(C_{jk} + (C^{\dag})_{kj}\big) \\
& = &
\frac{1}{2}\big(C_{jk} + C_{kj}\big) \\
& = &
\frac{1}{2}\big(\bra{j}C\ket{k} + \bra{k}C\ket{j}\big) \\
& = &
\frac{1}{2}\big(\tr(\bra{j}C\ket{k}) + \tr(\bra{k}C\ket{j})\big) \\
& = &
\frac{1}{2}\tr\big((C(\ket{k}\bra{j} + \ket{j}\bra{k})\big) \\
& = &
\frac{1}{2}\tr\big((A(\ket{k}\bra{j} + \ket{j}\bra{k})\big) \\
& & \qquad
\mbox{by assumption, using that
     $\ket{k}\bra{j} + \ket{j}\bra{k}$ is self-adjoint} \\
& = &
\cdots \mbox{ (as before)} \\
& = &
\textit{Re}(A_{jk}).
\end{array}$$

\noindent Similarly, $\textit{Im}(C_{jk}) = \textit{Im}(A_{jk})$, by
writing $\textit{Im}(C_{jk}) = \frac{1}{2}\big(-iC_{jk} +
i\overline{C_{jk}}\big)$ and using the self-adjoint operator
$-i\ket{j}\bra{k} + i\ket{k}\bra{j}$. \QED

\auxproof{
In order to see this we first recall that for a complex number 
$z = a + ib$ one has:
$$\begin{array}{rcl}
-iz + i\overline{z}
& = &
-i(a+ib) + i(a-ib) \\
& = &
-ia +b + ia +b \\
& = &
2b
\end{array}$$

\noindent Hence we calculate:
$$\begin{array}{rcl}
\textit{Im}(C_{jk})
& = &
\frac{1}{2}\big(-iC_{jk} + i\overline{C_{jk}}\big) \\
& = &
\frac{1}{2}\big(-iC_{jk} + i(C^{\dag})_{kj}\big) \\
& = &
\frac{1}{2}\big(-iC_{jk} + iC_{kj}\big) \\
& = &
\frac{1}{2}\big(-i\bra{j}C\ket{k} + i\bra{k}C\ket{j}\big) \\
& = &
\frac{1}{2}\big(\tr(-i\bra{j}C\ket{k}) + \tr(i\bra{k}C\ket{j})\big) \\
& = &
\frac{1}{2}\tr\big((C(-i\ket{k}\bra{j} + i\ket{j}\bra{k})\big) \\
& = &
\frac{1}{2}\tr\big((A(-i\ket{k}\bra{j} + i\ket{j}\bra{k})\big) \\
& & \qquad
\mbox{by assumption, using that
     $-i\ket{k}\bra{j} + i\ket{j}\bra{k}$ is self-adjoint} \\
& = &
\cdots \\
& = &
\textit{Im}(A_{jk}).
\end{array}$$

We make sure that the operators $\ket{k}\bra{j} + \ket{j}\bra{k}$ 
and $-i\ket{k}\bra{j} + i\ket{j}\bra{k}$
self-adjoint:
$$\begin{array}{rcl}
\big(\ket{k}\bra{j} + \ket{j}\bra{k}\big)^{\dag}
& = &
\big(\ket{k}\bra{j}\big)^{\dag} + \big(\ket{j}\bra{k}\big)^{\dag} \\
& = &
\ket{j}\bra{k} + \ket{k}\bra{j} \\
& = &
\ket{k}\bra{j} + \ket{j}\bra{k} \\
\big(- i\ket{k}\bra{j} + i\ket{j}\bra{k}\big)^{\dag}
& = &
\big(-i\ket{k}\bra{j}\big)^{\dag} + \big(i\ket{j}\bra{k}\big)^{\dag} \\
& = &
i\ket{j}\bra{k} -i\ket{k}\bra{j} \\
& = &
-i\ket{k}\bra{j} + i\ket{j}\bra{k}.
\end{array}$$
}

\end{proof}

Implicitly, the proof gives a formula for the inverse $\hs[\SA]^{-1}$
of the Hilbert-Schmidt map for self-adjoint operators.

\auxproof{
For future use we extract the formula for the inverse $\hs[\SA]^{-1}$
of the Hilbert-Schmidt map for self-adjoint operators, as obtained in
this proof.

\begin{lemma}
\label{hsSAinvLem}
Assuming a basis $\ket{1}, \ldots, \ket{n}$ for a Hilbert space $H$,
the inverse $\hs[\SA]^{-1}\colon (\SA(H)\multimap\R) \rightarrow
\SA(H)$ on a linear map $f\colon \SA(H) \rightarrow \R$ has matrix
entries given by the formula:
\begin{equation}
\label{hsSAinvEqn}
\begin{array}{rcl}
\hs[\SA]^{-1}(f)_{jk}
& = &
f\Big(\frac{\ket{j}\bra{k}+\ket{k}\bra{j}}{2}\Big) +
   if\Big(\frac{i\ket{j}\bra{k}-i\ket{k}\bra{j}}{2}\Big).
\end{array}
\end{equation}
\end{lemma}

\begin{proof}
As can be seen in the proof of Proposition~\ref{SADualProp}, the
inverse $\hs[\SA]^{-1}\colon (\SA(H)\multimap\R) \rightarrow \SA(H)$
on a linear map $f\colon \SA(H) \rightarrow \R$ is given by the
formula:
$$\begin{array}{rcl}
\hs[\SA]^{-1}(f)
& = &
\frac{1}{2}\big(A + A^{\dag}), \qquad \mbox{where:} \\
A 
& = & 
\hs[\BL]^{-1}\Big(\lamin{B}{\overline{\BL(H)}}{
   \frac{1}{2}\big(f(B+B^{\dag}) + if(iB - iB^{\dag})\big)}\Big). 
\end{array}$$

\noindent When we unravel the formula~(\ref{hsBLinvEqn}) for
$\hs[\BL]^{-1}$ we obtain as matrix entries for this inverse
$\hs[\SA]^{-1}(f)$:
$$\begin{array}[b]{rcl}
\hs[\SA]^{-1}(f)_{jk}
& = &
\frac{1}{2}\big(A_{jk} + \overline{A_{kj}}) 
   \qquad\mbox{for $A$ as above} \\
& = & 
\frac{1}{4}\Big(f\big(\ket{j}\bra{k}+\ket{k}\bra{j}\big) +
   if\big(i\ket{j}\bra{k}-i\ket{k}\bra{j}\big) \; + \\
& & \qquad
   \overline{f\big(\ket{k}\bra{j}+\ket{j}\bra{k}\big) +
   if\big(i\ket{k}\bra{j}-i\ket{j}\bra{k}\big)}\Big) \\
& = & 
\frac{1}{4}\Big(f\big(\ket{j}\bra{k}+\ket{k}\bra{j}\big) +
   if\big(i\ket{j}\bra{k}-i\ket{k}\bra{j}\big) \; + \\
& & \qquad
   f\big(\ket{k}\bra{j}+\ket{j}\bra{k}\big) -
   if\big(i\ket{k}\bra{j}-i\ket{j}\bra{k}\big)\Big) \\
& = &
\frac{1}{4}\Big(f\big(2\ket{j}\bra{k}+2\ket{k}\bra{j}\big) +
   if\big(2i\ket{j}\bra{k}-2i\ket{k}\bra{j}\big)\Big) \\
& = &
f\Big(\frac{\ket{j}\bra{k}+\ket{k}\bra{j}}{2}\Big) +
   if\Big(\frac{i\ket{j}\bra{k}-i\ket{k}\bra{j}}{2}\Big).
\end{array}\eqno{\QEDbox}$$
\end{proof}

One can also direct prove that this formula gives the inverse.
In that case the following equations for a self-adjoint matrix are
convenient.
$$\begin{array}{rclcrcl}
A 
& = &
\frac{1}{2}\sum_{j,k}A_{jk}\ket{j}\bra{k} +
   \overline{A_{jk}}\ket{k}\bra{j}
& \quad &
0
& = &
\sum_{j,k}A_{jk}\ket{j}\bra{k} - \overline{A_{jk}}\ket{k}\bra{j}.
\end{array}$$

Just to be sure we do the calculations, to see which properties
we need where. For $A\in\SA(H)$,
$$\begin{array}{rcl}
\lefteqn{\big(\hs[\SA]^{-1} \after \hs[\SA]\big)(A)_{jk}} \\
& = &
\hs[\SA]^{-1}(\tr(A-))_{jk} \\
& = &
\tr\Big(A\big(\frac{\ket{j}\bra{k}+\ket{k}\bra{j}}{2}\big)\Big) +
   i\tr\Big(A\big(\frac{i\ket{j}\bra{k}-i\ket{k}\bra{j}}{2}\big)\Big) \\
& = &
\frac{1}{2}\Big(\tr(\bra{k}A\ket{j}) + \tr(\bra{j}A\ket{k}) -
   \tr(\bra{k}A\ket{j}) + \tr(\bra{j}A\ket{k})\Big) \\
& = &
\frac{1}{2}\big(A_{jk} + A_{jk}\big) \\
& = &
A_{jk}.
\end{array}$$

\noindent Next assume $f\colon \SA(H) \rightarrow \R$ in $\Vect[\R]$
and $A\in\SA(H)$, where we write $A_{kj} = B_{kj} + iC_{kj}$ with
$B_{kj},C_{kj}\in\R$. Since $\overline{A_{jk}} = A_{kj}$ we have
$B_{jk} = B_{kj}$ and $C_{jk} = -C_{kj}$. Hence:
$$\begin{array}{rcl}
\lefteqn{\big(\hs[\SA] \after \hs[\SA]^{-1}\big)(f)(A)} \\
& = &
\tr\big(\hs[\SA]^{-1}(f)A\big) \\
& = &
\sum_{j}\big(\hs[\SA]^{-1}(f)A\big)_{jj} \\
& = &
\sum_{j,k}\hs[\SA]^{-1}(f)_{jk}A_{kj} \\
& = &
\sum_{j,k}\Big[f\Big(\frac{\ket{j}\bra{k}+\ket{k}\bra{j}}{2}\Big) +
   if\Big(\frac{i\ket{j}\bra{k}-i\ket{k}\bra{j}}{2}\Big)\Big]
   \big(B_{kj} + iC_{kj}\big) \\
& = &
\frac{1}{2}\sum_{j,k} B_{kj}f(\ket{j}\bra{k}+\ket{k}\bra{j}) +
   iC_{kj}f(\ket{j}\bra{k}+\ket{k}\bra{j}) \; + \\
& & \qquad iB_{kj}f(i\ket{j}\bra{k}-i\ket{k}\bra{j}) -
   C_{kj}f(i\ket{j}\bra{k}-i\ket{k}\bra{j}) \\
& = &
\frac{1}{2}\sum_{j,k} f\big(B_{kj}\ket{j}\bra{k}+B_{kj}\ket{k}\bra{j} -
   iC_{kj}\ket{j}\bra{k} + iC_{kj}\ket{k}\bra{j}\big) \; + \\
& & \qquad if\big(iB_{kj}\ket{j}\bra{k} - iB_{kj}\ket{k}\bra{j} +
   C_{kj}\ket{j}\bra{k} + C_{kj}\ket{k}\bra{j}\big) \\
& = &
\frac{1}{2}\sum_{j,k} f\big(B_{jk}\ket{j}\bra{k}+B_{kj}\ket{k}\bra{j} +
   iC_{jk}\ket{j}\bra{k} + iC_{kj}\ket{k}\bra{j}\big) \; + \\
& & \qquad if\big(iB_{jk}\ket{j}\bra{k} - iB_{kj}\ket{k}\bra{j} -
   C_{jk}\ket{j}\bra{k} + C_{kj}\ket{k}\bra{j}\big) \\
& = &
\frac{1}{2}\sum_{j,k} f\big(A_{jk}\ket{j}\bra{k}+A_{kj}\ket{k}\bra{j}\big) +
   if\big(iA_{jk}\ket{j}\bra{k} - iA_{kj}\ket{k}\bra{j}\big) \\
& = &
\frac{1}{2}\sum_{j,k} f\big(A_{jk}\ket{j}\bra{k}+
   \overline{A_{jk}}\ket{k}\bra{j}\big) +
   if\big(iA_{jk}\ket{j}\bra{k} - i\overline{A_{jk}}\ket{k}\bra{j}\big) \\
& = &
\frac{1}{2}\Big( f\big(\sum_{j,k}A_{jk}\ket{j}\bra{k}+
   \overline{A_{jk}}\ket{k}\bra{j}\big) +
   if\big(i\sum_{j,k}A_{jk}\ket{j}\bra{k} - 
      \overline{A_{jk}}\ket{k}\bra{j}\big)\Big) \\
& \smash{\stackrel{(*)}{=}} &
\frac{1}{2}\Big(f(2A) + if(0)\Big) \\
& = &
f(A).
\end{array}$$

The marked equation relies on the following formulas for self-adjoint $A$,
$$\begin{array}{rclcrcl}
A 
& = &
\frac{1}{2}\sum_{j,k}A_{jk}\ket{j}\bra{k} +
   \overline{A_{jk}}\ket{k}\bra{j}
& \qquad
&
0
& = &
\sum_{j,k}A_{jk}\ket{j}\bra{k} - \overline{A_{jk}}\ket{k}\bra{j}
\end{array}$$

The proofs go as follows.
$$\begin{array}{rcl}
A
& = &
\frac{1}{2}\sum_{j,k}(A_{jk} + A_{jk})\ket{j}\bra{k} \\
& = &
\frac{1}{2}\sum_{j,k}(A_{jk} + \overline{A_{kj}})\ket{j}\bra{k} \\
& = &
\frac{1}{2}\Big(\sum_{j,k}A_{jk}\ket{j}\bra{k} + 
   \sum_{j,k}\overline{A_{kj}}\ket{j}\bra{k}\Big) \\
& = &
\frac{1}{2}\Big(\sum_{j,k}A_{jk}\ket{j}\bra{k} + 
   \sum_{j,k}\overline{A_{jk}}\ket{k}\bra{j}\Big) \\
& = &
\frac{1}{2}\sum_{j,k}A_{jk}\ket{j}\bra{k} + 
   \overline{A_{jk}}\ket{k}\bra{j} \\
0
& = &
\frac{1}{2}\sum_{j,k}(A_{jk} - A_{jk})\ket{j}\bra{k} \\
& = &
\frac{1}{2}\sum_{j,k}(A_{jk} - \overline{A_{kj}})\ket{j}\bra{k} \\
& = &
\frac{1}{2}\Big(\sum_{j,k}A_{jk}\ket{j}\bra{k} - 
   \sum_{j,k}\overline{A_{kj}}\ket{j}\bra{k}\Big) \\
& = &
\frac{1}{2}\Big(\sum_{j,k}A_{jk}\ket{j}\bra{k} - 
   \sum_{j,k}\overline{A_{jk}}\ket{k}\bra{j}\Big) \\
& = &
\frac{1}{2}\sum_{j,k}A_{jk}\ket{j}\bra{k} - \overline{A_{jk}}\ket{k}\bra{j}
\end{array}$$
}

\subsection*{Positive operators}

An operator $A\colon H\rightarrow H$ is called positive if the inner
product $\inprod{Ax}{x}$ is a non-negative real number, for each $x\in
H$. In that case one writes $A\geq 0$. This is equivalent to: $A =
BB^{\dag}$, for some operator $B$, and also to: all eigenvalues are
non-negative reals.  In a spectral decomposition $A =
\sum_{j}\lambda_{j}\ket{j}\bra{j}$ a positive operator $A$ has
eigenvalues $\lambda_{j} \in \Rnn$ for all $j$. Hence the trace
$\tr(A)$ is a non-negative real number. The set of positive operators
on $H$ is written here as $\Pos(H)$. It forms a module over the
semiring $\Rnn$ of non-negative reals since positive operators are
closed under addition and under scalar multiplication with
$r\in\Rnn$. A positive operator is clearly self-adjoint, since
$A^{\dag} = (BB^{\dag})^{\dag} = B^{\dag\dag}B^{\dag} = BB^{\dag} =
A$. Thus there are inclusion maps $\Pos(H) \hookrightarrow \SA(H)
\hookrightarrow \BL(H)$.  We can describe taking positive operators as
a functor $\Pos\colon\Hilb\rightarrow \Mod[\Rnn]$ from Hilbert spaces
to modules over the non-negative real numbers. The action of $\Pos$ on
maps is like for $\SA$ and $\BL$ in~(\ref{BLVectEqn}), and is
well-defined, since if $C\colon H\rightarrow K$ in $\Hilb$ and $A\geq
0$, then $\Pos(C)(A) = CAC^{\dag} \geq 0$ since for each $x\in K$,
$$\begin{array}{rcccl}
\inprod{CAC^{\dag}x}{x}
& = &
\inprod{AC^{\dag}x}{C^{\dag}x}
& \geq &
0.
\end{array}$$

\noindent As an aside we recall that via positivity one obtains the
L\"owner order on arbitrary operators $A,B$, defined as: $A \leq B$
iff $B-A\geq 0$. Thus: $A\leq B$ iff $\exin{P}{\Pos(H)}{A+P =
  B}$. Hence the spaces $\Pos(H) \hookrightarrow \SA(H)
\hookrightarrow \BL(H)$ are actually ordered (see
also~\cite{Selinger04,Edalat04}).


\auxproof{
Suppose $A = BB^{\dag}$. Then $\inprod{Ax}{x} = \inprod{BB^{\dag}x}{x}
= \inprod{B^{\dag}x}{B^{\dag}x} \geq 0$.

Conversely, if $\inprod{Ax}{x} \geq 0$ and $A$ is self-adjoint, then
we can form its spectral decomposition $A =
\sum_{j}\lambda_{j}\ket{j}\bra{j}$ and conclude that $\lambda_{j} \geq
0$ for each $j$. Hence we take $B =
\sum_{j}\sqrt{\lambda_{j}}\ket{j}\bra{j}$. It is self-adjoint and
satisfies $BB^{\dag} = A$.

Next suppose $A$ is self-adjoint and positive, with spectral
decomposition $A = \sum_{j}\lambda_{j}\ket{j}\bra{j}$. Then for each
eigenvector $0 \leq \inprod{A\ket{k}}{k} = \lambda_{k}\inprod{k}{k}$.
Since $\inprod{k}{k} \geq 0$ also $\lambda_{k} \geq 0$.

Conversely, assume $A = \sum_{j}\lambda_{j}\ket{j}\bra{j}$ and
$\lambda_{j} \geq 0$. Then write an arbitrary element $x = \sum_{j}
x_{j}\ket{j}$, wrt. the basis $(\ket{j})$ of eigenvectors, so that $Ax
= \sum_{j}\lambda_{j}x_{j}$. But then $\inprod{Ax}{x} =
\sum_{j}\lambda_{j}x_{j}\overline{x_j} =
\sum_{j}\lambda_{j}|x_{j}|^{2} \geq 0$.
}

\auxproof{
We also include the verification that the l\"owner order is
a partial order on $\BL(H)$.
\begin{itemize}
\item Reflexivity is easy: $A-A = 0$ is positive.

\item Transitivity is also easy: if $A\leq B \leq C$, then $B-A \geq
  0$ and $C-B \geq 0$. Hence also $C - A = (C-B) + (B-A) \geq 0$, so
  $A\leq C$.

\item For antisymmetry assume $A\leq B$ and $B\leq A$, so that
  $B-A\geq 0$ and $A-B \geq 0$. This means that $\inprod{(A-B)x}{x} =
  0$ for all $x$. Hence we are done if can prove:
$$\all{x}{\inprod{Cx}{x}=0} \Longrightarrow C=0$$

\noindent This follows from a result that a bilinear form is
determined by its diagonal values (see
\textit{e.g.}~\cite[\S\S3]{Halmos57}). Write $\widehat{C}(x) = 
\inprod{Cx}{x}$. Then:
$$\begin{array}{rcl}
\inprod{Cx}{y}
& = &
\widehat{C}(\frac{1}{2}(x+y)) - 
   \widehat{C}(\frac{1}{2}(x-y)) +
  i\widehat{C}(\frac{1}{2}(x+iy)) -
    i\widehat{C}(\frac{1}{2}(x-iy)).
\end{array}\eqno{(*)}$$

\noindent Then we are done because it gives $\inprod{Cx}{y} = 0$ for
all $x,y$. In particular $\inprod{Cx}{Cx} = 0$, and thus $Cx=0$.
Since this holds for all $x$, we get $C=0$, as required.

For a proof of this equation (*) we first notice that:
$$\begin{array}{rcl}
\widehat{C}(\frac{1}{2}(x+y))
& = &
\inprod{\frac{1}{2}C(x+y)}{\frac{1}{2}(x+y)} \\
& = &
\frac{1}{2}\inprod{Cx}{\frac{1}{2}(x+y)} + 
   \frac{1}{2}\inprod{Cy}{\frac{1}{2}(x+y)} \\
& = &
\frac{1}{4}\inprod{Cx}{x} + \frac{1}{4}\inprod{Cx}{y} +
   \frac{1}{4}\inprod{Cy}{x} + \frac{1}{4}\inprod{Cy}{y} \\
\widehat{C}(\frac{1}{2}(x-y))
& = &
\inprod{\frac{1}{2}C(x-y)}{\frac{1}{2}(x-y)} \\
& = &
\frac{1}{2}\inprod{Cx}{\frac{1}{2}(x-y)} - 
   \frac{1}{2}\inprod{Cy}{\frac{1}{2}(x-y)} \\
& = &
\frac{1}{4}\inprod{Cx}{x} - \frac{1}{4}\inprod{Cx}{y} -
   \frac{1}{4}\inprod{Cy}{x} + \frac{1}{4}\inprod{Cy}{y}
\end{array}$$

\noindent Hence as first intermediate result we have:
$$\begin{array}{rcl}
\widehat{C}(\frac{1}{2}(x+y)) - \widehat{C}(\frac{1}{2}(x-y))
& = &
\frac{1}{2}\inprod{Cx}{y} + \frac{1}{2}\inprod{Cy}{x}.
\end{array}$$

\noindent And thus (using conjugate linearity in the second argument):
$$\begin{array}{rcl}
\widehat{C}(\frac{1}{2}(x+iy)) - \widehat{C}(\frac{1}{2}(x-iy))
& = &
-\frac{i}{2}\inprod{Cx}{y} + \frac{i}{2}\inprod{Cy}{x}.
\end{array}$$

\noindent By multiplying with $i$ on both sides we get:
$$\begin{array}{rcl}
i\widehat{C}(\frac{1}{2}(x+iy)) - i\widehat{C}(\frac{1}{2}(x-iy))
& = &
\frac{1}{2}\inprod{Cx}{y} - \frac{1}{2}\inprod{Cy}{x}.
\end{array}$$

\noindent Hence by adding this to the previous intermediate
result we get the required equation:
$$\begin{array}{rcl}
\lefteqn{\textstyle
\Big(\widehat{C}(\frac{1}{2}(x+y)) - 
   \widehat{C}(\frac{1}{2}(x-y))\Big)
  +\Big(i\widehat{C}(\frac{1}{2}(x+iy)) - 
    i\widehat{C}(\frac{1}{2}(x-iy))\Big)} \\
& = &
\Big(\frac{1}{2}\inprod{Cx}{y} + \frac{1}{2}\inprod{Cy}{x}\Big)
   + \Big(\frac{1}{2}\inprod{Cx}{y} - \frac{1}{2}\inprod{Cy}{x}\Big) \\
& = &
\inprod{Cx}{y}.
\end{array}$$

\end{itemize}
}

\begin{proposition}
\label{PosDualProp}
For $H\in\FdHilb$, the subset $\Pos(H) \hookrightarrow \SA(H)$ of
positive operators is a module over the non-negative reals
$\Rnn$, for which there is a natural isomorphism in $\Mod[\Rnn]$:
\begin{equation}
\label{PosDualEqn}
\vcenter{\xymatrix{
\Pos(H)\ar[r]^-{\hs[\Pos]}_-{\cong} & \Pos(H)^{*} =
\Pos(H) \multimap \Rnn
\quad\mbox{by}\quad
{\begin{array}{rcl}
\hs[\Pos](A)(B)
& = &
\tr(AB).
\end{array}}
}}
\end{equation}

\noindent This isomorphism gives rise to a map of adjunctions:
$$\xymatrix{
\FdHilb\ar@/^1.2ex/[rr]^-{(-)^{\dag}}\ar[d]_{\Pos} & \bot & 
   \FdHilb\rlap{$\op$}\ar@/^1.2ex/[ll]^-{(-)^{\dag}}\ar[d]^{\Pos} \\
\Mod[\Rnn]\ar@/^1.2ex/[rr]^-{(-)\multimap \Rnn} 
   & \bot & 
   \Mod[\Rnn]\rlap{$\op$}
   \ar@/^1.2ex/[ll]^-{(-)\multimap \Rnn}
}$$
\end{proposition}

\begin{proof}
We first have to check that $\tr(AB) \geq 0$, for $A,B\in\Pos(H)$, so
that indeed $\tr(A-)$ has type $\Pos(H)\rightarrow\Rnn$. We do so by
first writing the spectral decomposition as $A =
\sum_{j}\lambda_{j}\ket{j}\bra{j}$, with $\lambda_{j} \geq 0$. Then:
$$\begin{array}{rcl}
\tr(AB)
\hspace*{\arraycolsep} = \hspace*{\arraycolsep}
\sum_{j}\lambda_{j}\tr(\ket{j}\bra{j}B)
& = &
\sum_{j}\lambda_{j}\tr(\bra{j}B\ket{j}) \\
& = &
\sum_{j}\lambda_{j}\tr(\inprod{Bj}{j}) \\
& = &
\sum_{j}\lambda_{j}\inprod{Bj}{j} \\
& \geq &
0, \qquad \mbox{since }\inprod{Bj}{j} \geq 0.
\end{array}$$

\noindent These maps $\hs[\Pos] = \tr(A-)$ clearly preserve the module
structure: additions and scalar multiplication (with a non-negative
real number).  Next, assume we have a linear map $f\colon \Pos(H)
\rightarrow \Rnn$ in $\Mod[\Rnn]$. Like before, we wish to extend it,
this time to a map $f'\colon \SA(H) \multimap \R$. If we have an
arbitrary self-adjoint operator $B\in\SA(H)$ we can write it as
difference $B = B_{p} - B_{n}$ of its positive and negative parts
$B_{p}, B_{n}\in\Pos(H)$. One way to do it is to write $B =
\sum_{j}\lambda_{j}\ket{j}\bra{j}$ as spectral decomposition, and to
separate the (real-valued) eigenvalues $\lambda_j$ into negative and
non-negative ones. Then take:
\begin{equation}
\label{PosNegOperatorEqn}
\begin{array}{rcccccl}
B_{p}
& = &
\displaystyle\sum_{\lambda_{j}\geq 0}\lambda_{j}\ket{j}\bra{j}
& \quad\mbox{and}\quad &
B_{n}
& = &
\displaystyle\sum_{\lambda_{j}< 0}-\lambda_{j}\ket{j}\bra{j}.
\end{array}
\end{equation}

\noindent Now we can define $f'(B) = f(B_{p}) - f(B_{n}) \in
\R$. This outcome is independent of the choice of $B_{p},
B_{n}$, since if $C,D\in\Pos(H)$ also satisfy $B = C - D$, then
$B_{p} + D = C + B_{n}$, so that by linearity:
$$f(B_{p}) + f(D) = f(B_{p} + D) = f(C + B_{n}) = f(C) + f(B_{n}),$$

\noindent and thus:
$$f'(B)
=
f(B_{p}) - f(B_{n})
=
f(C) - f(D).$$

\noindent It is not hard to see that the resulting function $f'\colon
\SA(H) \rightarrow \R$ is linear (in $\Vect[\R]$). Hence by
Proposition~\ref{SADualProp} there is a unique $A =
\hs[\SA]^{-1}(f')\in\SA(H)$ with $f' = \hs[\SA](A) = \tr(A-) \colon
\SA(H) \rightarrow \R$. For a positive operator $B\in\Pos(H)$ we then
get $\tr(AB) = f'(B) = f(B) \geq 0$, since $B = B_{p}$ for such a
positive $B$.

\auxproof{
Linearity of $f'$.
$$\begin{array}{rcl}
f'(0)
& = &
f(0) - f(0) \\
& = &
0 - 0 \\
& = &
0 \\
f'(B) + f'(C)
& = &
\big(f(B_{p}) - f(B_{n})\big) + \big(f(C_{p}) - f(C_{n})\big) \\
& = &
\big(f(B_{p}) + f(C_{p})\big) - \big(f(B_{n}) + f(C_{n})\big) \\
& = &
\big(f(B_{p} + C_{p})\big) - \big(f(B_{n} + C_{n})\big) \\
& = &
f'(B) + f'(C), \qquad
   \mbox{since $B+C = (B_{p}+C_{p}) - (B_{n}+C_{n})$} \\
f'(rB)
& = &
f((rB)_{p}) - f((rB)_{n}) \\
& = &
\left\{\begin{array}{ll}
f(rB_{p}) - f(rB_{n}) & \quad \mbox{if } r \geq 0 \\
f((-r)B_{n}) - f((-r)B_{p}) & \quad \mbox{if } r < 0
\end{array}\right. \\
& = &
\left\{\begin{array}{ll}
rf(B_{p}) - rf(B_{n}) & \quad \mbox{if } r \geq 0 \\
(-r)f(B_{n}) - (-r)f(B_{p}) & \quad \mbox{if } r < 0
\end{array}\right. \\
& = &
rf(B_{p}) - rf(B_{n}) \\
& = &
rf'(B).
\end{array}$$
}

We now write $A = A_{p} - A_{n}$ as in~(\ref{PosNegOperatorEqn}),
where $A_{n} = \sum_{\lambda_{j}< 0}-\lambda_{j}\ket{j}\bra{j}$.
Projection operators of the form $\ket{j}\bra{j}$ are positive, so
that we get for each $j$ with $\lambda_{j} < 0$
$$\begin{array}{rcl}
0 
\hspace*{\arraycolsep} \leq \hspace*{\arraycolsep}
\tr(A\ket{j}\bra{j})
& = &
\tr(A_{p}\ket{j}\bra{j}) - \tr(A_{n}\ket{j}\bra{j}) \\
& = &
0 - (-\lambda_{j}) \\
& = &
\lambda_{j}.
\end{array}$$

\noindent But this is impossible, since we assumed $\lambda_{j} < 0$.
Hence $A_{n} = 0$, and $A = A_{p}$ is a positive operator. Thus we
have $f = \tr(A-)\colon \Pos(H) \rightarrow \Rnn$, as required, so
that we can take $\hs[\Pos]^{-1}(f) = \hs[\SA]^{-1}(f')$.

We briefly check uniqueness: if $C\in\Pos(H)$ also satisfies $f =
\tr(C-)$, then for an arbitrary $B\in\SA(H)$,
$$f'(B)
=
f(B_{p}) - f(B_{n})
=
\tr(CB_{p}) - \tr(CB_{n})
=
\tr(C(B_{p}-B_{n}))
=
\tr(CB).$$

\noindent But then $C=A$ by the uniqueness from
Proposition~\ref{SADualProp}. \QED
\end{proof}

\auxproof{
Like in Lemma~\ref{hsSAinvLem} we would like to have an explicit
formula for the inverse $\hs[\Pos]^{-1}$ of the Hilbert-Schmidt map
for positive operators.

\begin{lemma}
\label{hsPosinvLem}
Assume Hilbert space $H$ has basis $\ket{1}, \ldots, \ket{n}$.  For a
linear map $f\colon\Pos(H)\rightarrow\Rnn$ in $\Mod[\Rnn]$ we have:
\begin{equation}
\label{hsPosinvEqn}
\hspace*{-1em}\begin{array}{rcl}
\hs[\Pos]^{-1}(f)_{jk}
& = &
f\Big(\frac{(\ket{j}+\ket{k})(\bra{j}+\bra{k})}{4}\Big)
  - f\Big(\frac{(\ket{j}-\ket{k})(\bra{j}-\bra{k})}{4}\Big) \; + \\
& & \quad if\Big(\frac{(i\ket{j}+\ket{k})(-i\bra{j}+\bra{k})}{4}\Big) -
   if\Big(\frac{(i\ket{j}-\ket{k})(-i\bra{j}-\bra{k})}{4}\Big).
\end{array}
\end{equation}
\end{lemma}

\begin{proof}
The first step is to separate the self-adjoint matrices $U =
\frac{\ket{j}\bra{k} + \ket{k}\bra{j}}{2}$ and $V =
\frac{i\ket{j}\bra{k} - i\ket{k}\bra{j}}{2}$ from the
equation~(\ref{hsSAinvEqn}) for $\hs[\SA]^{-1}$ into positive and
negative parts. One has $U = U_{p} - U_{n}$ for positive operators:
$$\begin{array}{rclcrcl}
U_{p}
& = &
\frac{(\ket{j}+\ket{k})(\bra{j}+\bra{k})}{4}
& \qquad &
U_{n}
& = &
\frac{(\ket{j}-\ket{k})(\bra{j}-\bra{k})}{4}.
\end{array}$$

\noindent For instance, $U_p$ is positive, since for an arbitrary
element $x = \sum_{\ell}x_{\ell}\ket{\ell}$, then $U_{p}x = 
(x_{j}+x_{k})(\ket{j} + \ket{k})$ and so:
$$\begin{array}{rcl}
\inprod{U_{p}x}{x}
& = &
\frac{1}{4}\big(\overline{x_j}(x_{j}+x_{k}) + 
   \overline{x_k}(x_{j}+x_{k})\big) \\
& = &
\frac{1}{4}\overline{(x_{j}+x_{k})}(x_{j}+x_{k}) \\
& = &
\frac{1}{4}|x_{j}+x_{k}|^{2} \\
& \geq &
0.
\end{array}$$

Similarly, $U_{n}x = (x_{j}-x_{k})(\ket{j}-\ket{k})$ and so:
$$\begin{array}{rcl}
\inprod{U_{n}x}{x}
& = &
\frac{1}{4}\big(\overline{x_j}(x_{j}-x_{k}) - 
   \overline{x_k}(x_{j}+x_{k})\big) \\
& = &
\frac{1}{4}\overline{(x_{j}-x_{k})}(x_{j}-x_{k}) \\
& = &
\frac{1}{4}|x_{j}-x_{k}|^{2} \\
& \geq &
0.
\end{array}$$

\noindent Further,
$$\begin{array}{rcl}
U_{p} - U_{n}
& = &
\frac{1}{4}\Big((\ket{j}+\ket{k})(\bra{j}+\bra{k}) - 
  (\ket{j}-\ket{k})(\bra{j}-\bra{k})\Big) \\
& = &
\frac{1}{4}\Big(\ket{j}\bra{j} + \ket{j}\bra{k} + 
   \ket{k}\bra{j} + \ket{k}\bra{k} \; - \\
& & \qquad
  \ket{j}\bra{j} + \ket{j}\bra{k} + \ket{k}\bra{j} - \ket{k}\bra{k}\Big) \\
& = &
\frac{1}{2}\Big(\ket{j}\bra{k} + \ket{k}\bra{j}\Big) \\
& = &
U.
\end{array}$$

\noindent Similarly, $V = V_{p} - V_{n}$ for positive operators:
$$\begin{array}{rclcrcl}
V_{p}
& = &
\frac{(i\ket{j}+\ket{k})(-i\bra{j}+\bra{k})}{4}
& \qquad &
V_{n}
& = &
\frac{(i\ket{j}-\ket{k})(-i\bra{j}-\bra{k})}{4}.
\end{array}$$

We check, again for $x = \sum_{\ell}x_{\ell}\ket{\ell}$,
$$\begin{array}{rcl}
\inprod{V_{p}x}{x}
& = &
\frac{1}{4}\inprod{(-ix_{j}+x_{k})(i\ket{j}+\ket{k})}{x} \\
& = &
\frac{1}{4}\Big(i\overline{x_j}(-ix_{j}+x_{k}) + 
   \overline{x_k}(-ix_{j}+x_{k})\Big) \\
& = &
\frac{1}{4}\overline{(-ix_{j}+x_{k})}(-ix_{j}+x_{k}) \\
& = &
\frac{1}{4}|-ix_{j}+x_{k}|^{2} \\
& \geq &
0 \\
\inprod{V_{n}x}{x}
& = &
\frac{1}{4}\inprod{(-ix_{j}-x_{k})(i\ket{j}-\ket{k})}{x} \\
& = &
\frac{1}{4}\Big(i\overline{x_j}(-ix_{j}-x_{k}) -
   \overline{x_k}(-ix_{j}-x_{k})\Big) \\
& = &
\frac{1}{4}\overline{(-ix_{j}-x_{k})}(-ix_{j}-x_{k}) \\
& = &
\frac{1}{4}|-ix_{j}-x_{k}|^{2} \\
& \geq &
0
\end{array}$$

\noindent Further,
$$\begin{array}{rcl}
V_{p} - V_{n}
& = &
\frac{1}{4}\Big((i\ket{j}+\ket{k})(-i\bra{j}+\bra{k}) - 
  (i\ket{j}-\ket{k})(-i\bra{j}-\bra{k})\Big) \\
& = &
\frac{1}{4}\Big(\ket{j}\bra{j} + i\ket{j}\bra{k} -
   i\ket{k}\bra{j} + \ket{k}\bra{k} \; - \\
& & \qquad
  \ket{j}\bra{j} + i\ket{j}\bra{k} - i\ket{k}\bra{j} - 
   \ket{k}\bra{k}\Big) \\
& = &
\frac{1}{2}\Big(i\ket{j}\bra{k} - i\ket{k}\bra{j}\Big) \\
& = &
U.
\end{array}$$

\noindent Now we can extract the formula for $\hs[\Pos]^{-1}(f)$ from
the proof of Proposition~\ref{PosDualProp} and derive the following
matrix entries:
$$\hspace*{-.5em}\begin{array}[b]{rcl}
\hs[\Pos]^{-1}(f)
& = &
\hs[\SA]^{-1}\big(\lamin{B}{\SA(H)}{f(B_{p}) - f(B_{n})}\big)_{jk} \\
& = &
f\Big(\big(\frac{\ket{j}\bra{k} + \ket{k}\bra{j}}{2}\big)_{p}\Big) -
f\Big(\big(\frac{\ket{j}\bra{k} + \ket{k}\bra{j}}{2}\big)_{n}\Big) \\
& & \quad
if\Big(\big(\frac{i\ket{j}\bra{k} - i\ket{k}\bra{j}}{2}\big)_{p}\Big) -
if\Big(\big(\frac{i\ket{j}\bra{k} - i\ket{k}\bra{j}}{2}\big)_{n}\Big)
   \quad\mbox{by~\rlap{(\ref{hsPosinvEqn})}} \\
& = &
f\Big(\frac{(\ket{j}+\ket{k})(\bra{j}+\bra{k})}{4}\Big)
  - f\Big(\frac{(\ket{j}-\ket{k})(\bra{j}-\bra{k})}{4}\Big) \; + \\
& & \quad if\Big(\frac{(i\ket{j}+\ket{k})(-i\bra{j}+\bra{k})}{4}\Big) -
   if\Big(\frac{(i\ket{j}-\ket{k})(-i\bra{j}-\bra{k})}{4}\Big).
\end{array}\eqno{\QEDbox}$$
\end{proof}

We check the isomorphism using this formulation. First, for 
$A\in\Pos(H)$,
$$\begin{array}{rcl}
\lefteqn{\big(\hs[\Pos]^{-1} \after \hs[\Pos]\big)(A)_{jk}} \\
& = &
\hs[\Pos]^{-1}\big(\tr(A-)\big)_{jk} \\
& = &
\frac{1}{4}\Big[\tr\Big(A(\ket{j}+\ket{k})(\bra{j}+\bra{k})\Big) -
   \tr\Big(A(\ket{j}-\ket{k})(\bra{j}-\bra{k})\Big) \; + \\
& & \qquad
   i\tr\Big(A(i\ket{j}+\ket{k})(-i\bra{j}+\bra{k})\Big) -
   i\tr\Big(A(i\ket{j}-\ket{k})(-i\bra{j}-\bra{k})\Big)\Big] \\
& = &
\frac{1}{4}\Big[\tr\Big(A\big((\ket{j}+\ket{k})(\bra{j}+\bra{k}) -
   (\ket{j}-\ket{k})(\bra{j}-\bra{k})\big)\Big) \; + \\
& & \qquad
   i\tr\Big(A\big((i\ket{j}+\ket{k})(-i\bra{j}+\bra{k}) -
   (i\ket{j}-\ket{k})(-i\bra{j}-\bra{k})\big)\Big)\Big] \\
& = &
\frac{1}{4}\Big[\tr\Big(A\big(2\ket{j}\bra{k} + 
   2\ket{k}\bra{j}\big)\Big) + 
   i\tr\Big(A\big((2i\ket{j}\bra{k}) - 2i\ket{k}\bra{j}\big)\Big)\Big] \\
& = &
\frac{1}{2}\Big(\tr(\bra{k}A\ket{j}) + \tr(\bra{j}A\ket{k}) -
   \tr(\bra{k}A\ket{j}) + \tr(\bra{j}A\ket{k})\Big) \\
& = &
\frac{1}{2}\big(A_{jk} + A_{jk}\big) \\
& = &
A_{jk}.
\end{array}$$

Next assume a linear map $f\colon\Pos(H)\rightarrow\Rnn$ and a
positive operator $A\in\Pos(H)$, where we write $A_{kj} = B_{kj} +
iC_{kj}$ with $B_{kj},C_{kj}\in\R$. Since $\overline{A_{jk}} = A_{kj}$
we have $B_{jk} = B_{kj}$ and $C_{jk} = -C_{kj}$. We need to
distinguish whether these $B_{kj}, C_{kj}$ are non-negative or not,
because $f$ only preserves scalar multiplication with non-negative
scalars. We use the (non-indexed) abbreviations $U,V$ as in the
proof above.
$$\begin{array}{rcl}
\lefteqn{\big(\hs[\Pos] \after \hs[\Pos]^{-1}\big)(f)(A)} \\
& = &
\tr\big(\hs[\Pos]^{-1}(f)A\big) \\
& = &
\sum_{j}\big(\hs[\Pos]^{-1}(f)A\big)_{jj} \\
& = &
\sum_{j,k}\hs[\Pos]^{-1}(f)_{jk}A_{kj} \\
& = &
\sum_{j,k}\Big[f\Big(\frac{(\ket{j}+\ket{k})(\bra{j}+\bra{k})}{4}\Big)
  - f\Big(\frac{(\ket{j}-\ket{k})(\bra{j}-\bra{k})}{4}\Big) \; + \\
& & \quad if\Big(\frac{(i\ket{j}+\ket{k})(-i\bra{j}+\bra{k})}{4}\Big) -
   if\Big(\frac{(i\ket{j}-\ket{k})(-i\bra{j}-\bra{k})}{4}\Big)\Big]
   (B_{kj} + iC_{kj}) \\
& = &
\sum_{j,k}\Big[B_{kj}f(U_{p}) - B_{kj}f(U_{n}) + iB_{kj}f(V_{p}) 
   - iB_{kj}f(V_{n}) \; + \\
& & \qquad iC_{kj}f(U_{p}) - iC_{kj}f(U_{n}) - C_{kj}f(V_{p}) 
   + C_{kj}f(V_{n})\Big] \\
& = &
\left\{\begin{array}{ll}
\sum_{j,k}\Big[f(B_{kj}U_{p} + C_{kj}V_{n}) - f(B_{kj}U_{n} + C_{kj}V_{p}) 
   \; + \\
\qquad  if(B_{kj}V_{p} + C_{kj}U_{p}) - if(B_{kj}V_{n} + C_{kj}U_{n})\Big] 
  & \mbox{if } B_{kj} \geq 0, C_{kj} \geq 0 \\
\sum_{j,k}\Big[f(-B_{kj}U_{n} + C_{kj}V_{n}) - f(-B_{kj}U_{p} + C_{kj}V_{p}) 
   \; + \\
\qquad  if(-B_{kj}V_{n} + C_{kj}U_{p}) - if(-B_{kj}V_{p} + C_{kj}U_{n})\Big] 
  & \mbox{if } B_{kj} < 0, C_{kj} \geq 0 \\
\sum_{j,k}\Big[f(B_{kj}U_{p} - C_{kj}V_{p}) - f(B_{kj}U_{n} - C_{kj}V_{n}) 
   \; + \\
\qquad  if(B_{kj}V_{p} - C_{kj}U_{n}) - if(B_{kj}V_{n} - C_{kj}U_{p})\Big] 
  & \mbox{if } B_{kj} \geq 0, C_{kj} < 0 \\
\sum_{j,k}\Big[f(-B_{kj}U_{n} - C_{kj}V_{p}) - f(-B_{kj}U_{p} - C_{kj}V_{n}) 
   \; + \\
\qquad  if(-B_{kj}V_{n} - C_{kj}U_{n}) - if(-B_{kj}V_{p} - C_{kj}U_{p})\Big] 
  & \mbox{if } B_{kj} < 0, C_{kj} < 0
\end{array}\right. \\
& = &
\end{array}$$
}

This concludes our description of the spaces of operators $\BL(H)
\hookleftarrow \SA(H) \hookleftarrow \Pos(H)$ on a
(finite-dimensional) Hilbert space $H$, as naturally self-dual
modules. Before we proceed to density operators $\DM(H)$ and effects
$\Ef(H)$ on $H$ we wish to explore and exploit the similarities
between these modules (over $\C$, $\R$, and
$\Rnn$) in terms of algebras of a monad.

\section{Categories of modules as algebras}\label{ModuleSec}

We recall that a semiring~\cite{Golan99} is like a ring but without an
additive inverse. Modules are vector spaces except that the scalars
need only be a ring, and not a field. Here we generalise further and
will also consider modules over a semiring. In fact we have already
done in the previous section, when we talked about positive operators
forming a module over the non-negative reals $\Rnn$. As we now proceed
more systematically, we shall see that such a module over a semiring
consists of a commutative monoid of vectors, with scalar
multiplication by elements of the semiring. It will be captured as
algebra of the multiset monad.

In this section we thus start with the standard description of
categories of modules, over a semiring $S$, as categories of algebras
of a monad, namely of the multiset monad $\Mlt_S$ associated with
$S$. We shall be especially interested in the examples $S = \Rnn, \R,
\C$ giving us a uniform description of the categories of modules in
which the spaces of operators $\Pos(H)$, $\SA(H)$, $\BL(H)$ on a
Hilbert space $H$ live. The general theory of monads---see
\textit{e.g.}~\cite{MacLane71,BarrW85,Manes74,Borceux94}---gives us
certain structure for free, see Theorem~\ref{MonadAlgStructThm} below.

The main result in this section, Theorem~\ref{AlgCatAdjThm}, relates
the three spaces of operators $\Pos(H)$, $\SA(H)$, $\BL(H)$ via free
constructions between categories of modules.

To start, let $S$ be a semiring, consisting of a commutative additive
monoid $(S,+,0)$ and a multiplicative monoid $(S,\cdot,1)$, where
multiplication distributes over addition. One can define a
``multiset'' functor $\Mlt_{S}\colon\Sets\rightarrow\Sets$ by:
$$\begin{array}{rcl}
\Mlt_{S}(X)
& = &
\set{\varphi\colon X\rightarrow S}{\support(\varphi)\mbox{ is finite}},
\end{array}$$

\noindent where $\support(\varphi) = \setin{x}{X}{\varphi(x) \neq 0}$
is the support of $\varphi$. For a function $f\colon X\rightarrow Y$
one defines $\Mlt_{S}(f) \colon \Mlt_{S}(X) \rightarrow \Mlt_{S}(Y)$ by:
\begin{equation}
\label{MltEqn}
\begin{array}{rcl}
\Mlt_{S}(f)(\varphi)(y)
& = &
\sum_{x\in f^{-1}(y)}\varphi(x).
\end{array}
\end{equation}

\noindent Such a (finite) multiset $\varphi\in \Mlt_{s}(X)$ may be
written as formal sum $s_{1}\ket{x_{1}}+\cdots+s_{k}\ket{x_{k}}$ where
$\support(\varphi) = \{x_{1}, \ldots, x_{k}\}$ and $s_{i} =
\varphi(x_{i})\in S$ describes the ``multiplicity'' of the element
$x_{i}$. The ket notation $\ket{x_i}$ is justified because these
elements are vectors, and useful, because it distinguishes $x$ as
element of $X$ and as vector in $\Mlt_{S}(X)$. These formal sum are
quotiented by the usual commutativity and associativity
relations. Also, the same element $x\in X$ may be counted multiple
times, so that $s_{1}\ket{x} + s_{2}\ket{x}$ is considered to be the
same as $(s_{1}+s_{2})\ket{x}$. With this formal sum notation one can
write the application of $\Mlt_{S}$ on a map $f$ as
$\Mlt_{S}(f)(\sum_{i}s_{i}\ket{x_{i}}) = \sum_{i}s_{i}\ket{f(x_{i})}$.

This multiset functor is a monad, whose unit $\eta\colon X\rightarrow
\Mlt_{S}(X)$ is $\eta(x) = 1\ket{x}$, and multiplication $\mu\colon
\Mlt_{S}(\Mlt_{S}(X)) \rightarrow \Mlt_{S}(X)$ is
$\mu(\sum_{i}s_{i}\ket{\varphi_{i}})(x) =
\sum_{i}s_{i}\cdot\varphi_{i}(x)$, where $\cdot$ is multiplication in
$S$. 

In order to emphasise that elements of $\Mlt_{S}(X)$ are \emph{finite}
multisets, one may call $\Mlt_{S}$ the \emph{finitary} multiset monad.
In order to include non-finite multisets, one has to assume that
suitable infinite sums exist in the underlying semiring $S$. This is
less natural.

For the semiring $S=\NNO$ one gets the free commutative monoid
$\Mlt_{\NNO}(X)$ on a set $X$. The monad $\Mlt_{\NNO}$ is also known
as the `bag' monad, containing ordinary ($\NNO$-valued) multisets. If
$S=\mathbb{Z}$ one obtains the free Abelian group
$\Mlt_{\mathbb{Z}}(X)$ on $X$. The Boolean semiring $2 = \{0,1\}$
yields the finite powerset monad $\powersetfin = \Mlt_{2}$. Here we
shall be mostly interested in the cases where $S$ is $\Rnn$,
$\R$, or $\C$. 

An (Eilenberg-Moore) algebra $\alpha\colon\Mlt_{S}(X)\rightarrow X$
for the multiset monad corresponds to a monoid structure on
$X$---given by $x+y = \alpha(1\ket{x} + 1\ket{y})$---together with a
scalar multiplication $\bullet \colon S\times X\rightarrow X$ given by
$s\mathrel{\bullet} x = \alpha(s\ket{x})$. It preserves the additive
structure (of $S$ and of $X$) in each coordinate separately. This
makes $X$ a module, for the semiring $S$. Conversely, such an
$S$-module structure on a commutative monoid $M$ yields an algebra
$\Mlt_{S}(M)\rightarrow M$ by $\sum_{i}s_{i}\ket{x_{i}} \mapsto
\sum_{i}s_{i}\mathrel{\bullet}x_{i}$.  Thus the category of algebras
$\Alg(\Mlt_{S})$ is equivalent to the category $\Mod[S]$ of
$S$-modules. When $S$ happens to be a field, this category $\Mod[S]$
is the category $\Vect[S]$ of vector spaces over $S$. Thus we have a
uniform description of the three categories of relevance in the
previous section, namely:
$$\begin{array}{c}
\Alg(\Mlt_{\Rnn}) = \Mod[\Rnn] \\
\Alg(\Mlt_{\R}) = \Mod[\R] = \Vect[\R]
\qquad
\Alg(\Mlt_{\C}) = \Mod[\C] = \Vect[\C].
\end{array}$$

\noindent We continue this section with a basic result in the theory
of monads, which is stated without proof, but with a few subsequent
pointers.

\begin{theorem}
\label{MonadAlgStructThm}
Let $\cat{A}$ be a symmetric monoidal category, which is both complete
and cocomplete, and let $T\colon \cat{A} \rightarrow \cat{A}$ be a
monad on $\cat{A}$.  The category $\Alg(T)$ of algebras is:
\begin{enumerate}
\item[(a)] also complete, with limits as in $\cat{A}$;

\item[(b)] cocomplete as soon as certain special colimits exist in
$\Alg(T)$, namely colimits of reflexive pairs;

\item[(c)] symmetric monoidal closed in case these colimits exist and
  the monad $T$ is symmetric monoidal (commutative), where the free
  algebra functor $F\colon \cat{A} \rightarrow \Alg(T)$ preserves the
  monoidal structure (\textit{i.e.}~is strong monoidal). \QED
\end{enumerate}

\end{theorem}

\auxproof{
Let a monad $T$ on $\cat{C}$ be commutative via $\xi\colon T(X)
\otimes T(Y) \rightarrow T(X\otimes Y)$. Given three algebras
$\smash{TX\stackrel{a}{\rightarrow}X}$,
$\smash{TY\stackrel{b}{\rightarrow}Y}$,
$\smash{TZ\stackrel{c}{\rightarrow}Z}$, one calls a morphism
$f\colon X\otimes Y\rightarrow Z$ in $\cat{C}$ a bimorphism if the
following diagram commutes.
$$\xymatrix@R1.5pc{
T(X)\otimes T(Y)\ar[d]_{a\otimes b}\ar[r]^-{\xi} &
  T(X\otimes Y)\ar[r]^-{T(f)} & T(Z)\ar[d]^{c} \\
X\otimes Y\ar[rr]^-{f} & & Z
}$$

\noindent The aim is to obtain a tensor product $a\boxtimes b
= \smash{(TX\stackrel{a}{\rightarrow}X)} \boxtimes
\smash{(TY\stackrel{b}{\rightarrow}Y)}$ of algebras so that
algebra morphisms $a\boxtimes b\rightarrow c$ correspond
to such bimorphisms. This tensor product $a\boxtimes b$ arises
as coequaliser in the category $\Alg(T)$, of the form:
\begin{equation}
\label{AlgTensorEqn}
\xymatrix@C-.5pc{
\ensuremath{\left(\xy
(0,4)*{T^{2}(TX\otimes TY)};
(0,-4)*{T(TX\otimes TY)};
{\ar^{\mu} (0,2); (0,-2)};
\endxy\right)}\ar@<.5pc>[rr]^-{T(a\otimes b)}
\ar@<-.5pc>[rr]_-{\mu\after T(\xi)}
& &
\ensuremath{\left(\xy
(0,4)*{T^{2}(X\otimes Y)};
(0,-4)*{T(X\otimes Y)};
{\ar^{\mu} (0,2); (0,-2)};
\endxy\right)}
\ar[rr]^-{t} & & a\boxtimes b
}
\end{equation}


\noindent We abuse notation and write $X\boxtimes Y$ for the carrier
object of the algebra $a\boxtimes b$, which is thus a map
$T(X\boxtimes Y)\rightarrow X\boxtimes Y$. This coequaliser map
$t\colon T(X\otimes Y) \rightarrow X\boxtimes Y$ ensures that the map
$\sotimes \smash{\stackrel{\textrm{def}}{=}}\, t \after \eta \colon
X\otimes Y \rightarrow X\boxtimes Y$ is a bimorphism, and even in a
universal way: for an arbitrary bimorphism $f\colon X\otimes
Y\rightarrow Z$ there is a unique algebra morphism $\overline{f}
\colon a\boxtimes b \rightarrow c$ with $\overline{f} \after \,
\sotimes\, = f$. This $\sotimes$ forms a map $U(a)\otimes U(b)
\rightarrow U(a\boxtimes b)$, making the forgetful functor $U\colon
\Alg(T)\rightarrow \cat{C}$ monoidal. The free algebra $F(I)$ on the
unit $I\in\cat{C}$ is the unit for $\boxtimes$. The map $\xi\colon
T(X)\otimes T(Y) \rightarrow T(X\otimes Y)$ is by definition a
bimorphism between free algebras. Hence there is a unique map
$\xi^{F}\colon F(X)\boxtimes F(Y) \rightarrow F(X\otimes Y)$ with
$\xi^{F} \after \;\sotimes\; = \xi$. This map, makes $F$ monoidal; in
fact, this $\xi^F$ is an isomorphism, so that $F$ is strong
monoidal. \QED

\medskip

\textbf{More eleborate:}

\medskip

First, $\sotimes$ is a bimorphism, since:
$$\begin{array}{rcll}
(a\boxtimes b) \after T(\sotimes) \after \xi
& = &
(a\boxtimes b) \after T(c \after \eta) \after \xi \\
& = &
c \after \mu \after T(\eta) \after \xi 
   & \mbox{since $c$ is an algebra map} \\
& = &
c \after \mu \after \eta \after \xi \\
& = &
c \after \mu \after T(\xi) \after \eta \\
& = &
c \after T(a\otimes b) \after \eta 
   & \mbox{since $c$ is a coequaliser} \\
& = &
c \after \eta \after (a\otimes b) \\
& = &
\sotimes\, \after (a\otimes b).
\end{array}$$

\noindent Next, assume $f\colon X\otimes Y\rightarrow Z$ is a
bimorphism, so that $c \after T(f) \after \xi = f \after
(a\otimes b)$. The map $f^{\#} = c \after T(f) \colon
T(X\otimes Y) \rightarrow Z$ is then an algebra map
$\mu\rightarrow c$ coequalising the above reflexive pair:
$$\begin{array}{rcl}
c \after T(f^{\#})
& = &
c \after T(c \after T(f)) \\
& = &
c \after \mu \after T^{2}(f) \\
& = &
c \after T(f) \after \mu \\
& = &
f^{\#} \after \mu \\
f^{\#} \after \mu \after T(\xi) 
& = &
c \after T(f) \after \mu \after T(\xi) \\
& = &
c \after \mu \after T^{2}(f) \after T(\xi) \\
& = &
c \after T(c) \after T^{2}(f) \after T(\xi) \\
& = &
c \after T(c \after T(f) \after \xi) \\
& = &
c \after T(f \after T(a\otimes b)) \\
& = &
f^{\#} \after T(a\otimes b).
\end{array}$$

\noindent Hence there is a unique algebra homomorphism 
$\overline{f} \colon (a\boxtimes b) \rightarrow c$
with $\overline{f} \after t = f^{\#}$. Then: 
$$\begin{array}{rcl}
\overline{f} \after \; \sotimes
& = &
\overline{f} \after t \after \eta  \\
& = &
f^{\#} \after \eta \\
& = &
c \after T(f) \after \eta \\
& = &
c \after \eta \after f \\
& = &
f.
\end{array}$$

\noindent If also $g\colon (a\boxtimes b) \rightarrow c$
satisfies $g\after \; \sotimes\, = f$ then:
$$\begin{array}{rcl}
f^{\#}
& = &
c \after T(f) \\
& = &
c \after T(g \after \; \sotimes) \\
& = &
g \after (a\boxtimes b) \after T(t \after \eta) \\
& = &
g \after t \after \mu \after T(\eta) \\
& = &
g \after t.
\end{array}$$

\noindent Hence $g = \overline{f}$.

We show how to obtain $\rho^{\boxtimes} \colon a\boxtimes
F(I)\conglongrightarrow a$ in $\Alg(T)$, for $a\colon T(X)\rightarrow
X$ and $\smash{F(I) = (T^{2}(I) \stackrel{\mu}{\rightarrow} T(I))}$.
We first prove a slightly stronger property: for each algebra $\beta
\colon T(Y)\rightarrow Y$ and algebra map $f\colon X\rightarrow Y$
there is a unique map of algebras $\widehat{f} \colon a\boxtimes F(I)
\rightarrow b$ with:
$$\begin{array}{rcl}
f 
& = &
\widehat{f} \after \;\sotimes\; \after (\idmap{}\otimes\eta) \after \rho^{-1}
   \;\colon\; X\rightarrow X\otimes I\rightarrow X\otimes T(I)
   \rightarrow X\boxtimes T(I)\rightarrow Y.
\end{array}\eqno{(*)}$$

\noindent We then obtain $\rho^{\boxtimes}$ by taking $f=\idmap{}$.
Given such an algebra map $f\colon a\rightarrow b$ we define:
$$\xymatrix@C-.5pc{
f' \stackrel{\textrm{def}}{=} \Big(
   X\otimes T(I)\ar[r]^-{\eta\otimes\idmap{}} &
   T(X)\otimes T(I)\ar[r]^-{\xi} & T(X\otimes I)\ar[r]^-{T(\rho)}_-{\cong} &
   T(X)\ar[r]^-{a} & X\ar[r]^-{f} & Y\Big)
}$$

\noindent We first check that $f'$ is a bimorphism.
$$\begin{array}{rcl}
f' \after (a\otimes\mu)
& = &
f \after a \after T(\rho) \after \xi \after (\eta\otimes\idmap{}) 
   \after (a\otimes\mu) \\
& = &
f \after a \after T(\rho) \after \xi \after (T(a)\otimes\idmap{}) \after 
   (\eta\otimes\mu) \\
& = &
f \after a \after T(\rho) \after T(a\otimes\idmap{}) \after \xi \after 
   (\eta\otimes\mu) \\
& = &
f \after a \after T(a) \after T(\rho) \after \xi \after (\eta\otimes\mu) \\
& = &
f \after a \after \mu \after T(\rho) \after \xi \after (\eta\otimes\mu) \\
& \smash{\stackrel{(*)}{=}} &
f \after a \after \mu \after T^{2}(\rho) \after T(\xi) \after 
   T(\idmap{}\otimes\eta) \after \xi \after (\eta\otimes\mu) \\
& = &
f \after a \after T(\rho) \after \mu \after T(\xi) \after \xi \after
   (\idmap{}\otimes T(\eta)) \after (\eta\otimes\mu) \\
& = &
f \after a \after T(\rho) \after \xi \after (\mu\otimes\mu) \after
   (\idmap{}\otimes T(\eta)) \after (\eta\otimes\mu) \\
& = &
f \after a \after T(\rho) \after \xi \after (\idmap{}\otimes\mu) \\
& = &
f \after a \after T(\rho) \after \xi \after (\mu\otimes\mu) \after 
   (T(\eta)\otimes\idmap{})\\
& = &
f \after a \after T(\rho) \after \mu \after T(\xi) \after \xi \after 
   (T(\eta)\otimes\idmap{})\\
& = &
f \after a \after \mu \after T^{2}(\rho) \after T(\xi) \after 
   T(\eta\otimes\idmap{}) \after \xi \\
& = &
b \after T(f) \after T(a) \after T^{2}(\rho) \after T(\xi) \after 
   T(\eta\otimes\idmap{}) \after \xi \\
& = &
b \after T(f') \after \xi.
\end{array}$$

\noindent where the marked equation holds because:
$$\xymatrix{
T(X)\otimes I\ar[drr]_{\rho}\ar[r]^-{\idmap{}\otimes\eta} &
   T(X)\otimes T(I)\ar[r]^-{\xi} & T(X\otimes I)\ar[d]^{T(\rho)} \\
& & T(X)
}$$

\noindent We then get a unique algebra map $\widehat{f} \colon
a\boxtimes \mu_{I} \rightarrow b$ with $\widehat{f} \after \,\sotimes\,
= f'$. Further,
$$\begin{array}{rcl}
\widehat{f} \after \,\sotimes\, \after (\idmap{}\otimes\eta_{I}) \after \rho^{-1}
& = &
f' \after (\idmap{}\otimes\eta_{I}) \after \rho^{-1}\\
& = &
f \after a \after T(\rho) \after \xi \after
   (\eta\otimes\idmap{}) \after (\idmap{}\otimes\eta_{I}) \after \rho^{-1}\\
& = &
f \after a \after T(\rho) \after \eta \after \rho^{-1}\\
& = &
f \after a \after \eta \after \rho \after \rho^{-1}\\
& = &
f.
\end{array}$$

\noindent In fact, $\widehat{f}$ is the unique algebra map
with this property: if $g\colon a\boxtimes \mu_{I} \rightarrow b$
also satisfies $g \after \,\sotimes\, \after (\idmap{}\otimes\eta_{I}) 
\after \rho^{-1} = f$, then:
$$\begin{array}{rcl}
f'
& = &
f \after a \after T(\rho) \after \xi \after (\eta\otimes\idmap{}) \\
& = &
b \after T(f) \after T(\rho) \after \xi \after (\eta\otimes\idmap{}) \\
& = &
b \after T(g \after \;\sotimes\; \after (\idmap{}\otimes\eta_{I}) \after 
   \rho^{-1} \after \rho) \after \xi \after (\eta\otimes\idmap{}) \\
& = &
g \after (a\boxtimes\mu_{I}) \after T(\sotimes) \after \xi \after 
   (\idmap{}\otimes T(\eta_{I})) \after (\eta\otimes\idmap{}) \\
& = &
g \after \;\sotimes\; \after (a\otimes\mu) \after (\eta\otimes T(\eta)) 
   \qquad\mbox{since $\sotimes$ is a bimorphism} \\
& = &
g \after \,\sotimes.
\end{array}$$

\noindent Hence $g=\widehat{f}$.

We briefly check that $\rho^{\boxtimes} = \widehat{\idmap{}} \colon
X\boxtimes T(I)\rightarrow X$ is an isomorphism. The candidate inverse
is:
$$\xymatrix{
\sigma \stackrel{\textrm{def}}{=} \Big(X\ar[r]^-{\rho^{-1}} &
   X\otimes I\ar[r]^-{\idmap{}\otimes\eta} &
   X\otimes T(I)\ar[r]^-{\sotimes} & X\boxtimes T(I)\Big).
}$$

\noindent In one direction, by $(*)$ we immediately get:
$$\begin{array}{rcl}
\rho^{\boxtimes} \after \sigma
& = &
\widehat{\idmap{}} \after \;\sotimes\; \after (\idmap{}\otimes\eta) \after
   \rho^{-1} \\
& = &
\idmap{}
\end{array}$$

\noindent In the other direction we first need to see that $\sigma$
is an algebra map $a\rightarrow a\boxtimes F(I)$, in:
$$\begin{array}{rcll}
\lefteqn{(a\boxtimes F(I)) \after T(\sigma)} \\
& = &
(a\boxtimes F(I)) \after T(\sotimes) \after T(\idmap{}\otimes\eta)
   \after T(\rho^{-1}) \\
& = &
(a\boxtimes F(I)) \after T(\sotimes) \after T(\idmap{}\otimes\eta)
   \after \xi \after (\idmap{}\otimes\eta) \after \rho^{-1}
   & \mbox{see the previous triangle} \\
& = &
(a\boxtimes F(I)) \after T(\sotimes) \after \xi \after (\idmap{}\otimes T(\eta))
   (\idmap{}\otimes\eta) \after \rho^{-1} \\
& = &
\sotimes\; \after (a\otimes\mu) \after (\idmap{}\otimes T(\eta))
   (\idmap{}\otimes\eta) \after \rho^{-1}
   & \mbox{since $\sotimes$ is a bimorphism} \\
& = &
\sotimes\; \after (a\otimes\eta) \after \rho^{-1} \\
& = &
\sotimes\; \after (\idmap{}\otimes\eta) \after \rho^{-1} \after a \\
& = &
\sigma \after a.
\end{array}$$

\noindent Hence we obtain a unique $\widehat{\sigma} \colon a\boxtimes F(I)
\rightarrow a\boxtimes F(I)$ satisfying $(*)$, in:
$$\begin{array}{rcl}
\widehat{\sigma} \after \;\sotimes\; 
   \after (\idmap{}\otimes\eta) \after \rho^{-1}
& = &
\sigma \\
& = &
\sotimes\; \after (\idmap{}\otimes\eta) \after \rho^{-1}.
\end{array}$$

\noindent Hence $\widehat{\sigma} = \idmap{}$. But then also $\sigma\after
\rho^{\boxtimes} = \idmap{}$ since:
$$\begin{array}{rcl}
\sigma \after \rho^{\boxtimes} \after \;\sotimes\; \after (\idmap{}\otimes\eta) 
   \after \rho^{-1}
& = &
\sigma \after \idmap{} \\
& = &
\idmap{} \after \;\sotimes\; \after (\idmap{}\otimes\eta) \after \rho^{-1}.
\end{array}$$

In order to see that the free functor $F\colon \cat{C}\rightarrow 
\Alg(T)$ is monoidal we use the identity map $F(I)\rightarrow F(I)$
as $\zeta$. The map $\xi^{F}\colon F(X)\boxtimes F(Y)\rightarrow
F(X\otimes Y)$ is obtained in:
$$\xymatrix{
T(X)\otimes T(Y)\ar[rr]^-{\sotimes = t\after \eta}\ar[drr]_{\xi^T} & & 
   T(X)\boxtimes T(Y)\ar@{-->}[d]^{\xi^F} \\
& & T(X\otimes Y)
}$$

\noindent where $t$ is the coequaliser in $\Alg(T)$ of $\mu \after T(\xi),
T(\mu\otimes \mu) \colon T(T^{2}(X)\otimes T^{2}(Y)) \rightrightarrows
T(T(X)\otimes T(Y))$. The candidate inverse for $\xi^F$ is:
$$\xymatrix{
\sigma \stackrel{\textrm{def}}{=} \Big(T(X\otimes Y)\ar[r]^{T(\eta\otimes\eta)} &
   T(T(X)\otimes T(Y))\ar[r]^-{t} & T(X)\boxtimes T(Y)\Big).
}$$

\noindent Then:
$$\begin{array}{rcl}
\sigma \after \xi^{F} \after \;\sotimes\;
& = &
t \after T(\eta\otimes\eta) \after \xi^{T} \\
& = &
t \after \xi^{T} \after (T(\eta)\otimes T(\eta)) \\
& = &
t \after \xi^{T} \after (\mu\otimes\mu) \after (\eta\otimes\eta) \after 
   (T(\eta)\otimes T(\eta)) \\
& = &
t \after \mu \after T(\xi^{T}) \after \xi \after (\eta\otimes\eta) \after 
   (T(\eta)\otimes T(\eta)) \\
& = &
t \after T(\mu\otimes\mu) \after \xi \after (\eta\otimes\eta) \after 
   (T(\eta)\otimes T(\eta)) \quad \mbox{since $t$ is coequaliser} \\
& = &
t \after T(\mu\otimes\mu) \after \eta \after 
   (T(\eta)\otimes T(\eta)) \\
& = &
t \after \eta \after (\mu\otimes\mu) \after (T(\eta)\otimes T(\eta)) \\
& = &
t \after \eta \\
& = &
\;\sotimes \\
\xi^{F} \after \sigma
& = &
\xi^{F} \after t \after T(\eta\otimes\eta) \\
& = &
\mu \after T(\xi^{T}) \after T(\eta\otimes\eta) 
   \qquad \mbox{see the construction of maps like $\xi^F$} \\
& & \qquad \mbox{in the beginning of this ``auxproof''} \\
& = &
\mu \after T(\eta) \\
& = &
\idmap{}
\end{array}$$
}

A category of algebras is always ``as complete'' as its underlying
category, see \textit{e.g.}~\cite{Manes74,BarrW85}. Cocompleteness
always holds for algebras over \Sets and follows from a result of
Linton's, see~\cite[\S~9.3, Prop.~4]{BarrW85} using the existence of
coequalisers of reflexive pairs in $\Sets$. We shall mostly use this
result for $\cat{A} = \Sets$, so that we don't have to worry about
these special colimits; the monoidal structure on the underlying
category $\Sets$ is thus cartesian.  Monoidal structure $(I,\otimes)$
in categories of algebras goes back to~\cite{Kock71a} (see
also~\cite{Jacobs94a}). The tensor unit $I$ is simply $F(1)$, for the
free algebra functor $F\colon \Sets \rightarrow \Alg(T)$ and the final
(singleton) set $1$. The tensor $\otimes$ is obtained as a suitable
coequaliser of algebras. Algebra maps $X\otimes Y\rightarrow Z$ then
correspond to bi-homomorphisms $UX\times UY \rightarrow UZ$. In
particular, there is a universal bi-homomorphism $\sotimes \colon
UX\times UY \rightarrow U(X\otimes Y)$. The free functor preserves
these tensors.



The multiset monad $\Mlt_{S}$ is symmetric monoidal if $S$ is a
(multiplicatively) commutative semiring. In that case categories
$\Mod[S]$ are monoidal closed, with $S \cong \Mlt_{S}(1)$ as tensor
unit. Maps $M\otimes N\rightarrow K$ in $\Mod[S]$ correspond to
bilinear maps $M\times N\rightarrow K$ (linear in each argument
separately). The associated exponent is written as $\multimap$, like
before.

For modules $M,N\in \Mod[S]$ there are obvious correspondences:
$$\begin{prooftree}
\begin{prooftree}
{\xymatrix{ M\ar[r] & (N\multimap S)}} 
\Justifies
{\xymatrix{ M\otimes N\ar[r] & S}} 
\end{prooftree}
\Justifies
\begin{prooftree}
{\xymatrix{ N\otimes M\ar[r] & S}} 
\Justifies
{\xymatrix{ N\ar[r] & (M\multimap S)}} 
\end{prooftree}
\end{prooftree}$$

\noindent This means that there are adjunctions:
\begin{equation}
\label{ModuleDualDiag}
\vcenter{\xymatrix{
\Mod[S]\ar@/^2ex/[rr]^-{(-)\multimap S} & \bot & 
   \big(\Mod[S]\big)\rlap{$\op$}\ar@/^2ex/[ll]^-{(-)\multimap S}
}}
\end{equation}

\noindent as used in the previous section.

In summary, we have a sequence of categories of algebras of monads:
\begin{equation}
\label{AlgCatDiag}
\vcenter{\xymatrix@R-1.7pc{
\Alg(\Mlt_{\Rnn})\ar@{=}[d]
&
\Alg(\Mlt_{\R})\ar[l]\ar@{=}[d]
&
\Alg(\Mlt_{\C})\ar[l]\ar@{=}[d] \\
\Mod[\Rnn]
&
\Vect[\R]
&
\Vect[\C]
}}
\end{equation}

\noindent where the maps between them can be understood as arising
from maps of monads in the other direction:
$$\xymatrix@C-1pc{
\Mlt_{\Rnn}\ar@{=>}[r] &
\Mlt_{\R}\ar@{=>}[r] &
\Mlt_{\C}
& \mbox{via semiring inclusions} &
\Rnn\ar[r] &
\R\ar[r] &
\C.
}$$

\noindent This follows from the following general result.

\begin{proposition}
\label{SemiringMapProp}
A homomorphism of semirings $f\colon S\rightarrow S'$, preserving both
the additive and multiplicative monoid structures, gives rise to a map
of monads $\Mlt_{S} \Rightarrow \Mlt_{S'}$, by
$\big(\sum_{j}s_{j}\ket{x_{j}}\big) \mapsto
\big(\sum_{j}f(s_{j})\ket{x_{j}}\big)$, and thus to a functor
$\Alg(\Mlt_{S'}) \rightarrow \Alg(\Mlt_{S})$, by $\big(\Mlt_{S'}(X)
\rightarrow X\big) \longmapsto \big(\Mlt_{S}(X) \rightarrow
\Mlt_{S'}(X) \rightarrow X\big)$. This functor always has a left
adjoint. \QED
\end{proposition}

The left adjoint exists because categories of modules $\Alg(\Mlt_{S})
= \Mod[S]$ are cocomplete; it can be constructed via a coequaliser,
see~\textit{e.g.}~\cite{Jacobs94a,JacobsM12b}. Thus, modules over their
semirings have the structure of a bifibration~\cite{Jacobs99a}.

The different spaces of operators $\BL(H)$, $\SA(H)$, $\Pos(H)$ on a
Hilbert space $H$ turn out to be related via free constructions.  This
was used implicitly in the proofs of Propositions~\ref{SADualProp}
and~\ref{PosDualProp} in the previous section, and also
in~\cite{Busch03}.

\begin{theorem}
\label{AlgCatAdjThm}
Write the left adjoints to the two forgetful functors
in~(\ref{AlgCatDiag}) as:
\begin{equation}
\label{AlgCatAdjDiag}
\xymatrix{
\Mod[\Rnn]\ar[r]^-{\mathcal{R}} &
   \Vect[\R]\ar[r]^-{\mathcal{C}}&
   \Vect[\C].
}
\end{equation}

\noindent For a finite-dimensional Hilbert space $H$, the canonical
inclusion morphisms $\Pos(H) \hookrightarrow \SA(H)$ in
$\Mod[\Rnn]$, and $\SA(H) \hookrightarrow \BL(H)$
in $\Vect[\R]$ yield via these adjunctions (transposed)
maps that turn out to be isomorphisms:
$$\xymatrix{
\mathcal{R}\big(\Pos(H)\big)\ar[r]^-{\cong} & \SA(H)
& \mbox{and} &
\mathcal{C}\big(\SA(H)\big)\ar[r]^-{\cong} & \BL(H).
}$$

\noindent Thus we have the following situation of triangles commuting
up-to-isomorphism.
$$\xymatrix{
& \FdHilb\ar[dl]_{\Pos}\ar[d]^(0.6){\SA}\ar[dr]^{\BL} \\
\Mod[\Rnn]\ar[r]_-{\mathcal{R}}^-{\mbox{\small free}} &
   \Vect[\R]\ar[r]_-{\mathcal{C}}^-{\mbox{\small free}} &
   \Vect[\C].
}$$
\end{theorem}

\begin{proof}
The proof uses explicit constructions of the left adjoints
$\mathcal{R}$ and $\mathcal{C}$ in~(\ref{AlgCatAdjDiag}).  A module
$X$ over $\Rnn$ can be turned into a vector space over $\R$ via the
same construction that turns a commutative monoid into a commutative
group:
$$\begin{array}{rcl}
\mathcal{R}(X)
& = &
(X\times X)/\!\sim
\qquad\mbox{where}\quad 
\begin{array}[t]{rcl}
\lefteqn{(x_{1},x_{2}) \sim (y_{1}, y_{2})} \\
& \Longleftrightarrow &
\ex{z}{x_{1}+y_{2}+z = y_{1}+x_{2}+z}.
\end{array}
\end{array}$$

\noindent Addition is done componentwise: $[x_{1},x_{2}] +
          [y_{1},y_{2}] = [x_{1}+y_{1}, x_{2}+y_{2}]$, minus by
          reversal: $-[x_{1},x_{2}] = [x_{2},x_{1}]$, and scalar
          multiplication $\scalar \colon \R \times \mathcal{R}(X)
          \rightarrow \mathcal{R}(X)$ via:
$$\begin{array}{rcl}
r\scalar [x_{1},x_{2}]
& = &
\left\{\begin{array}{ll}
{[r\scalar x_{1}, r\scalar x_{2}]} & \mbox{if } r \geq 0 \\
{[(-r)\scalar x_{2}, (-r)\scalar x_{1}]} & \mbox{if } r < 0
\end{array}\right.
\end{array}$$

\noindent (Notice the reversal of the $x_i$ in the second case.)

\auxproof{ 
It is standard that $\mathcal{R}(X)$ is a commutative group (see
\textit{e.g.}~\cite[\S\S4.7]{Jacobs99a}), with $[0,0]$ as zero
element, $-[x_{1},x_{2}] = [x_{2},x_{1}]$ as minus, and inclusion
$\eta\colon X\rightarrow \mathcal{R}(X)$ given by $\eta(x) =
[x,0]$. We have to check that scalar multiplication behaves
appropriately.  This requires some care. First, it is well-defined: if
$(x_{1},x_{2}) \sim (y_{1},y_{2})$, then $x_{1}+y_{2}+z = y_{1}+x_{2}+z$,
for some $z\in X$, and thus if $r \geq 0$,
$$r\scalar x_{1} + r \scalar y_{2} + r\scalar z
=
r\scalar (x_{1}+y_{2}+z)
=
r\scalar (y_{1}+x_{2}+z)
=
r\scalar y_{1} + r \scalar x_{2} + r\scalar z,$$

\noindent and thus $(r\scalar x_{1}, r\scalar x_{2}) \sim (r\scalar y_{1},
r\scalar y_{2})$. Similarly, if $r<0$, 
$$\begin{array}{rcl}
(-r)\scalar x_{2} + (-r) \scalar y_{1} + (-r) \scalar z
& = &
(-r)\scalar (x_{2}+y_{1}+z) \\
& = &
(-r)\scalar (y_{1}+x_{2}+z) \\
& = &
(-r)\scalar (x_{1}+y_{2}+z) \\
& = &
(-r)\scalar (y_{2}+x_{1}+z) \\
& = &
(-r)\scalar y_{2} + (-r) \scalar x_{1} + (-r) \scalar z,
\end{array}$$

\noindent so that $((-r)\scalar x_{2}, (-r)\scalar x_{1}) \sim
((-r)\scalar y_{2}, (-r)\scalar y_{1})$

Next, scalar multiplication preserves sums of vectors:
$$\begin{array}{rcl}
\lefteqn{r \scalar \big([x_{1},x_{2}] + [y_{1},y_{2}]\big)} \\
& = &
r \scalar [x_{1}+y_{1}, x_{2}+y_{2}] \\
& = &
\left\{\begin{array}{ll}
[r \scalar (x_{1}+y_{1}), r \scalar (x_{2}+y_{2})] 
   & \mbox{if } r \geq 0 \\
{[(-r) \scalar (x_{2}+y_{2}), (-r) \scalar (x_{1}+y_{1})]}
   & \mbox{if } r < 0 
\end{array}\right. \\
& = &
\left\{\begin{array}{ll}
[r \scalar x_{1}+ r \scalar y_{1}, r \scalar x_{2} + r \scalar y_{2}] 
   & \mbox{if } r \geq 0 \\
{[(-r) \scalar x_{2} + (-r) \scalar y_{2}, (-r) \scalar x_{1} + 
   (-r) \scalar y_{1}]}
   & \mbox{if } r < 0 
\end{array}\right. \\
& = &
\left\{\begin{array}{ll}
[r \scalar x_{1}, r \scalar x_{2}] + [r \scalar y_{1}, r \scalar y_{2}] 
   & \mbox{if } r \geq 0 \\
{[(-r) \scalar x_{2}, (-r) \scalar x_{1}] + [(-r) \scalar y_{2}, 
   (-r) \scalar y_{1}]}
   & \mbox{if } r < 0 
\end{array}\right. \\
& = &
\left\{\begin{array}{ll}
r \scalar [x_{1}, x_{2}] + r \scalar [y_{1}, y_{2}] 
   & \mbox{if } r \geq 0 \\
r \scalar [x_{1}, x_{2}] + r \scalar [y_{1}, y_{2}] 
   & \mbox{if } r < 0 
\end{array}\right. \\
& = &
r \scalar [x_{1}, x_{2}] + r \scalar [y_{1}, y_{2}].
\end{array}$$

\noindent Showing that $(r+s)\scalar [x_{1}, x_{2}] = r\scalar
          [x_{1},x_{2}] + s\scalar [x_{1},x_{2}]$ requires some care.
This is easy if both $r,s\geq 0$ or $r,s < 0$. We shall do the
case where $r<0$ and $s\geq 0$, and distinguish:
\begin{itemize}
\item $s+r \geq 0$. Then:
$$\begin{array}{rcl}
r\scalar [x_{1},x_{2}] + s\scalar [x_{1},x_{2}]
& = &
[(-r)\scalar x_{2}, (-r)\scalar x_{1}] + [s\scalar x_{1}, s\scalar x_{2}] \\
& = &
[(-r)\scalar x_{2} + s\scalar x_{1}, (-r)\scalar x_{1} + s\scalar x_{2}] \\
& \smash{\stackrel{(*)}{=}} &
[(r+s)\scalar x_{1}, (r+s)\scalar x_{2}] \\
& = &
(r+s)\scalar [x_{1}, x_{2}]
\end{array}$$

\noindent where $\smash{\stackrel{(*)}{=}}$ holds since:
$$\begin{array}{rcl}
\lefteqn{\Big((-r)\scalar x_{2} + s\scalar x_{1}\Big) +
   \Big((r+s)\scalar x_{2}\Big)} \\
& = &
s \scalar x_{1} + \big((r+s) - r\big) \scalar x_{2} \\
& = &
s\scalar x_{1} + s\scalar x_{2} \\
& = &
\big((r+s)-r\big)\scalar x_{1} + s\scalar x_{2} \\
& = &
\Big((r+s)\scalar x_{1}) + \Big((-r)\scalar x_{1} + s\scalar x_{2}\Big)
\end{array}$$

\item $s+r < 0$. Then:
$$\begin{array}{rcl}
r\scalar [x_{1},x_{2}] + s\scalar [x_{1},x_{2}]
& = &
[(-r)\scalar x_{2}, (-r)\scalar x_{1}] + [s\scalar x_{1}, s\scalar x_{2}] \\
& = &
[(-r)\scalar x_{2} + s\scalar x_{1}, (-r)\scalar x_{1} + s\scalar x_{2}] \\
& \smash{\stackrel{(*)}{=}} &
[(-(r+s))\scalar x_{2}, (-(r+s))\scalar x_{1}] \\
& = &
(r+s)\scalar [x_{1}, x_{2}]
\end{array}$$

\noindent where $\smash{\stackrel{(*)}{=}}$ holds since:
$$\begin{array}{rcl}
\lefteqn{\Big((-r)\scalar x_{2} + s\scalar x_{1}\Big) + 
   \Big((-(r+s))\scalar x_{1}\Big)} \\
& = &
(-r)\scalar x_{2} + (s - (r+s))\scalar x_{1} \\
& = &
(-r) \scalar x_{2} + (-r)\scalar x_{1} \\
& = &
(-(r+s) + s)\scalar x_{2} + (-r) \scalar x_{1} \\
& = &
\Big((-(r+s))\scalar x_{2}\Big) + \Big((-r)\scalar x_{1} + s\scalar x_{2}\Big).
\end{array}$$
\end{itemize}

We also have to check $r\scalar (s\scalar [x_{1},x_{2}]) = (r\cdot s)
\scalar [x_{1}, x_{2}]$. When $r,s\geq 0$ this trivially holds. Further,
\begin{itemize}
\item If $r,s<0$, then:
$$\begin{array}{rcl}
r\scalar (s\scalar [x_{1},x_{2}])
& = &
r\scalar [(-s)\scalar x_{2}, (-s)\scalar x_{1}] \\
& = &
[(-r)\scalar ((-s)\scalar x_{1}), (-r)\scalar ((-s)\scalar x_{2})] \\
& = &
[((-r)\cdot (-s))\scalar x_{1}, ((-r)\cdot (-s))\scalar x_{2}] \\
& = &
[(r\cdot s)\scalar x_{1}, (r\cdot s)\scalar x_{2}] \\
& = &
(r\cdot s) \scalar [x_{1}, x_{2}].
\end{array}$$

\item If $r<0$ and $s\geq 0$, then:
$$\begin{array}{rcl}
r\scalar (s\scalar [x_{1},x_{2}])
& = &
r\scalar [s\scalar x_{1}, s\scalar x_{2}] \\
& = &
[(-r)\scalar (s\scalar x_{2}), (-r)\scalar (s\scalar x_{1})] \\
& = &
[((-r)\cdot s)\scalar x_{2}, ((-r)\cdot s)\scalar x_{1}] \\
& = &
[(-(r\cdot s))\scalar x_{2}, (-(r\cdot s))\scalar x_{1}] \\
& = &
(r\cdot s) \scalar [x_{1}, x_{2}].
\end{array}$$
\end{itemize}

Finally we check the bijective correspondence:
$$\begin{prooftree}
{\xymatrix{ \mathcal{R}(X)\ar[r]^-{f} & Y 
   \rlap{\qquad in $\Vect[\R]$}}}
\Justifies
{\xymatrix{ X\ar[r]_-{g} & Y 
   \rlap{\qquad\quad in $\Mod[\Rnn]$}}}
\end{prooftree}$$

\noindent Given $f$ we take $\widehat{f}(x) = f(\eta(x)) = f([x,0])$.
It obviously preserves the monoid structure, and also scalar
multiplication since for $r\in\Rnn$,
$$\widehat{f}(r\scalar x)
=
f([r\scalar x, 0])
=
f([r\scalar x, r\scalar 0])
=
f(r\scalar [x, 0])
=
r\scalar f([x, 0])
=
r \scalar \widehat{f}(x).$$

\noindent Conversely, given $g$ take $\widehat{g}([x_{1},x_{2}]) =
g(x_{1}) - g(x_{2})$. It is well-defined since if $(x_{1},x_{2}) \sim
(y_{1},y_{2})$ then $x_{1}+y_{2} = y_{1} + x_{2}$, and so $g(x_{1}) +
g(y_{2}) = g(x_{1}+ y_{2}) = g(y_{1} + x_{2}) = g(y_{1}) + g(x_{2})$.
Hence $g(x_{1}) - g(x_{2}) = g(y_{1}) - g(y_{2})$. Further,
$$\begin{array}{rcl}
\widehat{g}(0)
& = &
\widehat{g}([0,0]) \\
& = &
g(0) - g(0) \\
& = &
0 \\
\widehat{g}([x_{1},x_{2}] + [y_{1},y_{2}])
& = &
\widehat{g}([x_{1} + y_{1}, x_{2} + y_{2}]) \\
& = &
g(x_{1}+y_{1}) - g(x_{2} + y_{2}) \\
& = &
g(x_{1}) + g(y_{1}) - g(x_{2}) - g(y_{2}) \\
& = &
g(x_{1}) - g(x_{2}) + g(y_{1}) - g(y_{2}) \\
& = &
\widehat{g}([x_{1},x_{2}]) + \widehat{g}([y_{1},y_{2}]) \\
\widehat{g}(r\scalar [x_{1},x_{2}])
& = &
\left\{\begin{array}{ll}
\widehat{g}([r \scalar x_{1}, r\scalar x_{2}])
   & \mbox{if }r \geq 0 \\
\widehat{g}([(-r) \scalar x_{2}, (-r)\scalar x_{1}])
   & \mbox{if }r < 0
\end{array}\right. \\
& = &
\left\{\begin{array}{ll}
g(r \scalar x_{1}) - g(r\scalar x_{2})
   & \mbox{if }r \geq 0 \\
g((-r) \scalar x_{2}) - g((-r)\scalar x_{1})
   & \mbox{if }r < 0
\end{array}\right. \\
& = &
\left\{\begin{array}{ll}
r \scalar g(x_{1}) - r \scalar g(x_{2})
   & \mbox{if }r \geq 0 \\
(-r) \scalar g(x_{2}) - (-r) \scalar g(x_{1})
   & \mbox{if }r < 0
\end{array}\right. \\
& = &
\left\{\begin{array}{ll}
r \scalar \big(g(x_{1}) - g(x_{2})\big)
   & \mbox{if }r \geq 0 \\
(-r) \scalar \big(g(x_{2}) - g(x_{1})\big)
   & \mbox{if }r < 0
\end{array}\right. \\
& = &
\left\{\begin{array}{ll}
r \scalar \widehat{g}([x_{1},x_{2}]).
   & \mbox{if }r \geq 0 \\
r \scalar \big(g(x_{1}) - g(x_{2})\big)
   & \mbox{if }r < 0
\end{array}\right. \\
& = &
r \scalar \widehat{g}([x_{1},x_{2}]).
\end{array}$$

Finally,
$$\begin{array}{rcl}
\widehat{\widehat{f}}([x_{1},x_{2}])
& = &
\widehat{f}(x_{1}) - \widehat{f}(x_{2}) \\
& = &
f([x_{1}, 0]) - f([x_{2}, 0]) \\
& = &
f([x_{1}, 0]) - [x_{2}, 0]) \\
& = &
f([x_{1}, 0]) + [0, x_{2}]) \\
& = &
f([x_{1}, x_{2}]) \\
\widehat{\widehat{g}}(x)
& = &
\widehat{g}([x, 0]) \\
& = &
g(x) - g(0) \\
& = &
g(x).
\end{array}$$

For completeness we add that on a map $f\colon X\rightarrow Y$ in
$\Mod[\Rnn]$ we get $\mathcal{R}(f) \colon
\mathcal{R}(X) \rightarrow \mathcal{R}(Y)$ in
$\Vect[\R]$ by $\mathcal{R}(f)([x_{1},x_{2}]) =
        [f(x_{1}), f(x_{2})]$. We check preservation of scalar
multiplication. For $r\in\R$,
$$\begin{array}{rcl}
\mathcal{R}(f)(r \scalar[x_{1},x_{2}])
& = &
\left\{\begin{array}{ll}
\mathcal{R}(f)([r \scalar x_{1}, r \scalar x_{2}])
   & \mbox{if }r \geq 0 \\
\mathcal{R}(f)([(-r) \scalar x_{2}, (-r) \scalar x_{1}])
   & \mbox{if }r < 0
\end{array}\right. \\
& = &
\left\{\begin{array}{ll}
[f(r \scalar x_{1}), f(r \scalar x_{2})]
   & \mbox{if }r \geq 0 \\
{[f((-r) \scalar x_{2}), f((-r) \scalar x_{1})]}
   & \mbox{if }r < 0
\end{array}\right. \\
& = &
\left\{\begin{array}{ll}
[r \scalar f(x_{1}), r \scalar f(x_{2})]
   & \mbox{if }r \geq 0 \\
{[(-r) \scalar f(x_{2}), (-r) \scalar f(x_{1})]}
   & \mbox{if }r < 0
\end{array}\right. \\
& = &
r \scalar [f(x_{1}), f(x_{2})] \\
& = &
r \scalar \mathcal{R}(f)([x_{1},x_{2}]).
\end{array}$$
}

A vector space $X$ over $\R$ can be turned into a vector space
over $\C$, simply via $\mathcal{C}(X) = X\times X$. The
additive structure is obtained pointwise, and scalar multiplication
$\scalar \colon \C\times \mathcal{C}(X) \rightarrow
\mathcal{C}(X)$ is done as follows.
$$\begin{array}{rcl}
(a+ib)\scalar (x_{1},x_{2})
& = &
(a\scalar x_{1} - b\scalar x_{2}, b\scalar x_{1} + a\scalar x_{2}).
\end{array}$$

\auxproof{
We first check that scalar multiplication is an action:
$$\begin{array}{rcl}
\lefteqn{1\scalar (x_{1},x_{2})} \\
& = &
(1+i0) \scalar (x_{1}, x_{2}) \\
& = &
(1\scalar x_{1} - 0\scalar x_{2}, 0\scalar x_{1} + 1\scalar x_{2}) \\
& = &
(x_{1}, x_{2}) \\
\lefteqn{(c+id) \scalar \big((a+ib) \scalar (x_{1}, x_{2})\big)} \\
& = &
(c+id) \scalar (a\scalar x_{1} - b\scalar x_{2}, 
   b\scalar x_{1} + a\scalar x_{2}) \\
& = &
(c\scalar (a\scalar x_{1} - b\scalar x_{2}) - 
   d\scalar (b\scalar x_{1} + a\scalar x_{2}), \\
& & \qquad
   d\scalar (a\scalar x_{1} - b\scalar x_{2}) + 
      c\scalar (b\scalar x_{1} + a\scalar x_{2})) \\
& = &
((c\cdot a)\scalar x_{1} - (c\cdot b)\scalar x_{2} - 
   (d\cdot b)\scalar x_{1} - (d\cdot a)\scalar x_{2}, \\
& & \qquad
   (d\cdot a)\scalar x_{1} - (d\cdot b)\scalar x_{2} + 
      (c\cdot b)\scalar x_{1} + (c\cdot a)\scalar x_{2}) \\
& = &
((c\cdot a - d\cdot b) \scalar x_{1} - (c\cdot b + d\cdot a) \scalar x_{2}, \\
& & \qquad
   (c\cdot b + d\cdot a) \scalar x_{1} + (c\cdot a - d\cdot b) \scalar x_{2}) \\
& = &
((c\cdot a - d\cdot b) + i(c\cdot b + d\cdot a)) \scalar (x_{1}, x_{2}) \\
& = &
((c + id) \cdot (a + ib)) \scalar (x_{1}, x_{2}).
\end{array}$$

\noindent Scalar multiplication on $\mathcal{C}(X)$ is bilinear:
$$\begin{array}{rcl}
\lefteqn{((a+ib) + (c + id)) \scalar (x_{1}, x_{2})} \\
& = &
((a+c) + i(b+d)) \scalar (x_{1}, x_{2}) \\
& = &
((a+c)\scalar x_{1} - (b+d)\scalar x_{2}, 
   (b+d)\scalar x_{1} + (a+c)\scalar x_{2}) \\
& = &
(a\scalar x_{1} + c\scalar x_{1} - b\scalar x_{2} - d\scalar x_{2},
   b\scalar x_{1} + d\scalar x_{1} + a\scalar x_{2} + c\scalar x_{2}) \\
& = &
(a\scalar x_{1} - b\scalar x_{2}, b\scalar x_{1} + a\scalar x_{2}) 
   + (c\scalar x_{1} - d\scalar x_{2}, d\scalar x_{1} + c\scalar x_{2}) \\
& = &
(a+ib) \scalar (x_{1}, x_{2}) + (c+id) \scalar (x_{1}, x_{2}) \\
\lefteqn{(a+ib) \scalar ((x_{1},x_{2}) + (y_{1}, y_{2}))} \\
& = &
(a+ib) \scalar (x_{1}+y_{1}, x_{2}+y_{2}) \\
& = &
(a\scalar (x_{1}+y_{1}) - b\scalar (x_{2}+y_{2}), 
   b\scalar (x_{1}+y_{1}) + a\scalar (x_{2}+y_{2})) \\
& = &
(a\scalar x_{1} + a\scalar y_{1} - b\scalar x_{2} - b\scalar y_{2}, 
   b\scalar x_{1} + b\scalar y_{1} + a\scalar x_{2} + a\scalar y_{2}) \\
& = &
(a\scalar x_{1} - b\scalar x_{2}, b\scalar x_{1} + a\scalar x_{2}) +
   (a\scalar y_{1} - b\scalar y_{2}, b\scalar y_{1} + a\scalar y_{2}) \\
& = &
(a+ib) \scalar (x_{1},x_{2}) + (a+ib) \scalar (y_{1},y_{2}).
\end{array}$$

Finally we check the bijective correspondence:
$$\begin{prooftree}
{\xymatrix{ 
\mathcal{C}(X)\ar[r]^-{f} & Y 
   \rlap{\qquad in $\Vect[\C]$}}}
\Justifies
{\xymatrix{ X\ar[r]_-{g} & Y 
   \rlap{\qquad\quad in $\Vect[\R]$}}}
\end{prooftree}$$

\noindent Given $f$ take $\widehat{f}(x) = f(x,0)$. This $\widehat{f}$
preserves scalar multiplication:
$$\begin{array}{rcl}
\widehat{f}(r \scalar x)
& = &
f(r\scalar x, 0) \\
& = &
f(r\scalar x - 0 \scalar 0, 0 \scalar x + r \scalar 0) \\
& = &
f((r+i0) \scalar (x, 0)) \\
& = &
(r+i0) \scalar f(x, 0) \\
& = &
r \scalar \widehat{f}(x).
\end{array}$$

\noindent In the other direction, given $g$ take $\widehat{g}(x_{1}, x_{2})
= g(x_{1}) + i\scalar g(x_{2})$. This $\widehat{g}$ preserves the
additive structure and scalar multiplication:
$$\begin{array}{rcl}
\widehat{g}((x_{1},x_{2}) + (y_{1}, y_{2}))
& = &
\widehat{g}(x_{1}+y_{1}, x_{2}+y_{2}) \\
& = &
g(x_{1}+y_{1}) + i\scalar g(x_{2} + y_{2}) \\
& = &
g(x_{1}) + g(y_{1}) + i\scalar g(x_{2}) + i\scalar g(y_{2}) \\
& = &
g(x_{1}) + i\scalar g(x_{2}) + g(y_{1}) + i\scalar g(y_{2}) \\
& = &
\widehat{g}(x_{1},x_{2}) + \widehat{g}(y_{1},y_{2}) \\
\widehat{g}((a+ib) \scalar (x_{1},x_{2}))
& = &
\widehat{g}(a\scalar x_{1} - b\scalar x_{2}, b\scalar x_{1} + a\scalar x_{2}) \\
& = &
g(a\scalar x_{1} - b\scalar x_{2}) + 
   i\scalar g(b\scalar x_{1} + a\scalar x_{2}) \\
& = &
a\scalar g(x_{1}) - b\scalar g(x_{2}) + 
   i\scalar (b\scalar g(x_{1})) + i\scalar (a\scalar g(x_{2})) \\
& = &
(a + ib) \scalar g(x_{1}) + (a+ib) \scalar (i \scalar g(x_{2})) \\
& = &
(a + ib) \scalar \widehat{g}(x_{1}, x_{2}).
\end{array}$$

Finally we have:
$$\begin{array}{rcl}
\widehat{\widehat{f}}(x_{1}, x_{2})
& = &
\widehat{f}(x_{1}) + i\scalar \widehat{f}(x_{2}) \\
& = &
f(x_{1}, 0) + i \scalar f(x_{2}, 0) \\
& = &
f(x_{1}, 0) + i \scalar (x_{2}, 0)) \\
& = &
f((x_{1}, 0) + (0\scalar x_{2} - 1 \scalar 0, 1 \scalar x_{2} + 0 \scalar 0)) \\
& = &
f((x_{1}, 0) + (0, x_{2})) \\
& = &
f(x_{1}, x_{2}) \\
\widehat{\widehat{g}}(x)
& = &
\widehat{g}(x, 0) \\
& = &
g(x) + i\scalar g(0) \\
& = &
g(x) + i \scalar 0 \\
& = &
g(x).
\end{array}$$

For completeness we add that on a map $f\colon X\rightarrow Y$ in
$\Mod[\R]$ we get $\mathcal{C}(f) \colon \mathcal{C}(X)
\rightarrow \mathcal{C}(Y)$ in $\Vect[\C]$ by
$\mathcal{C}(f)(x_{1},x_{2}) = (f(x_{1}), f(x_{2}))$. We check
preservation of scalar multiplication. For $z = (a+ib)\in\C$,
$$\begin{array}{rcl}
\mathcal{C}(f)((a+ib) \scalar (x_{1},x_{2}))
& = &
\mathcal{C}(f)(a\scalar x_{1} - b\scalar x_{2}, 
   b\scalar x_{1} + a\scalar x_{2}) \\
& = &
(f(a\scalar x_{1} - b\scalar x_{2}), f(b\scalar x_{1} + a\scalar x_{2})) \\
& = &
(a\scalar f(x_{1}) - b\scalar f(x_{2}), b\scalar f(x_{1}) + a\scalar f(x_{2})) \\
& = &
(a+ib) \scalar (f(x_{1}), f(x_{2})) \\
& = &
(a+ib) \scalar \mathcal{C}(f)(x_{1},x_{2}).
\end{array}$$
}

The inclusion morphism $\Pos(H) \hookrightarrow \SA(H)$ in
$\Mod[\Rnn]$ yields as transpose the map $\varphi
\colon \mathcal{R}(\Pos(H)) \rightarrow \SA(H)$ in
$\Vect[\R]$ given by $\varphi([B_{1}, B_{2}]) = B_{1} -
B_{2}$. It is surjective since each $A\in\SA(H)$ can be written as $A
= A_{p} - A_{n}$ for $A_{p},A_{n}\in\Pos(H)$ as
in~(\ref{PosNegOperatorEqn}). Thus $A = \varphi([A_{p}, A_{n}])$.

Similarly, the inclusion $\SA(H) \hookrightarrow \BL(H)$ in
$\Vect[\R]$ gives rise to a transpose $\psi \colon
\mathcal{C}(\SA(H)) \rightarrow \BL(H)$ in $\Vect[\C]$,
given by $\psi(B_{1},B_{2}) = B_{1} + iB_{2}$. Also this map is
surjective since each $A\in\BL(H)$ can be written as $A =
\frac{1}{2}(A+A^{\dag}) + \frac{1}{2}i(-iA + iA^{\dag})$, where
$A+A^{\dag}$ and $-iA + A^{\dag}$ are self-adjoints. Hence $A =
\psi(\frac{1}{2}(A+A^{\dag}), \frac{1}{2}(-iA + iA^{\dag})).$ \QED

\auxproof{
Some explicit (but unnecessary) calculations from an earlier stage:
$$\begin{array}{rclcrcl}
\psi(B_{1}, B_{2})
& = &
B_{1} + iB_{2}.
& \qquad &
\psi^{-1}(A)
& = &
(\frac{1}{2}(A+A^{\dag}), \frac{1}{2}(-iA + iA^{\dag}))
\end{array}$$

\noindent We have an isomorphism indeed:
$$\begin{array}{rcl}
\lefteqn{\psi(\psi^{-1}(A))} \\
& = &
\psi^{-1}_{1}(A) + i\psi^{-1}_{2}(A) \\
& = &
\frac{1}{2}(A+A^{\dag}) + i\frac{1}{2}(-iA + iA^{\dag}) \\
& = &
\frac{1}{2}(A + A^{\dag} + A - A^{\dag}) \\
& = &
A. \\
\lefteqn{\psi^{-1}(\psi(B_{1}, B_{2}))} \\
& = &
\psi^{-1}(B_{1} + iB_{2}) \\
& = &
(\frac{1}{2}((B_{1} + iB_{2})+(B_{1} + iB_{2})^{\dag}), 
   \frac{1}{2}(-i(B_{1} + iB_{2}) + i(B_{1} + iB_{2})^{\dag})) \\
& = &
(\frac{1}{2}(B_{1} + iB_{2}+ B_{1}^{\dag} - iB_{2}^{\dag}), 
   \frac{1}{2}(-iB_{1} + B_{2} + iB_{1}^{\dag} + B_{2}^{\dag})) \\
& = &
(\frac{1}{2}(B_{1} + iB_{2}+ B_{1} - iB_{2}), 
   \frac{1}{2}(-iB_{1} + B_{2} + iB_{1} + B_{2})) \\
& = &
(B_{1}, B_{2})
\end{array}$$

\noindent We check that $\psi^{-1}, \psi$ preserve scalar
multiplication:
$$\begin{array}{rcl}
\lefteqn{\psi^{-1}((a+ib) \scalar A)} \\
& = &
(\frac{1}{2}(((a+ib) \scalar A)+((a+ib) \scalar A)^{\dag}), 
    \frac{1}{2}(-i((a+ib) \scalar A) + i((a+ib) \scalar A)^{\dag})) \\
& = &
(\frac{1}{2}(aA + ibA + aA^{\dag} - ibA^{\dag}), 
    \frac{1}{2}(-iaA + bA + iaA^{\dag} + bA^{\dag})) \\
& = &
(a\frac{1}{2}(A+A^{\dag}) - b\frac{1}{2}(-iA + iA^{\dag}), 
   b\frac{1}{2}(A+A^{\dag}) + a\frac{1}{2}(-iA + iA^{\dag})) \\
& = &
(a+ib) \scalar (\frac{1}{2}(A+A^{\dag}), \frac{1}{2}(-iA + iA^{\dag})) \\
& = &
(a+ib) \scalar \psi^{-1}(A) \\
\lefteqn{\psi\big((a+ib) \scalar (B_{1},B_{2})\big)} \\
& = &
\psi(aB_{1} - bB_{2}, bB_{1} + aB_{2}) \\
& = &
aB_{1} - bB_{2} + i(bB_{1} + aB_{2}) \\
& = &
(a+ib)B_{1} + (a+ib)iB_{2} \\
& = &
(a+ib) \scalar \psi(B_{1},B_{2}).
\end{array}$$
}
\end{proof}

\section{Convex sets and effect modules}\label{EffectSec}

In the previous section we have seen how the spaces of operators
$\Pos(H)$, $\SA(H)$, $\BL(H)$ fit in the context of modules. The
spaces $\DM(H)$ of density operators (states) and $\Ef(H)$ of effects
(statements/predicates) require more subtle structures that will be
introduced in this section, namely convex sets and effect modules.  We
show that they are related by a dual adjunction, and that there exists
a map of adjunctions from Hilbert spaces, like in
Section~\ref{FunctOperSec}.

First we recall the definition of the two sets of operators that
are relevant in this section.
$$\begin{array}{rcl}
\DM(H)
& = &
\setin{A}{\Pos(H)}{\tr(A)=1} \\
\Ef(H)
& = &
\setin{A}{\Pos(H)}{A\leq I},
\end{array}$$

\noindent where $I$ is the identity map $H\rightarrow H$ and $\leq$ is
the L\"owner order (described before Proposition~\ref{PosDualProp}).
A further subset of $\Ef(H)$ is the set of projections, given as:
$$\begin{array}{rcl}
\Proj(H)
& = &
\setin{A}{\BL(H)}{A^{\dag} = A = AA}.
\end{array}$$

\noindent For a projection $A\in\Proj(H)$ there is an orthosupplement
$A^{\perp}\in\Proj(H)$ with $A + A^{\perp} = I$. This shows $A \leq
I$, since $I - A = A^{\perp}$ is positive.

Before we investigate the algebraic structure of these sets of
operators we briefly mention the following alternative formulation of
effects. It is used for instance in~\cite{dHondtP06a}, where these
effects $A$ are called predicates; they give a ``quantum expectation
value'' $\tr(AB)$ for a density matrix $B$ (see the map $\hs[\Ef]$ in
Theorem~\ref{ConvEModAdjMapThm} below, elaborated in
Remark~\ref{WPRem}).

\begin{lemma}
\label{EfAltLem}
A positive operator $A\in\Pos(H)$ is an effect if and only if all of
its eigenvalues are in $[0,1]$.
\end{lemma}

\begin{proof}
Suppose $A$ is an effect with spectral decomposition $A =
\sum_{j}\lambda_{j}\ket{j}\bra{j}$, where we may assume that the
eigenvectors $\ket{j}$ form an orthonormal basis. The eigenvalues
$\lambda_{j}$ are necessarily real and positive. They satisfy:
$$\lambda_{j}
=
\lambda_{j}\inprod{j}{j}
=
\bra{j}\lambda_{j}\ket{j}
=
\bra{j}A\ket{j}
\leq
\bra{j}I\ket{j}
=
\inprod{j}{j}
=
1.$$

\noindent Conversely, assume a positive operator $A$ with spectral
decomposition $A = \sum_{j}\lambda_{j}\ket{j}\bra{j}$ where the
$\ket{j}$ form an orthonormal basis and $\lambda_{j}\in[0,1]$. Then:
$A = \sum_{j}\lambda_{j}\ket{j}\bra{j} \leq \sum_{j}\ket{j}\bra{j} =
I$. \QED
\end{proof}

\subsection{Convex sets}

We start with convex sets, and (conveniently) describe them via a
monad, so that we can benefit from general results like in
Theorem~\ref{MonadAlgStructThm}. Analogously to the multiset monad one
defines the distribution monad $\Dstr\colon\Sets\rightarrow\Sets$ as:
\begin{equation}
\label{DstrUnitEqn}
\begin{array}{rcl}
\Dstr(X)
& = &
\set{\varphi\colon X\rightarrow \unitR}{\support(\varphi)
   \mbox{ is finite and }\sum_{x\in X}\varphi(x) = 1}.
\end{array}
\end{equation}

\noindent Elements of $\Dstr(X)$ are convex combinations
$s_{1}\ket{x_{1}}+\cdots+s_{k}\ket{x_{k}}$, where the probabilities
$s_{i}\in\unitR$ satisfy $\sum_{i}s_{i} = 1$.  Unit and multiplication
making $\Dstr$ a monad can be defined as for the multiset monad
$\Mlt_S$. This multiplication $\mu$ is well-defined since:
$$\textstyle
\sum_{x}\mu(\sum_{i}s_{i}\varphi_{i})(x)
=
\sum_{x}\sum_{i}s_{i}\cdot \varphi_{i}(x)
=
\sum_{i}s_{i}\cdot \big(\sum_{x}\varphi_{i}(x)\big)
=
\sum_{i}s_{i}
=
1.$$

\noindent The distribution monad $\Dstr$ is always symmetric monoidal
(commutative). Here it is defined for probabilities in the unit
interval $[0,1]$, but the more general structure of an ``effect
monoid'' may be used instead, see~\cite{Jacobs11c}.

The following result goes back to~\cite{Swirszcz74}, see
also~\cite{Keimel08,Doberkat06,Jacobs10e}.

\begin{theorem}
The category $\Alg(\Dstr)$ of algebras of the monad $\Dstr$ is the
category \Conv of convex sets with affine maps between them. \QED
\end{theorem}

Here we shall identify such a convex set simply with an algebra
$a\colon\Dstr(X)\rightarrow X$ of the monad $\Dstr$. It thus consists
of a set $X$ in which there is an interpretation
$a(\sum_{j}s_{j}\ket{x_{j}})\in X$ for each formal convex combination
$\sum_{j}s_{j}\ket{x_{j}}\in\Dstr(X)$. In particular, for each
$r\in[0,1]$ and $x,y\in X$ there is an interpretation of the convex
sum $rx + (r-1)y$, namely as $a(r\ket{x} + (1-r)\ket{y})\in X$. The
unit interval $[0,1]$ of real numbers is an obvious example of a
convex set. Actually, it is a free one since $[0,1] \cong
\Dstr(\{0,1\})$. Affine maps preserve such interpretations of convex
combinations. We recall that in the present context all such convex
combinations involve only finitely many elements $x_j$.

\begin{lemma}
\label{DMFunLem}
Let $\FdHilbUn$ be the category of finite-dimensional Hilbert spaces
with unitary maps between them.  Taking density operators yields a
functor $\DM\colon \FdHilbUn \rightarrow \Conv = \Alg(\Dstr)$.
\end{lemma}

\begin{proof}
As is well-known, the set $\DM(H)$ of densitity operators is convex:
given finitely many $A_{j}\in\DM(H)$ and $r_{j}\in[0,1]$ with
$\sum_{j}r_{j} = 1$, the operator $A = \sum_{j}r_{j}A_{j}$ is positive
and has trace 1, since:
$$\textstyle\tr(A)
=
\tr(\sum_{j}r_{j}A_{j})
=
\sum_{j}r_{j}\tr(A_{j})
=
\sum_{j}r_{j}
=
1.$$

\noindent Moreover, if $U\colon H\rightarrow K$ is
unitary---\textit{i.e.}~$UU^{\dag} = I$ and (thus) $U^{\dag}U = I$, so
that $U^{\dag} = U^{-1}$---then $\DM(U)(A) = UAU^{\dag} \colon K
\rightarrow K$ is in $\DM(K)$, if $A\in\DM(H)$, since:
$$\tr\big(\DM(U)(A)\big)
=
\tr(UAU^{\dag})
=
\tr(U^{\dag}UA)
=
\tr(IA)
=
\tr(A)
=
1.\eqno{\QEDbox}$$
\end{proof}

The three example categories $\Alg(\Mlt_{S})$ of interest here---that
arise from multiset monads $\Mlt_S$ for $S=\Rnn,
\R, \C$---are different from the category
$\Alg(\Dstr)$ of convex sets in at least three aspects:
\begin{itemize}
\item These categories $\Alg(\Mlt_{S})$ are dually self-adjoint via
the functor $(-)\multimap S$ as in~(\ref{ModuleDualDiag}).

\item They have biproducts, because the monads $\Mlt_S$ are
  `additive', see~\cite{CoumansJ12}.

\item The tensor unit in $\Alg(\Dstr) = \Conv$ is the singleton
  set $1$, since $\Dstr(1) \cong 1$, so that tensors $\otimes$ in
  \Conv have projections (see~\cite{Jacobs94a}).
\end{itemize}

\noindent The mapping $X\mapsto \Conv(X, [0,1])$, for $X$ a convex
set, does not yield an adjunction as in~(\ref{ModuleDualDiag}), but
does lead to an interesting dual adjunction with a category of `effect
modules'. This will be the described in the next subsection.

\auxproof{
The left adjoint $F$ is described in~\cite{Jacobs10e},
namely on $X = \smash{(\Dstr(X)\stackrel{a}{\rightarrow} X)}
\in\Alg(\Dstr)$ as:
$$\begin{array}{rcl}
F(X)
& = &
\{0\} + \R_{>0}\times X,
\end{array}$$

\noindent with addition for $u,v\in F(X)$, in trivial
cases given by $u+0 = u = 0+u$ and:
$$\begin{array}{rcl}
(s,x)+(t,y)
& = &
(s+t, a(\frac{s}{s+t}x + \frac{t}{s+t}y)) 
\end{array}$$

\noindent A scalar multiplication $\bullet\colon \Rnn\times
F(X)\rightarrow F(X)$ is defined as:
$$\begin{array}{rcl}
s\scalar u
& = &
\left\{\begin{array}{ll}
0 & \mbox{if $u=0$ or $s=0$} \\
(s\cdot t, x) & \mbox{if }u = (t,x)\mbox{ and }s\neq 0.
\end{array}\right.
\end{array}$$
}

But first we conclude this part on convex sets with an observation
like in Theorem~\ref{AlgCatAdjThm}.  There is an obvious map of monads
$\Dstr \Rightarrow \Mlt_{\Rnn}$, that gives rise to an inclusion
functor $\Mod[\Rnn] = \Alg(\Mlt_{\Rnn}) \rightarrow \Alg(\Dstr) =
\Conv$, saying that modules over non-negative reals are convex
sets---in a trivial manner. For general reasons, this functor has a
left adjoint, that can be described explicitly in terms of a
representation contruction that goes back to~\cite{Stone49} (see
also~\cite{Jacobs10e}). This left adjoint $\mathcal{S} \colon \Conv
\rightarrow \Mod[\Rnn]$ is given on $X\in\Conv$ by:
$$\begin{array}{rcl}
\mathcal{S}(X)
& = &
\{0\} + \R_{>0}\times X,
\end{array}$$

\noindent with addition for $u,v\in \mathcal{S}(X)$, in trivial
cases given by $u+0 = u = 0+u$ and:
$$\begin{array}{rcl}
(s,x)+(t,y)
& = &
(s+t, \; \frac{s}{s+t}x + \frac{t}{s+t}y).
\end{array}$$

\noindent A scalar multiplication $\scalar\colon \Rnn\times
\mathcal{S}(X)\rightarrow \mathcal{S}(X)$ is defined as:
$$\begin{array}{rcl}
s\scalar u
& = &
\left\{\begin{array}{ll}
0 & \mbox{if $u=0$ or $s=0$} \\
(s\cdot t, x)\quad & \mbox{if }u = (t,x)\mbox{ and }s\neq 0.
\end{array}\right.
\end{array}$$

\noindent This makes $\mathcal{S}(X)$ a module over $\Rnn$.

\begin{theorem}
\label{ConvModCatAdjThm}
For a finite-dimensional Hilbert space $H$, transposing the inclusion
$\DM(H) \hookrightarrow \Pos(H)$ in \Conv gives an isomorphism
$\mathcal{S}(\DM(H)) \conglongrightarrow \Pos(H)$ in
$\Mod[\Rnn]$. In this way one obtains a triangle
commuting up-to-isomorphism:
$$\xymatrix@C-1.5pc@R-.5pc{
& \FdHilbUn\ar[dl]_{\DM}\ar[dr]^{\Pos} \\
\llap{$\Alg(\Dstr)=\;$}\Conv\ar[rr]_-{\mathcal{S}}^-{\mbox{\small free}} 
   & & \Mod[\Rnn]\rlap{$\;=\Alg(\Mlt_{\Rnn})$}
}$$
\end{theorem}

\begin{proof}
The induced map $\smash{\mathcal{S}(\DM(H)) = \{0\} + \Rnn\times
  \DM(H) \stackrel{\varphi}{\longrightarrow} \Pos(H)}$ is given by
$0\mapsto 0$ and $(r,A) \mapsto rA$. It is injective, since if $rA =
sB$ for $A,B\in\DM(H)$, then $r = r\cdot \tr(A) = \tr(rA) = \tr(sB) =
s\cdot \tr(B) = s$, and thus $A=B$. It is also surjective: since each
non-zero $B\in\Pos(H)$ can be written as $B = \tr(B)
(\frac{B}{\tr(B)}) = \varphi(\tr(B), \frac{B}{\tr(B)})$, where the
operator $\frac{B}{\tr(B)}$ has trace 1 by construction. \QED
\end{proof}

By combining this result with Theorem~\ref{AlgCatAdjThm} we see that
each of the spaces of operators $\BL(H)$, $\SA(H)$, $\Pos(H)$ can be
obtained from the space $\DM(H)$ of density operators via free
constructions. As we will see in Theorem~\ref{ConvEModAdjMapThm}
below, density matrices and effects can be translated back and forth:
$\Ef(H) \cong \Conv(\DM(H), [0,1])$ and $\DM(H) \cong \EMod(\Ef(H),
[0,1])$. Hence these density operators and effects are in a sense most
fundamental among the operators on a Hilbert space.

\subsection{Effect modules}

Effect modules are structurally like modules over a semiring. But
instead of a semiring of scalars one uses an effect monoid, such as
the unit interval $[0,1]$. Such an effect monoid is a monoid in the
category of effect algebras, just like a semiring is a monoid in the
category of commutative monoids. Thus, in order to define an effect
module, we need the notion of effect algebra and of monoid in effect
algebras. This will be introduced first.

But in order to define an effect algebra, we need the notion of
partial commutative monoid (PCM). Before reading the definition of
PCM, think of the unit interval $[0,1]$ with addition $+$. This $+$ is
obviously only a partial operation, which is commutative and
associative in a suitable sense. This will be formalised next.

A partial commutative monoid (PCM) consists of a set $M$ with a zero
element $0\in M$ and a partial binary operation $\ovee\colon M\times
M\rightarrow M$ satisfying the three requirements below. They involve
the notation $x\orthogonal y$ for: $x\ovee y$ is defined; in that case
$x,y$ are called orthogonal.
\begin{enumerate}
\item Commutativity: $x\orthogonal y$ implies $y\orthogonal x$ and
$x\ovee y = y\ovee x$;

\item Associativity: $y\orthogonal z$ and $x \orthogonal (y\ovee z)$
implies $x\orthogonal y$ and $(x\ovee y) \orthogonal z$ and also
$x \ovee (y\ovee z) = (x\ovee y)\ovee z$;

\item Zero: $0\orthogonal x$ and $0\ovee x = x$;
\end{enumerate}

\noindent For each set $X$ the lift $\{0\}+X$ of $X$, obtained by
adjoining a new element $0$, is an example of a PCM, with $u\ovee 0 =
u = 0 \ovee u$, and $\ovee$ undefined otherwise. Such structures are
also studied under the name `partially additive monoid',
see~\cite{ArbibM86}.

The notion of effect algebra is due to~\cite{FoulisB94}, see
also~\cite{DvurecenskijP00} for an overview.

\begin{definition}
\label{EffAlgDef}
An effect algebra is a PCM $(E, 0, \ovee)$ with an
orthosupplement. The latter is a unary operation $(-)^{\perp} \colon
E\rightarrow E$ satisfying:
\begin{enumerate}
\item $x^{\perp}\in E$ is the unique element in $E$ with $x\ovee
  x^{\perp} = 1$, where $1 = 0^\perp$;

\item $x\orthogonal 1 \Rightarrow x=0$.
\end{enumerate}

A homomorphism $E\rightarrow D$ of effect algebras is given by a
function $f\colon E\rightarrow D$ between the underlying sets
satisfying $f(1) = 1$, and if $x\orthogonal x'$ in $E$ then both $f(x)
\orthogonal f(x')$ in $D$ and $f(x\ovee x') = f(x) \ovee f(x')$.

Effect algebras and their homomorphisms form a category, called
\EA.
\end{definition}

The unit interval $[0,1]$ is a PCM with sum of $r,s\in[0,1]$ defined
if $r+s\leq 1$, and in that case $r\ovee s=r+s$. The unit interval is
also an effect algebra with $r^{\perp} = 1 - r$. Each orthomodular
lattice is an effect algebra, see~\cite{DvurecenskijP00,Foulis07} for
more information and examples. In particular, the projections
$\Proj(H)$ of a Hilbert space form an effect algebra, with $P
\orthogonal Q$ iff $P \leq Q^{\perp}$. In~\cite{Jacobs11c} a notion of
`convex category' is introduced in which homsets $\Hom(X,2)$ are
effect algebras (where $2 = 1+1$ and $1$ is final). Most importantly
in the current setting, the set of effects $\Ef(H)$, consisting of
positive operators $A \leq I$ is an effect algebra, with $A
\orthogonal B$ iff $A+B \leq I$, and in that case $A\ovee B = A+B$;
further, $A^{\perp} = I - A$. This yields a functor $\Ef \colon
\FdHilbUn \rightarrow \EA$.

\auxproof{
Assume a unitary map $U \colon H\rightarrow K$. We have $\Ef(U)(A) =
UAU^{\dag} \colon K \rightarrow K$. It is again a positive operator, which
is below the identity, since:
$$\Ef(U)(A)
=
UAU^{\dag}
\leq
UIU^{\dag}
=
UU^{\dag}
=
I.$$

\noindent Further, if $A+B\leq I$, then 
$$\Ef(U)(A) + \Ef(U)(B)
=
UAU^{\dag} + UBU^{\dag}
=
U(A+B)U^{\dag}
\leq
UIU^{\dag}
=
I.$$

\noindent Thus: $\Ef(U)(A) \ovee \Ef(U)(B) = UAU^{\dag} + UBU^{\dag} =
U(A+B)U^{\dag} = \Ef(U)(A\ovee B)$. Further, $\Ef(U)(I) = UIU^{\dag} = I$.
}

In~\cite{JacobsM12a} it is shown that the category \EA is
symmetric monoidal, where morphisms $E\otimes D\rightarrow C$ in
\EA correspond to `bimorphisms' $f\colon E\times D \rightarrow C$,
satisfying $f(1,1) = 1$, and for all $x,x'\in E$ and $y,y'\in D$,
$$\left\{\begin{array}{rcl}
x\orthogonal x'
& \Longrightarrow &
f(x,y) \orthogonal f(x',y) \;\mbox{ and }\;
f(x\ovee x',y)
\hspace*{\arraycolsep} = \hspace*{\arraycolsep}
f(x,y) \ovee f(x',y) \\
y\orthogonal y'
& \Longrightarrow &
f(x,y) \orthogonal f(x,y') \;\mbox{ and }\;
f(x,y \ovee y')
\hspace*{\arraycolsep} = \hspace*{\arraycolsep}
f(x,y) \ovee f(x,y').
\end{array}\right.$$

\noindent The tensor unit is the two-element effect algebra $2 =
\{0,1\}$. Since 2 is at the same time initial in \EA we have a `tensor
with coprojections' (see~\cite{Jacobs94a} for `tensors with
projections'). One can think of elements of the tensor $E\otimes D$ as
finite sums $\ovee_{j}\,x_{j}\sotimes y_{j}$, where one identifies:
$$\begin{array}{rclcrcl}
0\sotimes y
& = & 
0
& \qquad &
x\sotimes 0
& = &
0 \\
(x\ovee x')\sotimes y
& = &
(x\sotimes y) \ovee (x'\sotimes y)
& &
x\sotimes (y\ovee y')
& = &
(x\sotimes y) \ovee (x\sotimes y'),
\end{array}$$

\noindent when $x\orthogonal x'$ and $y\orthogonal y'$.

\begin{example}
\label{EATensorEx}
For an arbitrary set $X$ the powerset $\Pow(X)$ is a Boolean algebra,
and so an orthomodular lattice, and thus an effect algebra. For
$U,V\in\Pow(X)$ one has $U\orthogonal V$ iff $U\cap V = \emptyset$ and
in that case $U\ovee V = U\cup V$. The tensor product
$[0,1]\otimes\Pow(X)$ of effect algebras is then given by the set of
step functions $f\colon X\rightarrow [0,1]$; such functions have only
finitely many output values.  When $X$ is a finite set, say with $n$
elements, then $[0,1]\otimes \Pow(X) \cong [0,1]^{n}$,
see~\cite{Gudder96}.

As special case we have $[0,1]\otimes\{0,1\} \cong [0,1]$, since
$\{0,1\}$ is the tensor unit. One writes $\textit{MO}(n)$ for the
orthomodular lattice with $2n+2$ elements, namely $0,1,i,i^{\perp}$,
for $1\leq i\leq n$, with only minimal equations. Thus $\textit{MO}(0)
= \{0,1\}$ and $\textit{MO}(1) \cong \Pow(\{0,1\})$, so that
$[0,1]\otimes \textit{MO}(1) \cong [0,1]^{2}$. It can be shown that
$[0,1]\otimes\textit{MO}(2)$ is an octahedron.
\end{example}

Using this symmetric monoidal structure $(\otimes,2)$ on $\EA$ we can
consider, in a standard way, the category $\Mon(\EA)$ of monoids in
the category $\EA$ of effect algebras. Such monoids are similar to
semirings, which are monoids in the category of commutative monoids,
\textit{i.e.}~objects of $\Mon(\Cat{CMon})$. A monoid $S\in\Mon(\EA)$
consists of a set $S$ carrying effect algebra structure
$(0,\ovee,(-)^{\perp})$ and a monoid structure, written
multiplicatively, as in: $\smash{S\otimes S
  \stackrel{\cdot}{\rightarrow} S \leftarrow 2}$. Since $2$ is
initial, the latter map $S\leftarrow 2$ does not add any
structure. The monoid structure on $S$ is thus determined by a
bimorphism $\cdot \colon S \times S \rightarrow S$ that preserves
$\ovee$ in each variable separately and satisfies $1\cdot x = x =
x\cdot 1$.

For such a monoid $S\in\Mon(\EA)$ we can consider the category
$\Act_{S}(\EA) = \EMod_{S}$ of $S$-monoid actions (scalar
multiplications), or `effect modules' over $S$
(see~\cite[VII,\S4]{MacLane71}). Again this is similar to the
situation in Section~\ref{ModuleSec} where the category
$\Mod[S]$ of modules over a semiring $S$ may be described as the
category $\Act_{S}(\Cat{CMon})$ of commutative monoids with $S$-scalar
multiplication. In this section an effect module $X\in\Act_{S}(\EA)$
thus consists of an effect algebra $X$ together with an action (or
scalar multiplication) $\scalar \colon S\otimes X \rightarrow X$,
corresponding to a bimorphism $S\times X\rightarrow X$. A homomorphism
of effect modules $X\rightarrow Y$ consists of a map of effect
algebras $f\colon X\rightarrow Y$ preserving scalar multiplication
$f(s\scalar x) = s\scalar f(x)$ for all $s\in S$ and $x\in X$.

By completely general reasoning the forgetful functor $\EMod_{S}
\rightarrow \EA$ has a left adjoint, given by tensoring with
$S$, as in:
\begin{equation}
\label{FreeActEADiag}
\vcenter{\xymatrix@R-1pc{
\EMod_{S}\rlap{$\;=\Act_{S}(\EA)$}\ar[d]_{\dashv} \\
\EA\ar@/^3ex/[u]^{S\otimes(-)}
}}
\end{equation}

\noindent See~\cite[VII,\S4]{MacLane71} for details.

The main example of a (commutative) monoid in \EA is the unit interval
$[0,1] \in \EA$ via ordinary multiplication. If $r_{1}+r_{2}\leq 1$,
then we have the familiar distributivity in each variable, as in:
$$s\cdot (r_{1}\ovee r_{2})
=
s\cdot (r_{1}+r_{2})
=
(s\cdot r_{1}) + (s\cdot r_{2})
=
(s \cdot r_{1}) \ovee (s \cdot r_{2}).$$

\noindent We shall be most interested in the associated category
$\EMod_{[0,1]} = \Act_{[0,1]}(\EA)$. In the sequel `effect
module' will mean `effect module over $[0,1]$'. In particular, we
shall write $\EMod$ for $\EMod_{[0,1]}$. These effect
modules have been studied earlier under the name `convex effect
algebras', see~\cite{PulmannovaG98}. We prefer the name `effect
module' to emphasise the similarity with ordinary modules.

The effects $\Ef(H)$ of a Hilbert space form an example of an effect
module, with the usual scalar multiplication $[0,1] \times \Ef(H)
\rightarrow \Ef(H)$. It is not hard to see that this mapping $H\mapsto
\Ef(H)$ yields a functor $\FdHilbUn \rightarrow \EMod$.

\auxproof{
Given a unitary map $U\colon H\rightarrow K$ we have to show
that $\Ef(U)(A) = UAU^{\dag} \in \Ef(K)$ for $A\in\Ef(H)$. 
But $A \leq I$ yields $\Ef(U)(A) = UAU^{\dag} \leq UIU^{\dag} = I$.
Further, $\Ef(U) \colon \Ef(H) \rightarrow \Ef(K)$ preserves
scalar multiplication since:
$$\Ef(U)(r\scalar A)
=
U(rA)U^{-1}
=
rUAU^{-1}
=
r\scalar \Ef(U)(A).$$
}

A (dual) adjunction between convex sets and effect algebras is
described in~\cite{Jacobs10e}. Here it is strengthened to an
adjunction between convex sets and effect modules.

\begin{proposition}
\label{ConvEModAdjProp}
By ``homming into $[0,1]$'' one obtains an adjunction:
$$\xymatrix{
\Conv\ar@/^1.5ex/[rr]^-{\Conv(-,[0,1])} 
   & \bot & 
   \EMod\rlap{$\op$}
   \ar@/^1.5ex/[ll]^-{\EMod(-,[0,1])}
}$$
\end{proposition}

\begin{proof}
Given a convex set, the homset $\Conv(X,[0,1])$ of affine maps is an
effect module, with $f\orthogonal g$ iff $\allin{x}{X}{f(x) + g(x)
  \leq 1}$. In that case one defines $f\ovee g =
\lamin{x}{X}{f(x)+g(x)}$.  It is easy to see that this is again an
affine function. Similarly, the pointwise scalar product $r\scalar f =
\lamin{x}{X}{r\cdot f(x)}$ yields an affine function. This mapping $X
\mapsto \Conv(X,[0,1])$ gives a contravariant functor since for
$h\colon X\rightarrow X'$ in \Conv pre-composition with $h$ yields a
map $(-) \after h \colon \Conv(X', [0,1]) \rightarrow \Conv(X, [0,1])$
of effect modules.

\auxproof{
For a formal convex sum $\sum_{j}s_{j}x_{j}$ one has:
$$\begin{array}{rcl}
(f \ovee g)(\sum_{j}s_{j}x_{j})
& = &
f(\sum_{j}s_{j}x_{j}) + g(\sum_{j}s_{j}x_{j}) \\
& = &
\sum_{j}s_{j}f(x_{j}) + \sum_{j}s_{j}g(x_{j}) \\
& = &
\sum_{j}s_{j}(f(x_{j}) + g(x_{j})) \\
& = &
\sum_{j}s_{j}(f\ovee g)(x_{j}) \\
(r\scalar f)(\sum_{j}s_{j}x_{j})
& = &
r\cdot f(\sum_{j}s_{j}x_{j}) \\
& = &
r\cdot (\sum_{j}s_{j}f(x_{j})) \\
& = &
\sum_{j}s_{j} r\cdot f(x_{j}) \\
& = &
\sum_{j}s_{j}(r\scalar f)(x_{j}) \\
r \scalar (f \ovee g)
& = &
\lam{x}{r\cdot (f \ovee g)(x)} \\
& = &
\lam{x}{r\cdot (f(x) + g(x))} \\
& = &
\lam{x}{r\cdot f(x) + r \cdot g(x))} \\
& = &
\lam{x}{(r\scalar f)(x) + (r\scalar g)(x)} \\
& = &
(r \scalar f) + (r \scalar g) \\
(f \ovee g) \after h
& = &
\lam{x}{f(h(x)) + g(h(x))} \\
& = &
\lam{x}{(f \after h)(x) + (g \after h)(x)} \\
& = &
(f \after h) \ovee (f \after h) \\
(r \scalar f) \after h
& = &
\lam{x}{r \cdot f(h(x))} \\
& = &
\lam{x}{r \cdot (f \after h)(x)} \\
& = &
r \scalar (f \after h).
\end{array}$$
}

In the other direction, given an effect module $Y$, the homset
$\EMod(Y, [0,1])$ of effect module maps yields a convex set: for a
formal convex sum $\sum_{j}r_{j}\ket{f_{j}}$, where $f_{j} \colon Y
\rightarrow [0,1]$ in $\EMod$, we can define an actual sum $f\colon
Y\rightarrow [0,1]$ by $f(y) = \sum_{j} r_{j}\cdot f_{j}(y)$. This $f$
forms a map of effect modules. Again, functoriality is obtained via
pre-composition.

\auxproof{
$$\begin{array}{rcl}
f(y \ovee z)
& = &
\sum_{j}r_{j}\cdot f_{j}(y\ovee z) \\
& = &
\sum_{j}r_{j}\cdot (f_{j}(y) \ovee f_{j}(z)) \\
& = &
\sum_{j}r_{j}\cdot (f_{j}(y) + f_{j}(z)) \\
& = &
\sum_{j}r_{j}\cdot f_{j}(y) + r\cdot f_{j}(z)) \\
& = &
\sum_{j}r_{j}\cdot f_{j}(y) + \sum_{j}r\cdot f_{j}(z)) \\
& = &
f(y) + f(z) \\
& = &
f(y) \ovee f(z) \\
f(1)
& = &
\sum_{j}r_{j}\cdot f_{j}(1) \\
& = &
\sum_{j}r_{j} \cdot 1 \\
& = &
\sum_{j}r_{j} \\
& = &
1.
\end{array}$$

For a map $h\colon Y \rightarrow Y'$ of effect modules precomposition
with $h$ gives an affine map $(-) \after h \colon \EMod(Y', [0,1])
\rightarrow \EMod(Y, [0,1])$ since for a convex sum
$\sum_{j}r_{j}f_{j}$ of effect algebra maps $f_{j} \colon
Y'\rightarrow [0,1]$ we get:
$$\begin{array}{rcl}
(\sum_{j}r_{j}f_{j}) \after h
& = &
\lam{y}{(\sum_{j}r_{j}f_{j})(h(y))} \\
& = &
\lam{y}{\sum_{j}r_{j}\cdot f_{j}(h(y))} \\
& = &
\lam{y}{\sum_{j}r_{j}\cdot (f_{j} \after h)(y))} \\
& = &
\sum_{j}r_{j}\cdot (f_{j} \after h).
\end{array}$$
}

The dual adjunction between \Conv and $\EMod$ involves a
bijective correspondence that is obtained by swapping arguments, like
in~(\ref{ModuleDualDiag}). For $X\in\Conv$ and $Y\in\EMod$,
we have:
$$\begin{prooftree}
{\xymatrix{ 
X\ar[r]^-{f} & \EMod(Y,[0,1]) 
   \rlap{\hspace*{2em} in $\Conv$}}}
\Justifies
{\xymatrix{ Y\ar[r]_-{g} & \Conv(X, [0,1])
   \rlap{\hspace*{2.3em} in $\EMod$}}}
\end{prooftree}$$

\noindent What needs to be checked is that for a map $f$ of convex
sets as indicated, the swapped version $\widehat{f} =
\lamin{y}{Y}{\lamin{x}{X}{f(x)(y)}} \colon Y \rightarrow
\Conv(X,[0,1])$ is a map of effect modules---and similarly for $g$.
This is straightforward. \QED

\auxproof{
We show that $\widehat{f}(y) \colon X\rightarrow [0,1]$ is
a map of convex sets:
$$\begin{array}{rcl}
\widehat{f}(y)(\sum_{j}r_{j}x_{j})
& = &
f(\sum_{j}r_{j}x_{j})(y) \\
& = &
\big(\sum_{j}r_{j}f(x_{j})\big)(y) \\
& = &
\sum_{j}r_{j}f(x_{j})(y) \\
& = &
\sum_{j}r_{j}\widehat{f}(y)(x_{j}).
\end{array}$$

\noindent Next we have to check that $\widehat{f}$ is a map
of effect modules.
$$\begin{array}{rcl}
\widehat{f}(y_{1} \ovee y_{2})
& = &
\lam{x}{f(x)(y_{1} \ovee y_{2})} \\
& = &
\lam{x}{f(x)(y_{1}) \ovee f(x)(y_{2})} \\
& = &
\lam{x}{\widehat{f}(y_{1})(x) \ovee \widehat{f}(y_{2})(x)} \\
& = &
\widehat{f}(y_{1}) \ovee \widehat{f}(y_{2}) \\
\widehat{f}(1)
& = &
\lam{x}{f(x)(1)} \\
& = &
\lam{x}{1} \\
& = &
1
\end{array}$$

Next, starting from $g$ we take $\widehat{g} =
\lamin{x}{X}{\lamin{y}{Y}{g(y)(x)}}$. First, $\widehat{g}(x)$ is a
map of effect modules.
$$\begin{array}{rcl}
\widehat{g}(x)(y_{1} \ovee y_{2})
& = &
g(y_{1} \ovee y_{2})(x) \\
& = &
(g(y_{1}) \ovee g(y_{2}))(x) \\
& = &
g(y_{1})(x) \ovee g(y_{2})(x) \\
& = &
\widehat{g}(x)(y_{1}) \ovee \widehat{g}(x)(y_{2}) \\
\widehat{g}(x)(1)
& = &
g(1)(x) \\
& = &
(\lam{z}{1})(x) \\
& = &
1.
\end{array}$$

\noindent Also, $\widehat{g}$ is affine:
$$\begin{array}{rcl}
\widehat{g}(\sum_{j}r_{j}x_{j})
& = &
\lam{y}{g(y)(\sum_{j}r_{j}x_{j})} \\
& = &
\lam{y}{\sum_{j}r_{j}g(y)(x_{j})} \\
& = &
\lam{y}{\sum_{j}r_{j}\widehat{g}(x_{j})(y)} \\
& = &
\sum_{j}r_{j}\widehat{g}(x_{j}).
\end{array}$$
}
\end{proof}

With this adjunction in place we can give a clearer picture of density
matrices and effects, forming a map of adjunctions (like in
Section~\ref{FunctOperSec}. The isomorphisms involved are well-known,
see~\textit{e.g.}~\cite{Busch03}, but the framing of the relevant
structure in terms of maps of adjunctions is new.

\begin{theorem}
\label{ConvEModAdjMapThm}
There is a `dual adjunction' between convex sets and effect modules
as in the lower part of the diagram below. Further, there are 
natural isomorphisms:
\begin{equation}
\label{hsEfDMisoEqn}
\hspace*{-.7em}\vcenter{\xymatrix@C-.4pc@R-2pc{
\Ef(H)\ar[r]^-{\hs[\Ef]}_-{\cong} & \Conv\big(\DM(H), [0,1]\big)
\hspace*{-.5em} & \hspace*{-.5em}
\DM(H)\ar[r]^-{\hs[\DM]}_-{\cong} & \EMod\big(\Ef(H), [0,1]\big) \\
A\ar@{|->}[r] & \tr(A-)
\hspace*{-.5em} & \hspace*{-.5em}
B\ar@{|->}[r] & \tr(B-)
}}
\end{equation}

\noindent that give rise to a map of adjunctions given by states $\DM$
and statements (effects) $\Ef$ in:
$$\xymatrix{
\FdHilbUn\ar@/^1.2ex/[rr]^-{(-)^{\dag}}\ar[d]_{\DM} & \bot & 
   \FdHilbUn\rlap{$\op$}\ar@/^1.2ex/[ll]^-{(-)^{\dag}}\ar[d]^{\Ef} \\
\Conv\ar@/^1.2ex/[rr]^-{\Conv(-,[0,1])} 
   & \bot & 
   \EMod\rlap{$\op$}
   \ar@/^1.2ex/[ll]^-{\EMod(-,[0,1])}
}$$
\end{theorem}

\begin{proof}
This map of adjunctions involves natural
isomorphisms~(\ref{hsEfDMisoEqn}), in the categories $\EMod$ and
$\Conv$. We start with the first one, labeled $\hs[\Ef]$
in~(\ref{hsEfDMisoEqn}), and note that it is well-defined: for
$A\in\Ef(H)$ and $B\in\DM(H)$ one has:
$$\hs[\Ef](A)(B)
=
\tr(AB)
\leq
\tr(IB)
=
\tr(B)
=
1.$$

\noindent Injectivity of $\hs[\Ef]$ is obtained as follows. Assume
$A_{1},A_{2}\in\Ef(H)$ satisfy $\hs[\Ef](A_{1}) = \hs[\Ef](A_{2})$,
\textit{i.e.}~$\tr(A_{1}-) = \tr(A_{2}-) \colon \DM(H) \rightarrow
       [0,1]$. for an arbitrary non-zero element $x\in H$ there is a
       density matrix $B_{x} = \frac{\ket{x}\bra{x}}{|x|^{2}} \colon H
       \rightarrow H$. Thus $\tr(A_{1}B_{x}) = \tr(A_{2}B_{x})$. Then:
$$\begin{array}{rcl}
\inprod{(A_{1}-A_{2})x}{x}
& = &
\inprod{x}{A_{1}x} - \inprod{x}{A_{2}x} \\
& = &
\tr(\bra{x}A_{1}\ket{x}) - \tr(\bra{x}A_{2}\ket{x}) \\
& = &
\tr(A_{1}\ket{x}\bra{x}) - \tr(A_{2}\ket{x}\bra{x}) \\
& = &
|x|^{2}\big(\tr(A_{1}B_{x}) - \tr(A_{2}B_{x})\big) \\
& = &
0.
\end{array}$$

\noindent Since this equation holds for all $x\in H$, including $x=0$,
we get $A_{1} - A_{2} \geq 0$, and thus $A_{2} \leq A_{1}$. Similarly
$A_{1} \leq A_{2}$, and thus $A_{1} = A_{2}$.

For surjectivity of $\hs[\Ef]$ assume a morphism of convex sets
$h\colon \DM(H) \rightarrow [0,1]$. We turn it into a linear map
$h'\colon \Pos(H) \rightarrow \Rnn$ in the category $\Mod[\Rnn]$ of
modules over $\Rnn$ via:
$$\begin{array}{rcl}
h'(B)
& = &
\left\{\begin{array}{ll}
0 & \mbox{if $B=0$, or equivalently, $\tr(B) = 0$} \\
\tr(B) \cdot h\big(\frac{B}{\tr(B)}\big) & \mbox{otherwise.}
\end{array}\right.
\end{array}$$

\auxproof{
Brief check. Obviously $B=0 \Rightarrow \tr(B)=0$. Conversely,
assuming $B\in\Pos(H)$ with spectral decomposition $B =
\sum_{j}\lambda_{j}\ket{j}\bra{j}$ with $\lambda_{j} \in \Rnn$.
Then $\tr(B) = \sum_{j}\lambda_{j}=0$ implies $\lambda_{j}=0$
for each $j$, and thus $B=0$.
}

\noindent This is well-defined since $\tr(\frac{B}{\tr(B)}) =
\frac{\tr(B)}{\tr(B)} = 1$. We check linearity of $h'$. It is easy
to see that $h'(rB) = rh'(B)$, for $r\in\Rnn$, and for
non-zero $B,C\in\Pos(H)$ we have:
$$\begin{array}{rcl}
h'(B)+h'(C)
& = &
\tr(B)\cdot h\big(\frac{B}{\tr(B)}\big) + 
   \tr(C) \cdot h\big(\frac{C}{\tr(C)}\big) \\
& = &
\tr(B+C)\cdot \Big(\frac{\tr(B)}{\tr(B+C)}\cdot h\big(\frac{B}{\tr(B)}\big)
  + \frac{\tr(C)}{\tr(B+C)}\cdot h\big(\frac{C}{\tr(C)}\big)\Big) \\
& = &
\tr(B+C)\cdot \Big(h\big(\frac{\tr(B)}{\tr(B+C)}\cdot \frac{B}{\tr(B)} + 
    \frac{\tr(C)}{\tr(B+C)}\cdot \frac{C}{\tr(C)}\big)\Big) \\
& & \qquad \mbox{since $h$ preserves convex sums and:} \\
& & \qquad
  \frac{\tr(B)}{\tr(B+C)} + \frac{\tr(C)}{\tr(B+C)} =
  \frac{\tr(B)}{\tr(B)+\tr(C)} + \frac{\tr(C)}{\tr(B)+\tr(C)} =
  1 \\
& = &
\tr(B+C)\cdot h\big(\frac{B + C}{\tr(B+C)}\big) \\
& = &
h'(B+C).
\end{array}$$

\auxproof{
$$\begin{array}{rcl}
h'(rB)
& = &
\tr(rB)\cdot h\big(\frac{rB}{\tr(rB)}\big) \\
& = &
r\tr(B)\cdot h\big(\frac{rB}{r\tr(B)}\big) \\
& = &
r\tr(B)\cdot h\big(\frac{B}{\tr(B)}\big) \\
& = &
rh'(B).
\end{array}$$
}

\noindent By Proposition~\ref{PosDualProp} there is a unique
$A=\hs[\Pos]^{-1}(h')\in\Pos(H)$ with $h' = \tr(A-) \colon \Pos(H)
\rightarrow \Rnn$. For a density operator $B\in\DM(H) \hookrightarrow
\Pos(H)$ we get $\tr(AB) = h'(B) = h(B)\in [0,1]$. We claim that $A$
is an effect, \textit{i.e.}~is in $\Ef(H) \hookrightarrow
\Pos(H)$. Write $A = \sum_{j}\lambda_{j}\ket{j}\bra{j}$ as spectral
decomposition, where the $\ket{j}$ form an orthonormal basis. By
Lemma~\ref{EfAltLem} we need to prove $\lambda_{j} \leq 1$. Each
operator $\ket{j}\bra{j}$ is a density matrix, and thus $\lambda_{j} =
\tr(A\ket{j}\bra{j}) = h'(\ket{j}\bra{j}) = h(\ket{j}\bra{j}) \leq 1$.

We turn to the second map $\hs[\DM]$
in~(\ref{hsEfDMisoEqn}). Injectivity is obtained like for $\hs[\Ef]$,
using that each operator $\ket{x}\bra{x}$ is a projection and thus an
effect. For surjectivity assume a map of effect modules $g\colon
\Ef(H) \rightarrow [0,1]$, we extend it to a linear map $g'\colon
\Pos(H)\rightarrow \Rnn$ by:
$$\begin{array}{rcl}
g'(B)
& = &
n \cdot g(\frac{1}{n}B)
\qquad \mbox{where $n\in\NNO$ is such that $\frac{1}{n}B\in\Ef(H)$.}
\end{array}$$

\noindent Such an $n$ can be found in the following way. Take the
spectral decomposition $B = \sum_{j}\lambda_{j}\ket{j}\bra{j}$, where
$\lambda_{j}\geq 0$, because $B$ is positive, and the $\ket{j}$ form
an orthonormal basis. We can find an $n\in\NNO$ with $\lambda_{j} \leq
n$ for each $j$. Then $\frac{1}{n}B =
\sum_{j}\frac{1}{n}\lambda_{j}\ket{j}\bra{j}$ is an effect by
Lemma~\ref{EfAltLem}.  We also have to check that the definition of
$g'$ is independent of the choice of $n$: if also $\frac{1}{m}B \in
\Ef(H)$, assume, without loss of generality $m\leq n$; then we use
that $g$ is a map of $[0,1]$-actions:
$$\textstyle n\cdot g(\frac{1}{n}B)
=
n\cdot g(\frac{m}{n}\cdot \frac{1}{m}B)
=
n\cdot \frac{m}{n}\cdot g(\frac{1}{m}B)
=
m\cdot g(\frac{1}{m}B).$$

\noindent It is easy to see that the map $g'$ is linear.  Hence by
Proposition~\ref{PosDualProp} there is a (unique)
$B=\hs[\Pos]^{-1}(g')\in\Pos(H)$ with $g' = \tr(B-) \colon \Pos(H)
\rightarrow \Rnn$. Then for $A\in\Ef(H)$ we have $g(A) = g'(A) =
\tr(BA) \in [0,1]$. In particular $1 = g(I) = \tr(BI) = \tr(B)$, so
that $B\in\DM(H)$.

\auxproof{
Linearity of $g'$:
$$\begin{array}{rcl}
g'(0)
& = &
g(0) \\
& = &
0 \\
g'(B+C)
& = &
n\cdot g(\frac{1}{n}(B+C)) \qquad \mbox{where } \frac{1}{n}(B+C)\in\Ef(H) \\
& = &
n\cdot g(\frac{1}{n}B+\frac{1}{n}C) \\
& = &
n\cdot (g(\frac{1}{n}B) + g(\frac{1}{n}C)) \\
& = &
n\cdot g(\frac{1}{n}B) + n\cdot g(\frac{1}{n}C) \\
& = &
g'(B) + g'(C).
\end{array}$$

\noindent For preservation of scalar multiplication assume
$\frac{1}{n}B\in\Ef(H)$ and $r\in\R$. If $r\leq 1$ we are done,
so assume $r > 1$ with $\frac{r}{m} \leq 1$. Then $\frac{1}{nm}rB \in
\Ef(H)$, so:
$$\begin{array}{rcl}
g'(rB)
& = &
nm \cdot g(\frac{1}{nm}rB) \\
& = &
nm \frac{r}{m} g(\frac{1}{n}B) \\
& = &
rng(\frac{1}{n}B) \\
& = &
rg'(B).
\end{array}$$
}

One of the equations that $\hs[\Ef]$ and $\hs[\DM]$ should satisfy to
ensure that we have a map of adjunctions is the following;
the other one is similar and left to the reader.
$$\xymatrix@R-1pc{
\DM(H)\ar[rr]^-{\eta=\lam{B}{\lam{h}{h(B)}}}\ar@{=}[ddrr] & &
   \EMod\big(\Conv(\DM(H), [0,1]), [0,1]\big)\ar[d]^-{(-) \after \hs[\Ef]} \\
& & \EMod(\Ef(H), [0,1])\ar[d]^{\hs[\DM]^{-1}} \\
& & \DM(H)
}$$

\noindent This triangle commutes since for $B\in\DM(H)$,
$$\begin{array}[b]{rcl}
\big(\hs[\DM]^{-1} \after \big((-) \after \hs[\Ef]\big) \after \eta\big)(B)
& = &
\hs[\DM]^{-1} \after \eta(B) \after \hs[\Ef] \\
& = &
\hs[\DM]^{-1} \after \lam{A}{\eta(B)(\hs[\Ef](A))} \\
& = &
\hs[\DM]^{-1} \after \lam{A}{\hs[\Ef](A)(B)} \\
& = &
\hs[\DM]^{-1} \after \lam{A}{\tr(AB)} \\
& = &
\hs[\DM]^{-1} \after \lam{A}{\tr(BA)} \\
& = &
\hs[\DM]^{-1} \after \tr(B-) \\
& = &
B.
\end{array}\eqno{\QEDbox}$$
\end{proof}

\auxproof{
Like before we like to explicitly describe the inverses of the
Hilbert-Schmidt maps $\hs[\Ef]$ and $\hs[\DM]$ from~(\ref{hsEfDMisoEqn})
via matrix entries.

\begin{lemma}
\label{hsEfDMinvLem}
Assume our Hilbert space $H$ has an orthonormal basis $\ket{1}$,
\ldots, $\ket{n}$.  For a map of convex sets $h\colon
\DM(H)\rightarrow [0,1]$ and a map of effect modules $g\colon\Ef(H)
\rightarrow [0,1]$ we have matrix entries:
$$\begin{array}{rcl}
\hs[\Ef]^{-1}(h)_{jk}
& = &
\left\{\begin{array}{l}
h(\ket{j}\bra{j}) \qquad \mbox{if $k=j$, and otherwise:} \\[.5em]
\frac{1}{2}\Big[h\Big(\frac{(\ket{j}+\ket{k})(\bra{j}+\bra{k})}{2}\Big)
  - h\Big(\frac{(\ket{j}-\ket{k})(\bra{j}-\bra{k})}{2}\Big) \; + \\
\quad ih\Big(\frac{(i\ket{j}+\ket{k})(-i\bra{j}+\bra{k})}{2}\Big) -
   ih\Big(\frac{(i\ket{j}-\ket{k})(-i\bra{j}-\bra{k})}{2}\Big)\Big]
\end{array}\right.
\\[3em]
\hs[\DM]^{-1}(g)_{jk}
& = &
f\Big(\frac{(\ket{j}+\ket{k})(\bra{j}+\bra{k})}{4}\Big)
  - f\Big(\frac{(\ket{j}-\ket{k})(\bra{j}-\bra{k})}{4}\Big) \; + \\
& & \quad if\Big(\frac{(i\ket{j}+\ket{k})(-i\bra{j}+\bra{k})}{4}\Big) -
   if\Big(\frac{(i\ket{j}-\ket{k})(-i\bra{j}-\bra{k})}{4}\Big).
\end{array}$$
\end{lemma}

\begin{proof}
Omitting the zero case, the proof of Theorem~\ref{ConvEModAdjMapThm}
gives:
$$\begin{array}{rcl}
\hs[\Ef]^{-1}(h)_{jk}
& = &
\hs[\Pos]^{-1}\Big(\lamin{B}{\Pos(H)}{\tr(B)}
   {h\big(\frac{B}{\tr(B)}\big)}\Big)_{jk}
\end{array}$$

\noindent Lemma~\ref{hsPosinvLem} provides a formulation of the
inverse $\hs[\Pos]^{-1}$ on the right, via a sum of four function
applications involving four positive operators.  Each of these
operators has trace $\frac{1}{2}$ if $j\neq k$. For instance:
$$\begin{array}{rcl}
\tr\Big(\frac{(\ket{j}+\ket{k})(\bra{j}+\bra{k})}{4}\Big)
& = &
\frac{1}{4}\tr\big(\ket{j}\bra{j} + \ket{j}\bra{k} + \ket{k}\bra{j} +
   \ket{k}\bra{k}\big) \\
& = &
\frac{1}{4}\Big(\tr(\inprod{j}{j}) + \tr(\inprod{k}{j}) + 
  \tr(\inprod{j}{k}) + \tr(\inprod{k}{k})\Big) \\
& = &
\left\{\begin{array}{ll}
1 \quad & \mbox{if }j=k \\
\frac{1}{2} & \mbox{otherwise}
\end{array}\right.
\end{array}$$

\noindent The case $j=k$ needs to be handled separately.

First we compute the three remaining traces:
$$\begin{array}{rcl}
\tr\Big(\frac{(\ket{j}-\ket{k})(\bra{j}-\bra{k})}{4}\Big)
& = &
\frac{1}{4}\tr\big(\ket{j}\bra{j} - \ket{j}\bra{k} - \ket{k}\bra{j} +
   \ket{k}\bra{k}\big) \\
& = &
\frac{1}{4}\Big(\tr(\inprod{j}{j}) - \tr(\inprod{k}{j}) -
  \tr(\inprod{j}{k}) + \tr(\inprod{k}{k})\Big) \\
& = &
\left\{\begin{array}{ll}
0 \quad & \mbox{if }j=k \\
\frac{1}{2} & \mbox{otherwise}
\end{array}\right. \\
\tr\Big(\frac{(i\ket{j}+\ket{k})(-i\bra{j}+\bra{k})}{4}\Big)
& = &
\frac{1}{4}\tr\big(\ket{j}\bra{j} + i\ket{j}\bra{k} - i\ket{k}\bra{j} +
   \ket{k}\bra{k}\big) \\
& = &
\frac{1}{4}\Big(\tr(\inprod{j}{j}) + i\tr(\inprod{k}{j}) -
  i\tr(\inprod{j}{k}) + \tr(\inprod{k}{k})\Big) \\
& = &
\frac{1}{2} \\
\tr\Big(\frac{(i\ket{j}-\ket{k})(-i\bra{j}-\bra{k})}{4}\Big)
& = &
\frac{1}{4}\tr\big(\ket{j}\bra{j} - i\ket{j}\bra{k} + i\ket{k}\bra{j} +
   \ket{k}\bra{k}\big) \\
& = &
\frac{1}{4}\Big(\tr(\inprod{j}{j}) - i\tr(\inprod{k}{j}) +
  i\tr(\inprod{j}{k}) + \tr(\inprod{k}{k})\Big) \\
& = &
\frac{1}{2}.
\end{array}$$

\noindent Thus:
$$\begin{array}{rcl}
\hs[\Ef]^{-1}(h)_{jj}
& = &
h\Big(\frac{(\ket{j}+\ket{j})(\bra{j}+\bra{j})}{4}\Big)
  - 0 \; + \\
& & \quad \frac{1}{2}ih\Big(\frac{(i\ket{j}+\ket{j})
      (-i\bra{j}+\bra{j})}{2}\Big) -
   \frac{1}{2}ih\Big(\frac{(i\ket{j}-\ket{j})
      (-i\bra{j}-\bra{j})}{2}\Big) \\
& = &
h\Big(\frac{4\ket{j}\bra{j}}{4}\Big) +
   \frac{1}{2}ih\Big(\frac{(i+1)(1-i)\ket{j}\bra{j}}{2}\Big) -
   \frac{1}{2}ih\Big(\frac{(i-1)(-i-1)\ket{j}\bra{j}}{2}\Big) \\
& = &
h(\ket{j}\bra{j}) +
   \frac{1}{2}ih\Big(\frac{2\ket{j}\bra{j}}{2}\Big) -
   \frac{1}{2}ih\Big(\frac{2\ket{j}\bra{j}}{2}\Big) \\
& = &
h(\ket{j}\bra{j}) +
   \frac{1}{2}ih(\ket{j}\bra{j}) - \frac{1}{2}ih(\ket{j}\bra{j}) \\
& = &
h(\ket{j}\bra{j}).
\end{array}$$

For the second equation in the lemma, involving $\hs[\DM]^{-1}$,
we can apply Lemma~\ref{hsPosinvLem} directly, since the operators
used there are all effects. This can be seen as follows. Clearly
we have
$$\begin{array}{rcl}
\frac{(\ket{j}+\ket{k})(\bra{j}+\bra{k})}{4}
& \leq &
\frac{(\ket{j}+\ket{k})(\bra{j}+\bra{k})}{\sqrt{2}}.
\end{array}$$

\noindent Now we are done since the operator on the right is a
projection (which can be verified easily) and thus an effect.
The same approach works in the other three cases. \QED

Write $V = \frac{(\ket{j}+\ket{k})(\bra{j}+\bra{k})}{\sqrt{2}}$. Then:
$$\begin{array}{rcl}
VV
& = &
\frac{1}{2}(V\ket{j}+V\ket{k})(\bra{j}+\bra{k}) \\
& = &
\frac{1}{2}\Big((\ket{j}+\ket{k}) + (\ket{j}+\ket{k})\Big)
   (\bra{j}+\bra{k}) \\
& = &
\frac{1}{2}2(\ket{j}+\ket{k})(\bra{j}+\bra{k}) \\
& = &
V.
\end{array}$$

Similarly for $V =
\frac{(\ket{j}-\ket{k})(\bra{j}-\bra{k})}{\sqrt{2}}$. Then:
$$\begin{array}{rcl}
VV
& = &
\frac{1}{2}(V\ket{j}-V\ket{k})(\bra{j}-\bra{k}) \\
& = &
\frac{1}{2}\Big((\ket{j}-\ket{k}) + (\ket{j}-\ket{k})\Big)
   (\bra{j}-\bra{k}) \\
& = &
\frac{1}{2}2(\ket{j}-\ket{k})(\bra{j}-\bra{k}) \\
& = &
V.
\end{array}$$

And for $V =
\frac{(i\ket{j}+\ket{k})(-i\bra{j}+\bra{k})}{\sqrt{2}}$. Then:
$$\begin{array}{rcl}
VV
& = &
\frac{1}{2}(iV\ket{j}+V\ket{k})(-i\bra{j}+\bra{k}) \\
& = &
\frac{1}{2}\Big((i\ket{j}+\ket{k}) + (i\ket{j}+\ket{k})\Big)
   (-i\bra{j}+\bra{k}) \\
& = &
\frac{1}{2}2(i\ket{j}+\ket{k})(-i\bra{j}+\bra{k}) \\
& = &
V.
\end{array}$$ 

Finally for $V =
\frac{(i\ket{j}-\ket{k})(-i\bra{j}-\bra{k})}{\sqrt{2}}$. Then:
$$\begin{array}{rcl}
VV
& = &
\frac{1}{2}(iV\ket{j}-V\ket{k})(-i\bra{j}-\bra{k}) \\
& = &
\frac{1}{2}\Big((i\ket{j}-\ket{k}) + (i\ket{j}-\ket{k})\Big)
   (-i\bra{j}-\bra{k}) \\
& = &
\frac{1}{2}2(i\ket{j}-\ket{k})(-i\bra{j}-\bra{k}) \\
& = &
V.
\end{array}$$ 
\end{proof}
}

\begin{remark}
\label{WPRem}
In~\cite{dHondtP06a} a quantum weaket precondition calculus is
developed using effects on a finite-dimensional Hilbert space as
predicates and density matrices as states. The underlying duality can
be made explicit in the current setting. Programs act on states and
are thus modeled as ``state transformer'' maps $\DM(H) \rightarrow
\DM(K)$. Here we ignore complete positivity aspects and simply
consider these state transformers as affine maps, \textit{i.e.}~as
maps in the category \Conv. Corresponding to such programs there are
``predicate transformers'' $\Ef(K) \rightarrow \Ef(H)$ going in the
opposite direction. Naturally we consider them to be maps of effect
modules. The (dual) correspondence between state transformers and
predicate transformers can then be derived using the adjunction $\Conv
\leftrightarrows \EMod\op$ from Proposition~\ref{ConvEModAdjProp} and
the isomorphisms~(\ref{hsEfDMisoEqn}) from
Theorem~\ref{ConvEModAdjMapThm}:
$$\begin{prooftree}
\begin{prooftree}
\xymatrix{\DM(H)\ar[r] & \DM(K)}\rlap{\hspace*{6em}in \Conv}
\Justifies
\xymatrix{\DM(H)\ar[r] & \EMod(\Ef(K),[0,1])}
\using{\rlap{\em\small(\ref{hsEfDMisoEqn})}}
\end{prooftree}
\Justifies
\begin{prooftree}
\xymatrix{\Ef(K)\ar[r] & \Conv(\DM(H),[0,1])}
\Justifies
\xymatrix{\Ef(K)\ar[r] & \Ef(H)}\rlap{\hspace*{6.8em}in \EMod}
\using{\rlap{\em\small(\ref{hsEfDMisoEqn})}}
\end{prooftree}
\using{\rlap{\em\small(Prop.~\ref{ConvEModAdjProp})}}
\end{prooftree}\qquad\qquad$$

\noindent Such correspondences form the basis of Dijkstra's seminal
work on program correctness, see \textit{e.g.}~\cite{DijkstraS90}.
For a state transformer $f\colon \DM(H)\rightarrow \DM(K)$ the
corresponding predicate transformer $\weakprec(f, -) \colon \Ef(K)
\rightarrow \Ef(H)$ is the ``weakest precondition operation''. It is
given by:
$$\begin{array}{rcl}
\weakprec(f,A)
& = &
\hs[\Ef]^{-1}\Big(\lamin{B}{\DM(H)}{\hs[\DM]\big(f(B)\big)(A)}\Big) \\
& = &
\hs[\Ef]^{-1}\Big(\lamin{B}{\DM(H)}{\tr\big(f(B)A\big)}\Big),
\end{array}$$

\noindent where we use the isomorphisms $\hs[\DM]\colon \DM(K)
\conglongrightarrow \EMod(\Ef(K),[0,1])$ and $\hs[\Ef]^{-1} \colon
\Conv(\DM(H),[0,1]) \conglongrightarrow \Ef(H)$
from~(\ref{hsEfDMisoEqn}).  By elaborating the formulas for the matrix
entries $\weakprec(f,A)_{jk}$, the weakest precondition can be
computed explicitly (for instance, by a computer algebra tool).
\end{remark}

The dual adjunction $\Conv \leftrightarrows \EMod\op$ from
Proposition~\ref{ConvEModAdjProp} can be restricted to a (dual)
equivalence of categories, giving a probabilistic version of Gelfand
duality, see~\cite{JacobsM12b}. One obtains an equivalence $\CCHobs
\simeq \BEMod\op$ between `observable' convex compact Hausdorff spaces
and Banach effect modules. The latter are suitably complete with
respect to a definable norm. The map of adjunctions from
Theorem~\ref{ConvEModAdjMapThm} then restricts to a map of
equivalences:
$$\xymatrix{
\FdHilbUn\ar@/^1.2ex/[rr]^-{(-)^{\dag}}\ar[d]_{\DM} & \simeq & 
   \FdHilbUn\rlap{$\op$}\ar@/^1.2ex/[ll]^-{(-)^{\dag}}\ar[d]^{\Ef} \\
\CCHobs\ar@/^1.2ex/[rr]^-{\mathit{Hom}(-,[0,1])} 
   & \simeq & 
   \BEMod\rlap{$\op$}
   \ar@/^1.2ex/[ll]^-{\mathit{Hom}(-,[0,1])}
}$$

\noindent We refer to~\cite{JacobsM12b} for further details. This
equivalence leads to a reformulation of Gleason's
Theorem~\cite{Gleason57}. In original form it says that projections on
a Hilbert space $H$ (of dimension at least 3) correspond to measures:
$$\begin{array}{rcl}
\DM(H) & \cong & \EA\big(\Proj(H),\; [0,1]\big).
\end{array}$$

\noindent In~\cite{JacobsM12b} it is shown that Gleason's theorem
is equivalent to:
$$\begin{array}{rcl}
\Ef(H)
& \cong &
[0,1]\otimes \Proj(H).
\end{array}$$

\noindent This says that effects form the free effect module on
projections. We can now summarise how the whole edifice of operators
on a Hilbert space $H$ can be obtained from its projections
$\Proj(H)$, see Figure~\ref{FreeConstructsFig}.

\begin{figure}
\begin{center}
\begin{tabular}{c||c|c}
\textbf{Operators} & \textbf{Formula} & \textbf{Description} \\
\hline\hline 
effects & 
   $\Ef(H) \cong [0,1]\otimes \Proj(H)$ &
   Gleason's Theorem 
   \vrule height5mm depth3mm width0mm \\
\hline
density matrices &
   $\DM(H) \cong \EMod(\Ef(H),[0,1])$ &
   Theorem~\ref{ConvEModAdjMapThm}
   \vrule height5mm depth3mm width0mm \\
\hline
positive operators &
   $\Pos(H) \cong \mathcal{S}(\DM(H))$ &
   Theorem~\ref{ConvModCatAdjThm}
   \vrule height5mm depth3mm width0mm \\
\hline
self-adjoint operators &
   $\SA(H) \cong \mathcal{R}(\Pos(H))$ &
   Theorem~\ref{AlgCatAdjThm}
   \vrule height5mm depth3mm width0mm \\
\hline
bounded operators &
   $\BL(H) \cong \mathcal{C}(\SA(H))$ &
   Theorem~\ref{AlgCatAdjThm}
   \vrule height5mm depth3mm width0mm
\end{tabular}
\end{center}
\caption{Various operators on a Hilbert space $H$, constructed from
  the projections $\Proj(H)$}
\label{FreeConstructsFig}
\end{figure}

This concludes our overview of the categorical structure of the
various operators on a (finite dimensional) Hilbert space.

\auxproof{
\section{Starting from the effects}

In the previous section we saw that the spaces of operators $\BL(H)$,
$\SA(H)$ and $\Pos(H)$ can be obtained from density matrices $\DM(H)$
via free constructions. Now we show that they can also be obtained
from effects $\Ef(H)$. We do this by transforming
$[0,1]$-(effect)modules into $\Rnn$-modules, via an extension of the
partially defined sum $\ovee$ to a totally defined sum $+$. This
totalisation process will be sketched below; for more details and
applications see~\cite{JacobsM12a}.

We start with effect algebras. First we need to describe what our
totalised effect algebras will look like. For this purpose we define
the category $\BCM$ of barred commutative monoids as follows. The
objects are commutative monoids $M$ that are positive,
\textit{i.e.}~satisfy $a+b=0$ implies $a=b=0$, along with a special
element $u\in M$ called the unit of M, such that $a+b=a+c=u$ implies
$b=c$.  The morphisms between such BCMs are monoid homomorphisms that
preserve the unit (or `bar').

\begin{definition}
\label{TotalisationFunDef}
There is a totalisation functor $\partot\colon\EA\to\BCM$ given on
$E\in EA$ as a quotient of the free commutative monoid on $E$ in:
$$\begin{array}{rcl}
\partot(E)
& = &
(\, \Mlt_{\NNO}(E)/\!\sim, \; [1\cdot 1]_{\sim} \, ),
\end{array}$$

\noindent where $\sim$ is the smallest congruence such that $1\cdot
x+1\cdot y = 1\cdot(x\ojoin y)$ for all $x\bot y$, and $1\cdot 1$ is
the singleton multiset containing $1\in E$. We will usually not write
the square brackets denoting $\sim$ equivalence classes.
\end{definition}

As examples of this totalization process we have $\partot(\{0,1\})
\cong (\NNO,1)$ and $\partot([0,1]) \cong (\Rnn,1)$, and, as we shall
see later on, $\partot(\Ef(H))\cong \Pos(H)$.

Recall that all monoids come equiped with a preorder $\preceq$ given
by $a\preceq b$ if{f} there exists some $c$ such that $a+c=b$. We use
this to define the partialisation operation, opposite to totalisation.

\begin{definition}
\label{PartialisationFunDef}
Define a functor $\totpar\colon\BCM\to\EA$ by:
$$\begin{array}{rcl}
\totpar(M,u)
& = &
\setin{a}{M}{0\preceq a\preceq u}.
\end{array}$$

\noindent On this unit interval $a\ojoin b$ is defined as $a+b$
whenever this sum is below $u$.
\end{definition}

In~\cite{JacobsM12a} it is shown that the functor $\partot$ is a left
adjoint to $\totpar$ and that this adjunction is a coreflection,
\textit{i.e.}~its unit $\idmap\Rightarrow\totpar\partot$ is a natural
isomorphism. The adjunction involves straightforward bijective
correspondences:
$$\begin{prooftree}
{\xymatrix{(\partot(E),[11])\ar[r]^-{f} & (M,u)}}
\Justifies
{\xymatrix{E\ar[r]_-{g} & \totpar(M,u)}}
\end{prooftree}$$

\noindent given by $\overline{f}(x) = f([1x])$ and 
$\overline{g}([\sum_{i}n_{i}x_{i}]) = \sum_{i}n_{i}\cdot g(x_{i})$.

Here we are interested in effect modules rather than just effect
algebras. The functor $\partot$ is strongly symmatric monoidal,
essentially due to the construction of the tensor product. Hence it
forms a functor $\partot\colon\Mon(\EA) \rightarrow \Mon(\BCM)$,
mapping an effect monoid to a internal monoid in $\BCM$---which could
be called a ``barred semiring'', with $\Rnn$ as our main example. In
fact we get a coreflection between $\Mon(\EA)$ and $\Mon(\BCM)$ in
this way. We continue to write $\partot$ and $\totpar$ for this
restricted adjunction.

The same thing works for actions of an effect monoid $S$: if
$E\in\Act_{S}(\EA) = \EMod_{S}$ then we can equip the monoid
$\partot(E)$ with a $\partot(S)$ module structure in a canonical
way. The scalar multiplication is:
\begin{equation}
\label{FreeBModScalarEqn}
\begin{array}{rcl}
\big[\sum_{i}n_{i}s_{i}\big] \scalar \big[\sum_{j}m_{j}e_{j}\big]
& = &
\big[\sum_{i,j} (n_{i}\cdot m_{j})(s_{i}\scalar e_{j})\,\big].
\end{array}
\end{equation}

\noindent Again this leads to a coreflection between the categories of
actions $\Act_{S}(\EA)$ and $\Act_{\partot(S)}(\BCM)$, which we shall
also write as $\partot\dashv\totpar$. For a barred semi\-ring
$R\in\Mon(\BCM)$ we will write $\cat{BMod}_{R}$ instead of
$\Act_{R}(\BCM)$ and just $\cat{BMod}$ for $\cat{BMod}_{\Rnn}$ in the
special case when $R=\Rnn$.

Right now we are interested in the case $S=\unitR$, so we will
describe the category $\cat{BMod}$ of modules over $\partot(S) \cong
\Rnn$ in a bit more detail. The objects of $\cat{BMod}$ are BCMs
$(M,u)$ with a scalar multiplication $\scalar\colon \Rnn \times M\to
M$ such that the following conditions hold.
$$\begin{array}{rclcrcl}
\alpha\scalar (m+n)
& = &
(\alpha\scalar m) + (\alpha\scalar n)
& \quad &
(\alpha+\beta)\scalar m
& = &
(\alpha\scalar m)+(\beta\scalar m)\\
(\alpha\beta)\scalar m
& = &
\alpha\scalar(\beta\scalar m) 
& &
1\scalar m 
& = &
m.
\end{array}$$

\noindent Morphisms in $\cat{BMod}$ are BCM maps that preserve the
scalar multiplication.

If $E\in\cat{EMod}$ then the scalar product on $\partot(E)$ is defined
as follows. For $\alpha\in \Rnn$ let $a\in\NNO$ and $b\in\unitR$ be
its integral and decimal part respectively; then $\alpha \in\Rnn$
corresponds via $\Rnn \cong \partot(\unitR)$ to $[a1 +
  1b]\in\partot(\unitR)$, so that the scalar
multiplication~\eqref{FreeBModScalarEqn} yields:
$$\begin{array}{rcl}
\alpha\scalar(\sum_{i} n_{i}x_{i})
& = &
\sum_{i} (an_{i})x_{i} + \sum_{i} n_i(b\scalar x_{i}).
\end{array}$$

Note that for any $[0,1]$-module $M$ the unit $u$ of $\partot M$ is
strong, in the sense that for each $x\in\partot M$ there is an
$n\in\NNO$ such that $x\preceq nu$. Indeed, if $x =
\sum_{i}n_{i}y_{i}$, then $y_{i}\leq 1\in M$ so that $x \leq
\sum_{i}n_{i}1 = (\sum_{i}n_{i})1 = (\sum_{i}n_{i})u$, where $u =
    [11]\in\partot M$. Because of this property we can restrict the
    $\partot\dashv\totpar$ adjuction to the category $\Bmodu$ of
    $\Rnn$-modules with a strong unit.


\begin{lemma} 
\label{L:BmoduCancel}
If $M\in\Bmodu$ then the cancellation law holds in $M$.
\end{lemma}

\begin{proof}
Let $x,y,z\in M$ and suppose $x+y=x+z$. Since $u$ is a strong unit we
can find an $n$ such that $x+y\preceq nu$. Therefore:
$$\begin{array}{rcccl}
\frac{1}{n}\scalar x+ \frac{1}{n} \scalar y
& = &
\frac{1}{n}\scalar x+\frac{1}{n}\scalar z
& \preceq &
u,
\end{array}$$

\noindent and since the cancellation law holds below the unit in a BCM
we can conclude $\frac{1}{n}\scalar y =\frac{1}{n}\scalar z$. But we
have $y = \sum_{i=1}^n \frac{1}{n}\scalar y$ and therefore $y=z$. \QED
\end{proof}

An immediate consequence is that the preorder $\preceq$ is a partial
order so we shall write $\leq$ instead of $\preceq$. Further, if $x
\leq y$, there is a unique element $z$ with $x+z = y$. We shall write
$y-x$ for this $z$.

\auxproof{
Suppose $x \preceq y$ and $y\preceq x$, say via $x+u = y$ and
$y + v = x$. Then $x+(u+v) = y+v = x = x+0$, so $u+v=0$ by
cancellation, and thus $u=v=0$ by positivity. Hence $x=y$.
}

\begin{lemma} 
\label{L:[0,1]equiv1}
The coreflection $\partot\dashv\totpar$ between $\cat{EMod}$ and
$\Bmodu$ is an equivalence of categories.
\end{lemma}

\begin{proof}
We only need to show that the counit $\varepsilon \colon
\partot\totpar \Rightarrow\idmap$ of the adjunction is an isomorphism.
So let $M\in\Bmodu$, a typical element of $\partot\totpar M$ is an
equivalence class of formal sums like $\sum n_ix_i$ where $n_i\in
\mathbb{N}$ and $M\ni x_i\leq u$. The counit $\eps$ sends the class
represented by this formal sum to its interpretation as an actual sum
in $M$.

To show that $\eps$ is surjective suppose $m\in M$. We can find a
natural number $n$ such that $m\leq nu$ so that $\frac{1}{n}\scalar m
\leq u$. This gives us in $M$:
$$\begin{array}{rcccl}
m
& = &
n\cdot (\frac{1}{n}\scalar m)
& = &
\eps(n(\frac{1}{n}\scalar m)).
\end{array}$$

\noindent To prove injectivity suppose that $\eps(\sum_{i} n_ix_i) =
\eps(\sum_{j} k_jy_j)$. Define $N=\sum n_i+\sum k_j\in\NNO$ then we
also have in $M$:
$$\begin{array}{rcl}
\sum n_i\cdot (\frac{1}{N}\scalar x_i)
& = &
\eps (\sum n_i(\frac{1}{N}\scalar x_i)) \\
& = &
\eps(\frac{1}{N}\scalar(\sum n_ix_i)) \\
& = &
\frac{1}{N}\scalar\eps(\sum n_ix_i) \\
& = &
\frac{1}{N}\scalar\eps(\sum k_jy_j) \\
& = &
\sum k_j(\frac{1}{N}\scalar y_j)
\end{array}$$

\noindent Because $N$ is big enough the terms $\ovee_{i}\, n_i\cdot
(\frac{1}{N}\scalar x_i)$ and $\ovee_{j}\, k_j\cdot(\frac{1}{N}\scalar
y_j)$ are both defined in the effect module $\totpar M$ and by the
previous calculation they are equal. This means that $\sum_{i}
n_i(\frac{1}{N}\scalar x_i)$ and $\sum_{j} k_j(\frac{1}{N}\scalar y_k)$
represent equal elements of $\partot\totpar M$ and therefore the
equation:
$$\begin{array}{rccccccl}
\sum_{i} n_ix_i
& = &
N\scalar (\sum_{i} n_i(\frac{1}{N}\scalar x_i)) 
& = &
N\scalar (\sum_{j} k_j(\frac{1}{N}\scalar y_j))
& = &
\sum_{j} k_jy_j
\end{array}$$

\noindent holds in $\partot\totpar M$. \QED
\end{proof}

We will take this equivalence one step further. Define a category
\poVectu as follows; the objects are partially ordered vector spaces
over $\reals$ with a strong order unit $u$, \textit{i.e.}~a positive
element $u\geq 0$ such that for any $x\in V$ there is a natural number
$n$ with $x\leq nu$. The morphisms in \poVectu are linear maps that
preserve the order and the unit.


\begin{theorem}
\label{T:[0,1]equiv2}
The category $\cat{EMod}$ of effect modules is equivalent to
the category $\poVectu$ of ordered vector spaces with strong unit.
\end{theorem}

\begin{proof}
We will prove that $\Bmodu$ is equivalent to $\poVectu$ the result
then follows from lemma~\ref{L:[0,1]equiv1}.

There is a functor $F\colon\poVectu\to\Bmodu$ by taking the positive
cone $F(V) = \setin{x}{V}{x \geq 0}$ of a partially ordered vector
space $V$. This is clearly functorial. In the reverse direction there
is $G\colon\Bmodu\to\poVectu$, which is essentially just the usual
construction of turning a cancellative monoid into a group. Thus,
$G(M) = (M\times M)/\!\!\sim$, where $(x_{1},x_{2}) \sim
(y_{1},y_{2})$ iff $x_{1}+y_{2} = y_{1}+x_{2}$. We can use a slightly
simplified version of the functor $\mathcal{R}\colon \Mod[\Rnn]
\rightarrow \Vect[\R]$ in Theorem~\ref{AlgCatAdjThm}, since
cancellation holds in the current context
(Lemma~\ref{L:BmoduCancel}). We order $G(M)$ via $[x_{1},x_{2}] \leq
[y_{1},y_{2}]$ iff $x_{1}+y_{2} \leq y_{1}+x_{2}$. As strong unit we
take $\eta(u) = [u,0]\in G(M)$.

It is easy to see that $F,G$ give an equivalence of categories. \QED

\auxproof{
Functoriality of $F\colon\poVectu\to\Bmodu$ is trivial: each $f\colon
V \rightarrow W$ restricts to $\setin{x}{V}{x \geq 0} \rightarrow
\setin{y}{W}{y \geq 0}$, since $x \geq 0$ implies $f(x) \geq f(0) =
0$. Clearly $F(V)$ is an $\Rnn$-module, since $r\geq 0$ and $x\geq 0$
implies $rx\geq 0$. We have a strong order unit $u\in F(V)$ by
definition.

Similarly, $G$ is functorial: $g\colon M\rightarrow K$ gives
$G(g)\colon G(M) \rightarrow G(K)$ by $[x_{1},x_{2}] \longmapsto
[f(x_{1}), f(x_{2})]$. This is clearly well-defined by linearity
of $f$. This $G(f)$ preserves the unit:
$$G(f)(\eta(u))
=
G(f)([u,0])
=
[f(u),f(0)]
=
[u,0].$$

\noindent It also preserves the order:
$$\begin{array}{rcl}
[x_{1},x_{2}] \leq [y_{1},y_{2}]
& \Longleftrightarrow &
x_{1}+y_{2} \leq y_{1}+x_{2} \\
& \Longrightarrow &
f(x_{1})+f(y_{2}) = f(x_{1}+y_{2}) \leq f(y_{1}+x_{2}) = f(y_{1})+f(x_{2}) \\
& \Longleftrightarrow &
[f(x_{1}),f(x_{2})] \leq [f(y_{1}),f(y_{2})] \\
& \Longleftrightarrow &
G(f)([x_{1},x_{2}]) \leq G(f)([y_{1},y_{2}]).
\end{array}$$

There is a map
$$\xymatrix{
V \ar[r]^-{\varphi} & GF(V)
\quad\mbox{by}\quad
x\ar@{|->}[r] & {\left\{\begin{array}{ll} [x,0] & \mbox{ if }x \geq 0 \\
   {[0,-x]} & \mbox{ if } x < 0 \end{array}\right.}
}$$

\noindent Obviously, $\varphi$ preserves $0$ and $u$. We check that
it also preserves addition and the order.

For addition, consider two elements $x,y\in V$. We distinguish several
cases.
\begin{itemize}
\item $x\geq 0, y\geq 0$, and thus $x+y\geq 0$. This is easy:
$$\varphi(x+y)
= 
[x+y,0]
=
[x,0]+[y,0]
=
\varphi(x)+\varphi(y).$$

\item $x<0, y<0$, and so $x+y<0$. Again easy:
$$\varphi(x+y)
=
[0,-(x+y)]
=
[0,-x-y]
=
[0,-x] + [0,-y]
=
\varphi(x)+\varphi(y).$$

\item $x\geq 0, y<0$ with $x+y\geq 0$. Then:
$$\begin{array}{rcl}
\varphi(x+y)
& = &
[x+y,0] \\
& = &
[x,-y] \qquad\mbox{since } (x+y)+ (-y) = x = x + 0 \\
& = &
[x,0] + [0,-y] \\
& = &
\varphi(x) + \varphi(y).
\end{array}$$

\item $x\geq 0, y<0$ with $x+y < 0$. Then:
$$\begin{array}{rcl}
\varphi(x+y)
& = &
[0, -(x+y)] \\
& = &
[x,-y] \qquad\mbox{since } 0 + (-y) = -y = -(x+y) + x \\
& = &
[x,0] + [0,-y] \\
& = &
\varphi(x) + \varphi(y).
\end{array}$$
\end{itemize}

\noindent Next assume $x \leq y$. In order to show $\varphi(x) \leq
\varphi(y)$ we need to distinguish several cases:
\begin{itemize}
\item $x \geq 0$. Then $y\geq 0$ and:
$$\varphi(x)
=
[x,0]
\leq 
[y,0]
=
\varphi(y).$$

\item $y < 0$. Then $x < 0$ and $-y = x + (-x-y) \leq y + (-x-y) =
  -x$, so:
$$\varphi(x)
=
[0,-x]
\leq
[0,-y]
=
\varphi(y).$$

\item $x<0, y\geq 0$. Then $y-x \geq 0$ and thus:
$$\varphi(x)
=
[0,-x] 
\leq
[y,0] 
=
\varphi(y).$$

\item $x\geq 0, y<0$ is impossible.
\end{itemize}

\noindent The inverse $\varphi^{-1}\colon GF(V) \rightarrow V$ is
given by:
$$\begin{array}{rcl}
\varphi^{-1}([x_{1},x_{2}])
& = &
x_{1}-x_{2}.
\end{array}$$

\noindent Then $V \cong GF(V)$ since:
$$\begin{array}{rcl}
\varphi^{-1}(\varphi(x))
& = &
{\left\{\begin{array}{ll} \varphi^{-1}([x,0]) & \mbox{ if }x \geq 0 \\
   \varphi^{-1}([0,-x]) & \mbox{ if } x < 0 \end{array}\right.} \\
& = &
{\left\{\begin{array}{ll} x-0 & \mbox{ if }x \geq 0 \\
   0 - (-x) & \mbox{ if } x < 0 \end{array}\right.} \\
& = &
x \\
\varphi(\varphi^{-1}([x_{1},x_{2}])
& = &
\varphi(x_{1}-x_{2}) \\
& = &
{\left\{\begin{array}{ll} [x_{1}-x_{2},0] & \mbox{ if }x_{1}-x_{2}\geq 0 \\
   {[0,-(x_{1}-x_{2})]} & \mbox{ if } x_{1}-x_{2} < 0 \end{array}\right.} \\
& = &
[x_{1},x_{2}].
\end{array}$$

Similarly, there is a map:
$$\xymatrix{
FG(M)\ar[r]^-{\psi} & M
\quad\mbox{by}\quad
\big([x_{1},x_{2}]\geq 0\big)\ar@{|->}[r] & 
   x_{1}- x_{2}.
}$$

\noindent The latter is well-defined, since $[x_{1},x_{2}] \geq 0 =
          [0,0]$ means $x_{1} \geq x_{2}$. Thus $x_{1}- x_{2}$ is
          defined in $M$.

Clearly, $\psi$ preserves zero $[0,0]$ and the unit $[u,0]$. It also
preserves addition since:
$$\begin{array}{rcl}
\psi\big([x_{1},x_{2}] + [y_{1},y_{2}]\big)
& = &
\psi\big([x_{1}+y_{1}, x_{2}+y_{2}]\big) \\
& = &
(x_{1}+y_{1}) - (x_{2}-y_{2}) \\
& = &
(x_{1}-x_{2}) + (y_{1}-y_{2}) \\
& = &
\psi([x_{1},x_{2}]) + \psi([x_{2},y_{2}]).
\end{array}$$

\noindent The inverse of $\psi$ is easy, since each $x\in M$ is
positive ($0 \leq x$, since $0+x=x$), so we can define $\psi^{-1}(x)
= [x,0]$. Then $FG(M) \cong M$ since:
$$\begin{array}{rcl}
\psi^{-1}(\psi([x_{1},x_{2}]))
& = &
\psi^{-1}(x_{1}-x_{2}) \\
& = &
[x_{1}-x_{2}, 0] \\
& = &
[x_{1},x_{2}] \\
\psi(\psi^{-1}(x))
& = &
\psi([x,0]) \\
& = &
x-0 \\
& = &
x.
\end{array}$$
}
\end{proof}

What does this mean for the effects of a Hilbert space? The positive
operators $\Pos(H)$ can be viewed as an object of $\Bmodu$, with the
identity $I$ as strong unit, so that by definition $\totpar(\Pos(H)) =
\Ef(H)$. Hence $\partot(\Ef(H)) \cong \Pos(H)$ by
lemma~\ref{L:[0,1]equiv1}. From here we can apply the free
constructions again and construct the spaces of operators $\SA(H)$ and
$\BL(H)$.

\marginpar{Is there a relation between $\Conv \rightarrow \Mod[\R]$
  and $\EMod \rightarrow \cat{BModu}$, via the adjunction $\Conv
  \rightleftarrows \Mod[\R]\op$?}

A construction similar to ours for embedding effect modules in
partially ordered vector spaces can be found
in~\cite{PulmannovaG98}. This construction in fact gives the same
result as the one we presented above. However~\cite{PulmannovaG98}
seems to have missed the equivalence of categories that it entails;
instead, only half of it is presented, as a representation theorem.
}

\subsection*{Acknowledgements}

The first steps of the research underlying this work was carried out
during a sabbatical visit of the first author (BJ) to the Quantum
Group at Oxford University in April and May 2010. Special thanks, for
discussion and/or feedback, go to Bob Coecke, Rick Dejonghe, Chris
Heunen, Klaas Landsman, Bas Spitters, and Dusko Pavlovi{\'c}.


\bibliographystyle{plain} 

\end{document}